\documentclass{article}
\usepackage[utf8]{inputenc}
\usepackage{amsmath,amssymb,mathpazo,times,epsfig,bigints,relsize}
\usepackage[numbers,sort&compress]{natbib}
\usepackage{hyperref}
\usepackage{color}

\usepackage{graphicx}
\usepackage{enumerate}
\usepackage{bm}
\usepackage{authblk}
\usepackage{subcaption}
\usepackage{titlesec}
\usepackage{authblk}
\usepackage{blindtext}
\usepackage[T1]{fontenc}
\usepackage[utf8]{inputenc}
\usepackage{amsmath}
\usepackage{ntheorem}
\DeclareMathOperator{\sech}{sech}
\newtheorem{prop}{Proposition}

\advance\marginparwidth20mm

\makeatletter
\renewcommand\section{\@startsection {section}{1}{\z@}%
  {-4.6ex \@plus -0.4ex \@minus -0.2ex}{2.4ex \@plus0.2ex}%
  {\normalfont\bfseries}}
\renewcommand\subsection{\@startsection{subsection}{2}{\z@}%
  {-3.2ex\@plus -0.2ex \@minus -0.2ex}{1.8ex \@plus0.2ex}%
  {\normalsize\it}}
\renewcommand\paragraph{\@startsection{paragraph}{4}{\z@}%
  {2.0ex \@plus 0.2ex \@minus 0.2ex}{-1em}%
  {\normalfont\normalsize\bfseries}}
\def\@maketitle{\null\vskip2.4em%
  \begin{center} \let\footnote\thanks \vskip3.4em%
  {\Large\@title \par}\vskip2.4em%
  {\lineskip 0.5em\begin{tabular}[t]{c}\@author\end{tabular}\par}\vskip1em%
  {\@date}%
  \end{center}%
  \par\vskip3.4em}
\long\def\@makecaption#1#2{\vskip\abovecaptionskip
  \sbox\@tempboxa{\small #1: #2\par}%
  \ifdim \wd\@tempboxa >\hsize \small #1: #2\par
  \else \global \@minipagefalse \hb@xt@\hsize{\hfil\box\@tempboxa\hfil}%
  \fi
  \vskip\belowcaptionskip}
\advance\textwidth 4mm
\advance\textheight 24mm
\advance\voffset -18mm
\advance\hoffset -10mm
\allowdisplaybreaks

\def\Im{\mathop{\rm Im}\nolimits}

\def\Real{\mathbb{R}}

\def\Complex{\mathbb{C}}
\let\<=\langle
\let\>=\rangle
\let\==\bar
\let\^=\hat
\let\~=\tilde
\let\.=\dot

\def\sech{\mathop{\rm sech}\nolimits}

\def\diag{\mathop{\rm diag}\nolimits}
\def\qed{\nopagebreak~{\small$\square$}\medbreak}

\let\le=\leqslant

\let\@=\mathbf

\let\eref=\eqref
\let\eqref=\undefined

\let\trueint=\int
\let\trueiint=\iint
\let\trueiiint=\iiint
\let\trueoint=\oint
\let\truesum=\sum
\let\trueprod=\prod
\def\@int#1{\mathchoice
{\@@int\displaystyle\textstyle{#1}}%
  {\@@int\textstyle\scriptstyle{#1}}%
  {\@@int\scriptstyle\scriptscriptstyle{#1}}%
  {\@@int\scriptscriptstyle\scriptscriptstyle{#1}}\!\int}
  \def\@@int#1#2#3{{\setbox0=\hbox{$#1{#2#3}{\textstyle\trueint}$}
  \kern0.4\wd0\vcenter{\hbox{$#2#3$}}\kern-0.666\wd0}}
\def\ddashint{\@int=}
\def\dashint{\@int-}
\def\int{\mathop{\textstyle\trueint}\limits}
\def\iint{\mathop{\textstyle\trueiint}\limits}
\def\iiint{\mathop{\textstyle\trueiiint}\limits}

\def\oint{\mathop{\textstyle\trueoint}\limits}
\def\oiint{\mathop{\circ\kern-1em\textstyle\trueint\kern-0.6em\trueint}\limits}
\def\sum{\mathop{\textstyle\truesum}\limits}
\def\prod{\mathop{\textstyle\trueprod}\limits}

\def\intinfty{\kern-0.2em\mathop{\textstyle\trueint}\limits_{\!-\infty\,}^{\,\,\infty\!}\kern-0.25em}
\def\txtintinfty{\kern-0.2em\mathop{\textstyle\trueint}\nolimits_{\!-\infty\,}^{\,\,\infty\!}\kern-0.25em}
\def\iintinfty{\mathop{\textstyle\trueiint}\limits_{\!\!-\infty\,\,}^{\,\,\infty\!}\kern-0.10em}
\def\iiiintinfty{\mathop{\textstyle\trueiiint}\limits_{\!\!-\infty\,\,}^{\,\,\infty\!}\kern-0.20em}
\makeatother

\newtheorem{theorem}{Theorem}[section]

\newtheorem{lemma}[theorem]{Lemma}
\newtheorem{proposition}[theorem]{Proposition}
\newtheorem{corollary}[theorem]{Corollary}
\newtheorem{remark}[theorem]{Remark}

\def\[{\begin{equation}}
\def\]{\end{equation}}
\def\be{\begin{equation}}
\def\ee{\end{equation}}
\def\bse{\begin{subequations}}
\def\ese{\end{subequations}}
\def\bea{\begin{eqnarray}}
\def\eea{\end{eqnarray}}

\newdimen\figwdlt
\newdimen\figwdrt
\title{Soliton interactions and Yang-Baxter maps for the complex coupled short-pulse equation}
\author[1]{Vincent Caudrelier}
\author[2]{Aikaterini Gkogkou}
\author[2]{Barbara Prinari\footnote{Corresponding author: bprinari@buffalo.edu}}
\affil[1]{School of Mathematics, University of Leeds, Leeds, LS2 9JT, UK}
\affil[2]{Department of Mathematics, University at Buffalo, Buffalo, NY 14260, USA}
\date{ }

\begin{document}

\maketitle

\begin{abstract}

The complex coupled short pulse equation (ccSPE) describes the propagation of ultra-short optical pulses in nonlinear birefringent fibers. The system admits a variety of vector soliton solutions: fundamental solitons, fundamental breathers, composite breathers (generic or non-generic), as well as so-called self-symmetric composite solitons.
In this work, we use the dressing method and the Darboux matrices corresponding to the
various types of solitons to investigate soliton interactions in the focusing ccSPE. The study combines
refactorization problems on generators of certain rational loop groups, and
long-time asymptotics of these generators, as well as the main refactorization theorem for the dressing factors which leads to the
Yang-Baxter property for the refactorization map and the vector soliton interactions. Among the results obtained in this paper, we derive explicit formulas for the polarization shift of fundamental solitons which are the analog of the well-known formulas for the interaction of vector solitons in the Manakov system. Our study also reveals that upon interacting with a fundamental breather, a fundamental soliton becomes a fundamental breather and, conversely, that the interaction of two fundamental breathers generically yields two fundamental breathers with a polarization shifts, but may also result into a fundamental soliton and a fundamental breather. Explicit formulas for the coefficients that characterize the fundamental breathers, as well as for their polarization vectors are obtained. The interactions of other types of solitons are also derived and discussed in detail and illustrated with plots. New Yang-Baxter maps are obtained  in the process.
\end{abstract}

\section{Introduction}

Mathematical models of nonlinear wave propagation can often be reduced to a class of nonlinear partial differential equations known
as integrable systems. One of the most widely studied integrable systems is the nonlinear Schr\"odinger (NLS)
equation, which in the last 50 years has been shown to be a universal model for weakly dispersive nonlinear wave trains, with physical applications
ranging from deep water waves, plasma physics and nonlinear optics, to magneto-static spin waves, low temperature physics and Bose-Einstein condensation.
On the other hand, the propagation of ultra-short optical pulses (width$\sim 10^{-15} s$ and much smaller than the carrier frequency) in nonlinear media is better described by the so-called ``complex short-pulse'' equation (cSPE):
\begin{equation}\label{e:1.1}
u_{xt}=u+ \frac{\sigma}{2}(|u|^2u_x)_x\,, \qquad \sigma=\pm 1
\end{equation}
where $u=u(x,t)$ is a complex function representing the electric field associated to the propagating optical pulse. The cSPE was introduced relatively recently in \cite{Baofeng}, and like NLS, the sign of $\sigma$ distinguishes the two dispersion regimes ($\sigma=1$ corresponding to the anomalous dispersion regime, or focusing cSPE, and $\sigma=-1$ to normal dispersion, or defocusing cSPE). If one restricts $u(x,t)$ to be a real function (representing, in this case, the magnitude of the electric field), the above equation reduces to the (real) short-pulse equation (SPE), which was originally introduced in the context of differential geometry \cite{surfaces}, and was later derived as a model for the propagation of ultra-short pulses in nonlinear silica optics \cite{Schaefer}. Equations of short-pulse type:
$$
Q_{x\tau}=4iQ-2i(RQQ_x)_x, \quad R_{x\tau}=4iR-2i(QRR_x)_x
$$
were obtained in the earlier works \cite{ZQ0,ZQ01} through the negative Wadati-Konno-Ichikawa flow \cite{WKI,ZQ1,ZQ2,ZQ3}. These equations reduce to \eref{e:1.1} for $R=-\sigma Q^*$ but with a complex time $t=4i \tau$.

A key feature of the SPE and the cSPE is that, in addition to standard smooth solitons, both admit loop soliton solutions, which are not single-valued, and ``cuspons'', and also solutions that oscillate between single- and multi-valued states. For applications to birefringent fibers, two orthogonally polarized modes have to be considered, and in analogy to the Manakov system \cite{Manakov}, which is the extension of the NLS equation to 2-components, several generalizations of the SPE were proposed in the literature for the propagation of polarized ultra-short pulse in anisotropic media. While there is a sizeable amount of literature on the SPE, on its two- and multi-component generalizations and discretization (see \cite{ZQ0,ZQ01,ZQ1,ZQ2,ZQ3,Sak1,Sak2,Mat1,
Mat2,KKV,Kuetche,Parkes2008,Gambino,Pelin2,
Coclite,Pelin,Brunelli1,Brunelli2,BdM,Pietr,
Mat3,Feng2012,Kofane14,Qiao17,Hassan,
FMOmultiSP}), the study of the cSPE and of its vector version, the complex coupled SPE (ccSPE) also introduced in \cite{Baofeng}, namely:
\begin{equation}
\label{e:ccSP}
\mathbf{u}_{xt}=\mathbf{u}+\frac{\sigma }{2}(||\mathbf{u}||^2\mathbf{u}_x)_x\,, \qquad \mathbf{u}=(u_1,u_2)^T\,, \qquad \sigma=\pm 1\,,
\end{equation}
where $\mathbf{u}(x,t)$ is a two-component complex vector function and $\sigma$ again distinguishes between the focusing and defocusing equations, is obviously much more recent and less extensive.
Like NLS and the Manakov system, the defocusing cSPE and ccSPE only admit dark solitons, i.e., solitons on a non-zero background. Soliton solutions for the focusing cSPE equation have been constructed in \cite{Baofeng,FengShen_ComplexSPE,Xujian,LFZPhysD,Mee,Qilao}, and dark soliton solutions of the defocusing cSPE have been obtained in \cite{FLZ2021,FenglingzhuPRE}. The inverse scattering transform (IST) to solve the initial-value problem for the focusing cSPE equation was developed in \cite{PTF}, and the long-time asymptotic behavior was analyzed in \cite{FanJDE}. As to the focusing ccSPE, several types of solutions were presented in \cite{Zhu18,Guo18,Optik,FL22}, and the IST was developed in \cite{ABBA}.

The main goal of this work is to study interactions of vector solitons of the focusing ccSPE. It is known \cite{APT2,CZ} that the interactions between solitons in the Manakov model, or more generally vector NLS, give rise to maps on their polarization vectors which provide solutions of the set-theoretical Yang-Baxter equation \cite{Drin}. In this context and also in the context of discrete integrable systems, such maps are known as Yang-Baxter maps \cite{Ves}. They arise in a much larger variety of contexts, and we refer the interested reader to \cite{DS} for an overview of various key areas where the set-theoretical Yang-Baxter equation (and its companion, the set-theoretical reflection equation \cite{CZ,CCZ}) can arise. From the point of view of soliton dynamics, such maps ensure that multicomponent soliton interactions are elastic and that the scattering of a multisoliton solution factorizes {\it consistently} into a succession of two-soliton interactions. This is a well-known key feature of scalar solitons, but it is more intricate to derive in the multicomponent case. Nevertheless, the interplay between multicomponent integrable equations and Yang-Baxter maps is well documented and finds its roots in the refactorization properties appearing in the underlying dressing method \cite{ZS}. This was used extensively e.g. in \cite{CZ}.
In this paper, we investigate Yang-Baxter maps for the focusing ccSPE, and use them to unravel the nature of the corresponding soliton interactions.
An essential new feature compared to the vector NLS case is the variety of possible one-soliton solution that the model admits: fundamental solitons, fundamental breathers, composite breathers (generic or non-generic), as well as so-called self-symmetric solitons. In a first instance, by considering the interaction of two fundamental solitons,  we derive a formula analogous to Manakov's result for the polarization shift of interacting vector NLS solitons. This gives a first example of the Yang-Baxter maps involved in ccSPE. To get the full picture, we take advantage of the ideas illustrated above, and classify the possible dressing factors creating those various types of solitons. We then derive the ``master'' Yang-Baxter map arising from the refactorization of the most general elementary dressing factors. Combining this with a long-time asymptotic analysis of all the possible two-soliton solutions yields the various maps on the polarizations of the solitons. All of them enjoy the Yang-Baxter property, being derived from the ``master'' Yang-Baxter map, but take on different explicit forms.

For the rest of this work, we restrict our attention to the focusing ccSPE (so we assume $\sigma =1$, and simply refer to Eq.~\eref{e:ccSP} with $\sigma=1$ as the ccSPE), and consider solutions that are rapidly decaying as $|x|\rightarrow \infty$.
The structure of the paper is as follows. In Sec.~2 we give a brief overview of the IST for the ccSPE as developed in \cite{ABBA}, and of its one-soliton solutions, which include fundamental solitons, fundamental breathers, and composite breathers, depending on the rank and structure of the norming constant associated to the soliton. We also discuss in detail the case of self-symmetric discrete eigenvalues, and derive the explicit expression of a self-symmetric soliton. In Sec.~3 we discuss the reductions of the ccSPE to the case of real solutions. In Sec.~4 we provide the explicit expressions of the (matrix) transmission coefficients corresponding to a 1-fundamental soliton solution, a 1-fundamental breather solution, and a 1-self-symmetric soliton solution. In Sec.~5 we use Manakov's method \cite{Manakov} to investigate the pairwise interactions of two fundamental solitons, and also the interaction of self-symmetric solitons. Sec.~6 reviews the main idea of the dressing method and the notion of dressing factors (or Darboux-B\"acklund matrices), as well as the main refactorization theorem for such dressing factors which leads to the Yang-Baxter property for the refactorization map. It also contains the classification of the elementary dressing factors necessary to build the three types of solitons in the ccSPE, as well as their various degenerations. Finally, the long-time analysis of various two-soliton solutions leads to the derivation of the various Yang-Baxter maps on the polarization vectors of the solitons. Our study reveals that upon interacting with a fundamental breather, a fundamental soliton becomes a fundamental breather and, conversely, that the interaction of two fundamental breathers generically yields two fundamental breathers with a polarization shifts, but may also result into a fundamental soliton and a fundamental breather. Explicit formulas for the coefficients that characterize the fundamental breathers, as well as for their polarization vectors are obtained. The interactions of other types of solitons are also discussed in detail and illustrated with plots.
Finally, Appendices A-E provide more technical details regarding the derivation of the explicit expression of the analytic scattering coefficients in various cases, hodograph transformation and exact two-soliton solutions.

\section{Overview of the IST and one-soliton solutions}

Below, we give a succinct overview of the IST for the ccSPE as developed in \cite{ABBA}. The ccSPE \eref{e:ccSP} with $\sigma=1$ possess the following Lax pair:
\begin{subequations}\label{e:1.3b}
\begin{gather}
\Phi_x=X \Phi=\begin{pmatrix} -ikI_2 & k U_x \\ -k V_x & ikI_2\end{pmatrix}\Phi \,, \\
\Phi_t=T \Phi=\begin{pmatrix} \frac{i}{4k}I_2-\frac{i}{2}kUV  & -\frac{i}{2}U+\frac{1}{2}k UVU_x \\  -\frac{i}{2}V-\frac{1}{2}k VUV_x &
-\frac{i}{4k}I_2+\frac{i}{2}k VU \end{pmatrix}\Phi \,,
\end{gather}
\end{subequations}
where the matrices $U$ and $V$ are given by
\begin{gather}
\label{e:U}
U=\begin{pmatrix} -iu_1 & -iu_2 \\ iu_2^* &-iu_1^* \end{pmatrix}\,, \qquad V=U^\dagger\,,
\end{gather}
and $I_n$ denotes the $n\times n$ identity matrix. In \cite{ABBA}, the gauge transformation
\begin{gather}
\label{e:gauge1}
\hat{\Phi}(x,t,k) = P^\dagger (x,t) \Phi(x,t,k),
\end{gather}
with $P$ chosen so that it diagonalizes the matrix $iX/k$, namely
\begin{subequations}
\begin{gather}
P = p \begin{pmatrix}
I_2 & -\alpha\\
\alpha^\dagger & I_2
\end{pmatrix}, \quad p^2 = \frac{1+q}{2q}, \quad q=\sqrt{1 + ||u_x||^2}, \quad \alpha = \frac{i U_x}{1+q}\,,
\end{gather}
\end{subequations}
was used to control the behavior of the eigenfunctions at $k=0$ and $k=\infty$.
Indeed, with such a choice for $P$, the gauge transformation \eref{e:gauge1} reduces the Lax pair \eref{e:1.3b} to
\begin{gather}
\label{e:4new}
\hat{\Phi}_x + Q_x \hat{\Phi} = \hat{X} \hat{\Phi}, \quad \hat{\Phi}_t + Q_t \hat{\Phi} = \hat{T} \hat{\Phi},
\end{gather}
where
\begin{subequations}
\begin{gather}
Q_x=ikq\Sigma_3 \,, \qquad Q_t=\left(\frac{1}{4ik}+\frac{i}{2}kq ||\mathbf{u}||^2\right)\Sigma_3\,, \qquad \Sigma_3=\diag(I_2, -I_2)\,,
\label{e:QxQt}\\
\hat{X} =\begin{pmatrix}
\frac{1}{2(1+q)} q_xI_2-\frac{1}{2q(1+q)}U_xV_{xx} & \frac{i}{2q}U_{xx}-\frac{i q_x}{2q(1+q)}U_x \\
\frac{i}{2q}V_{xx}-\frac{i q_x}{2q(1+q)}V_x  & -\frac{1}{2(1+q)} q_x I_2+\frac{1}{2q(1+q)}V_{xx} U_x
\end{pmatrix}\,, \\
\hat{T}=\frac{i}{4kq}\begin{pmatrix} (1-q)I_2 & -iU_x \\ iV_x & -(1-q)I_2\end{pmatrix}+ \\
+\frac{p^2}{2} \begin{pmatrix} (\alpha_t\alpha^\dagger-\alpha\alpha^\dagger_t) -i(U\alpha^\dagger+\alpha V) & 2\alpha_t -iU+i\alpha V\alpha
\\ -2\alpha_t^\dagger-iV +i\alpha^\dagger U\alpha^\dagger & (\alpha_t^\dagger \alpha -\alpha^\dagger \alpha_t) +i(\alpha^\dagger U+V\alpha) \end{pmatrix}\,.
\notag
\end{gather}
\end{subequations}
Eqs.~\eref{e:QxQt} can be integrated explicitly, giving $Q(x,t)=i\theta(x,t,k)\Sigma_3$ where
\begin{gather}
\label{e:theta,xi}
\theta (x,t,k) = k \xi(x,t) - t/4k\,, \qquad \xi (x,t) = x - \bigint_{x}^{\infty} \left( \sqrt{1+||u_y||^2}-1\right) dy\,.
\end{gather}
Then, under the assumption $U(x,t) \to 0$ sufficiently rapidly as $x \to \pm \infty$, one can show that $\hat{X},\hat{T}\to 0$ in this limits and hence
define the Jost eigenfunctions
\begin{gather}
\hat{\Phi}_\pm(x,t,k) = \left( \hat{\Phi}_{\pm,1}(x,t,k), \hat{\Phi}_{\pm,2}(x,t,k) \right) \sim I_4 e^{-i \theta(x,t,k) \Sigma_3}, \quad x \to \pm \infty,
\end{gather}
as simultaneous solutions of the Lax pair \eref{e:4new}. It is convenient to consider modified eigenfunctions with constant asymptotic behavior
\begin{gather}
\label{e:defM}
M_{\pm}(x,t,k) = \left( M_{\pm,1}(x,t,k), M_{\pm,2}(x,t,k)\right) = \hat{\Phi}_{\pm}(x,t,k) e^{i \theta(x,t,k) \Sigma_3} \sim I_4, \quad x \to \pm \infty,
\end{gather}
and one can prove that the $4 \times 2$ columns $M_{-,1}, M_{+,2}$ are analytic for $k\in \mathbb{C}^+$ and continuous for $k \in \mathbb{R}$, and the columns $M_{+,1}, M_{-,2}$ are analytic for $k\in \mathbb{C}^-$ and continuous for $k \in \mathbb{R}$. Since $\hat{\Phi}_{+}$ and $\hat{\Phi}_{-}$ are two fundamental solutions of the Lax pair for any $k \in \mathbb{R}$, one can define a $4 \times 4$ matrix $S(k)$ (independent of $x,t$) such that
\begin{gather}\label{e:8new}
\hat{\Phi}_{-}(x,t,k) = \hat{\Phi}_{+}(x,t,k) S(k)\,, \quad S(k) = \begin{pmatrix}
a(k) & \bar{b}(k)\\
b(k) & \bar{a}(k)
\end{pmatrix}\,, \quad k \in \mathbb{R},
\end{gather}
whose $2 \times 2$ blocks are such that $a(k)$ (respectively, $\bar{a}(k)$) is analytic in $\mathbb{C}^+$ (respectively, in $\mathbb{C}^-$) and continuous for $k \in \mathbb{R}$, while $b(k), \bar{b}(k)$ are in general only defined for $k \in \mathbb{R}$.
Equation \eref{e:8new} for $k \in \mathbb{R}$ can be written as
\begin{subequations}
\label{e:28}
\begin{align}
M_{-,1}(x,t,k)a^{-1}(k) = M_{+,1}(x,t,k) + M_{+,2}(x,t,k) e^{2i\theta(x,t,k)}\rho(k)\,,\label{e:28a}\\
M_{-,2}(x,t,k)\bar{a}^{-1}(k)= M_{+,2}(x,t,k) + M_{+,1}(x,t,k) e^{-2i\theta(x,t,k)} \bar{\rho}(k)\,,\label{e:28b}
\end{align}
\end{subequations}
where the functions $M_{-,1}(x,t,k)a^{-1}(k)$ and $M_{-,2}(x,t,k)\bar{a}^{-1}(k)$ are meromorphic in the upper/lower half $k$-plane respectively, and
\begin{gather}
\label{e:reflection}
\rho(k) = b(k)a^{-1}(k)\,, \qquad \bar{\rho}(k) = \bar{b}(k) \bar{a}^{-1}(k)\, \qquad k \in \mathbb{R},
\end{gather}
are the (matrix) reflection coefficients.

For future reference, we note that one can also express the columns of $\hat{\Phi}_{+}$ in terms of the columns of $\hat{\Phi}_{-}$ as
\begin{gather}
\label{e:linearindependence3}
\hat{\Phi}_{+}(x,t,k) = \hat{\Phi}_{-}(x,t,k)S^{-1}(k)\,, \quad
S^{-1}(k)=\begin{pmatrix}
\bar{c}(k) & d(k)\\
\bar{d}(k) & c(k)
\end{pmatrix}\,, \quad
k \in \mathbb{R},
\end{gather}
where $c, d, \bar{c}, \bar{d}$ are $2 \times 2$ matrix functions of $k$, and \eref{e:linearindependence3} can be written in terms of the analytic groups of columns of the modified eigenfunctions as
\begin{subequations}\label{e:7new}
\begin{gather}
M_{+,1}(x,t,k)\bar{c}^{-1}(k) = M_{-,1}(x,t,k) + M_{-,2}(x,t,k) e^{2i\theta(x,t,k)}\bar{r}(k)\,,\label{e:7newa}\\
M_{+,2}(x,t,k)c^{-1}(k)= M_{-,2}(x,t,k) + M_{-,1}(x,t,k) e^{-2i\theta(x,t,k)}r(k)\,,\label{e:7newb}
\end{gather}
\end{subequations}
where $r(k) = d(k)c^{-1}(k)$ and $\bar{r}(k) = \bar{d}(k)\bar{c}^{-1}(k)$ are the (matrix) reflection coefficients from the right defined for system \eref{e:linearindependence3}.
For future convenience, we refer to $a(k),\bar{a}(k)$ as the (inverses) of the ``left'' transmission coefficients, and to $c(k),\bar{c}(k)$ as the (inverses) of the ``right'' transmission coefficients.

The Lax pair \eref{e:1.3b} admits two symmetries, $k \to k^*$ and $k \to -k^*$, which induce corresponding symmetries in the scattering data.
Specifically, the first symmetry implies
\begin{subequations}
\label{e:symm_data}
\begin{gather}
\bar{\rho}(k)=-\rho^\dagger(k)\,, \quad k\in \Real\,, \qquad \det \bar{a}(k)=\det a^\dagger(k^*)\,, \quad k\in \Complex^-\,,
\end{gather}
and the second symmetry gives
\begin{gather}
a^*(-k^*)=\sigma_2 a(k)\sigma_2\,, \quad k\in \Complex^+\,, \qquad \bar{a}^*(-k^*)=\sigma_2\bar{a}(k)\sigma_2\,, \quad k\in \Complex^-\,,\\
\rho^*(-k)=\sigma_2\rho(k)\sigma_2\,, \qquad \bar{\rho}^*(-k)=\sigma_2\bar{\rho}(k)\sigma_2\,, \qquad k\in \Real\,,
\end{gather}
\end{subequations}
where $\sigma_2 = \begin{pmatrix}
0 & -i\\
i & 0
\end{pmatrix}$ is the second Pauli matrix
(see \cite{ABBA} for details). The discrete spectrum consists of the values of $k \in \mathbb{C}/\mathbb{R}$, for which the scattering problem admits eigenfunctions in $L^2(\mathbb{R})$, and discrete eigenvalues appear in symmetric quartets:
\begin{gather}
Z = \{ k_n, -k_n^*, -k_n, k_n^*\}_{n=1}^N,
\end{gather}
where, for each $n$, $k_n, -k_n^*$ are the zeros of $\det a(k)$ and coincide with the values of $k\in \mathbb{C}^+$ where $\hat{\Phi}_{-,1}$ and $\hat{\Phi}_{+,2}$ become linearly dependent, and $-k_n, k_n^*$ are the zeros of $\det \bar{a}(k)$, which coincide with the values of $k\in \mathbb{C}^-$ where $\hat{\Phi}_{+,1}$ and $\hat{\Phi}_{-,2}$ become linearly dependent. Moreover:
\begin{itemize}
\item[1.] if $\mathrm{rank}\,\,a(k_n)\equiv \mathrm{rank}\,\,a(-k_n^*)=1$ and $\mathrm{rank}\,\,\bar{a}(-k_n)\equiv \mathrm{rank}\,\,\bar{a}(k_n^*)=1$, then the zeros of $\det a(k)$ in $\mathbb{C}^{+}$ and the zeros of $\det \bar{a}(k)$ in $\mathbb{C}^{-}$ are simple;
\item[2.] if $a(k_n)\equiv a(-k_n^*)=0_{2 \times 2}$ and $\bar{a}(-k_n)\equiv \bar{a}(k_n^*)=0_{2 \times 2}$, then the zeros of $\det a(k)$ in $\mathbb{C}^{+}$ and the zeros of $\det \bar{a}(k)$ in $\mathbb{C}^{-} $ are double.
\end{itemize}
In both cases, it is shown in \cite {ABBA} that the points $k_n,-k_n^*$ (resp., $k_n^*,-k_n$) are simple poles for the function $M_{-,1} a^{-1}$ (resp., $M_{-,2} \bar{a}^{-1}$) in $\mathbb{C}^{+}$ (resp., $\mathbb{C}^{-}$), and one can define the corresponding residues as follows
\begin{subequations}
\label{e:res}
\begin{gather}
\mathrm{Res}_{k=k_n}[M_{-,1}(x,t,k) a^{-1}(k)] = e^{2i(k_n \xi-t/4k_n)}M_{+,2}(x,t,k_n)C_n,\\
\mathrm{Res}_{k=-k_n^*}[M_{-,1}(x,t,k) a^{-1}(k)] = e^{2i(-k_n^* \xi+t/4k_n^*)}M_{+,2}(x,t,-k_n^*)\tilde{C}_n,\\
\mathrm{Res}_{k=k_n^*}[M_{-,2}(x,t,k) \bar{a}^{-1}(k)] = e^{-2i(k_n^* \xi-t/4k_n^*)}M_{+,1}(x,t,k_n^*)\bar{C}_n,\\
\mathrm{Res}_{k=-k_n}[M_{-,2}(x,t,k) \bar{a}^{-1}(k)] = e^{-2i(-k_n \xi+t/4k_n)}M_{+,1}(x,t,-k_n)\bar{\tilde{C}}_n,
\end{gather}
\end{subequations}
where $C_n$ is the $2\times2$ norming constant associated to the discrete eigenvalue $k_n$, and
\begin{gather}
\label{e:eq3}
\bar{C}_n = - C_n^\dagger, \quad \tilde{C}_n = - \sigma_2 C_n^* \sigma_2, \quad \bar{\tilde{C}}_n = - \sigma_2 \bar{C}_n^* \sigma_2\,.
\end{gather}
In the first case, i.e., when $a(k),\bar{a}(k)$ evaluated at the discrete eigenvalues are rank-1 matrices, the norming constants are rank-one matrices themselves; in the second case, the norming constants can be either full-rank or rank-one matrices. The above results were established in \cite{ABBA}, under the implicit assumption that the matrices $a$ and $\bar{a}$ have equal ranks. In Appendix~\ref{appendix:ranks}, we prove that $a(k)$ and $\bar{a}(k)$ necessarily have the same rank at each of the eigenvalues of a given quartet, namely $\mathrm{rank}\,\, a(k_n)=\mathrm{rank}\,\, a(-k_n^*)=\mathrm{rank}\,\, \bar{a}(k_n^*)=\mathrm{rank}\,\, \bar{a}(-k_n)$ for each $n$.

The starting point of the formulation of the inverse problem is equation \eref{e:28}, regarded as the jump condition across the real $k$-axis between the eigenfunctions that are meromorphic in $\mathbb{C}^+$, and those that are meromorphic in $\mathbb{C}^-$. Specifically, one introduces the sectionally meromorphic matrix function:
\[\mu_{\pm}(x(\xi,t),t,k)= \begin{cases}\label{e:jumpcondition}
      \big( M_{-,1}(x(\xi,t),t,k)a^{-1}(k)\,\,\, M_{+,2}(x(\xi,t),t,k)\big) \qquad k \in \Complex^{+}\\
     \big( M_{+,1}(x(\xi,t),t,k)\,\,\, M_{-,2}(x(\xi,t),t,k) \bar{a}^{-1}(k)\big) \qquad k \in \Complex^{-}
   \end{cases}
\]
and then defines
\begin{gather}\label{e:sectionfunction}
\hat{\mu}_{\pm}(\xi,t,k) = \mu_{\pm}\left( x(\xi,t), t,k\right), \qquad
\breve{\mu}_{\pm}(\xi,t,k)  = \hat{\mu}^{-1}_{\infty}(\xi,t)\hat{\mu}_{\pm}(\xi,t,k), \quad
\end{gather}
where $\hat{\ }$ denotes functions in which the $x$-dependence has been replaced by a $\xi$-dependence, and $\hat{\mu}_\infty(\xi,t)=\lim_{|k|\rightarrow \infty}\hat{\mu}_\pm(\xi,t,k)$ (see \cite{ABBA} for details). Using Eqs. \eref{e:sectionfunction}, \eref{e:jumpcondition}, one can write \eref{e:28} as
\begin{gather}\label{e:14new}
\breve{\mu}_{+} = \breve{\mu}_{-} [I_4 - \hat{J}], \quad \hat{J} = e^{-i \hat{\theta} \Sigma_3} J(k) e^{i \hat{\theta} \Sigma_3}, \quad J(k) = \begin{pmatrix}
\bar{\rho}(k) \rho(k) & \bar{\rho}(k)\\
-\rho(k) & 0_{2 \times 2}
\end{pmatrix},
\end{gather}
where $\hat{\theta} = k \xi - t/4k$, for all $k \in \mathbb{R}$, which, supplemented with the normalization condition
\begin{gather}\label{e:normalizationcondition}
\breve{\mu}_{\pm} = I_4, \quad k \to \infty,
\end{gather}
defines a Riemann Hilbert problem (RHP) with poles across the real $k$-axis. The formal solution of the RHP is then given by the system
\begin{align}
\label{e:RHP5}
\breve{\mu}(\xi,t,k) = I_4 + \mathlarger{\mathlarger{\mathlarger{\sum}}}_{n=1}^{\mathrm{\it{N}}} \dfrac{\mathrm{Res}_{k=k_n}\breve{\mu}_{+}(\xi,t,k)}{k-k_n}\notag\\
+ \mathlarger{\mathlarger{\mathlarger{\sum}}}_{n=1}^{\mathrm{\it{N}}} \dfrac{\mathrm{Res}_{k=-k_n^*}\breve{\mu}_{+}(\xi,t,k)}{k+k_n^*} + \mathlarger{\mathlarger{\mathlarger{\sum}}}_{n=1}^{\mathrm{\it{N}}} \dfrac{\mathrm{Res}_{k=k_n^*}\breve{\mu}_{-}(\xi,t,k)}{k-k_n^*}\notag\\+ \mathlarger{\mathlarger{\mathlarger{\sum}}}_{n=1}^{\mathrm{\it{N}}} \dfrac{\mathrm{Res}_{k=-k_n}\breve{\mu}_{-}(\xi,t,k)}{k+k_n}
-\frac{1}{2\pi i} \bigints_{\mathbb{R}}\dfrac{\breve{\mu}_{-}(\xi,t,\zeta)\hat{J}(\xi,t,\zeta)}{\zeta-(k \pm i0)}d\zeta,
\end{align}
which is closed using the residue conditions
\begin{subequations}\label{e:92}
\begin{gather}
\mathrm{Res}_{k=k_n} \breve{\mu}_{+,1}(\xi,t,k) = e^{2i \hat{\theta}(\xi,t,k_n)} \breve{\mu}_{+,2}(\xi,t,k_n) C_n,\,\label{e:92a}\\
\mathrm{Res}_{k=-k_n^*} \breve{\mu}_{+,1}(\xi,t,k) = e^{2i \hat{\theta}(\xi,t,-k_n^*)} \breve{\mu}_{+,2}(\xi,t,-k_n^*) \tilde{C}_n,\,\label{e:92b}\\
\mathrm{Res}_{k=k_n^*} \breve{\mu}_{-,2}(\xi,t,k) = e^{-2i \hat{\theta}(\xi,t,k_n^*)} \breve{\mu}_{-,1}(\xi,t,k_n^*) \bar{C}_n,\,\label{e:92c}\\
\mathrm{Res}_{k=-k_n} \breve{\mu}_{-,2}(\xi,t,k) = e^{-2i \hat{\theta}(\xi,t,-k_n)} \breve{\mu}_{-,1}(\xi,t,-k_n) \bar{\tilde{C}}_n\,,\label{e:92d}
\end{gather}
\end{subequations}
for each $n=1,2,...,\mathrm{\it{N}}$.

The last step of the inverse problem amounts to reconstructing the solution of the ccSPE from the solution $\breve{\mu}_{\pm}$ of the RHP. Specifically, as shown in \cite{ABBA}, the reconstruction formula is given by:
\begin{subequations}
\label{e:reconstruction}
\begin{gather}
\label{e:u}
\mathbf{u}(\xi,t) = \lim_{k\to 0}\frac{i}{k} \left(\left(\breve{\mu}^{-1}(\xi,t,0)\breve{\mu}(\xi,t,k)\right)_{1,3}, \left( \breve{\mu}^{-1}(\xi,t,0)\breve{\mu}(\xi,t,k)\right)_{1,4} \right)^{T},\\
x-\xi = \lim_{k\to 0}\frac{i}{k}\left[\left( \breve{\mu}^{-1}(\xi,t,0)\breve{\mu}(\xi,t,k)\right)_{1,1}-1\right]\,,\label{e:x}
\end{gather}
\end{subequations}
where \eref{e:x} expresses the original variable $x$ in terms of the travel-time parameter $\xi$. Whenever $x=x(\xi)$ in \eref{e:x} is monotonic, so that for each $x$ there is a unique $\xi$ such that $\xi=\xi(x)$, then \eref{e:u} and \eref{e:x} can be used to obtain the solution of the ccSPE in the original physical variables, namely $\mathbf{u}(x,t)$.

Pure solitons can be obtained by setting $\rho(k) = \bar{\rho}(k) = 0_{2 \times 2}$, for all $k \in \mathbb{R}$ in Eq.~\eref{e:RHP5}. In this case, the system reduces to the following linear algebraic system for the eigenfunctions
\begin{gather}
\label{eq28}
\breve{\mu}(\xi,t,k) = I_4 + \mathlarger{\mathlarger{\mathlarger{\sum}}}_{n=1}^{\mathrm{\it{N}}} \dfrac{\mathrm{Res}_{k=k_n}\breve{\mu}_{+}(\xi,t,k)}{k-k_n}\notag\\
+ \mathlarger{\mathlarger{\mathlarger{\sum}}}_{n=1}^{\mathrm{\it{N}}} \dfrac{\mathrm{Res}_{k=-k_n^*}\breve{\mu}_{+}(\xi,t,k)}{k+k_n^*} + \mathlarger{\mathlarger{\mathlarger{\sum}}}_{n=1}^{\mathrm{\it{N}}} \dfrac{\mathrm{Res}_{k=k_n^*}\breve{\mu}_{-}(\xi,t,k)}{k-k_n^*}+ \mathlarger{\mathlarger{\mathlarger{\sum}}}_{n=1}^{\mathrm{\it{N}}} \dfrac{\mathrm{Res}_{k=-k_n}\breve{\mu}_{-}(\xi,t,k)}{k+k_n},
\end{gather}
with residue conditions given by \eref{e:res}, which can be solved analytically.

One-soliton solutions are obtained by setting $\mathrm{\it{N}=1}$ in the last relation, and solving the corresponding system for the upper and lower blocks of the eigenfunctions. The different types of one-soliton solutions that ccSPE admits depend on the choice of the norming constant $C_1$. Specifically, if $k_1=\eta_1+i\nu_1\in \Complex^+$ and $C_1=(\pmb{\gamma}\,,\mathbf{0})$ with $\pmb{\gamma} = (\alpha_1,\beta_1)^{T}$, then the corresponding solution is a fundamental soliton, which is the natural vector generalization of scalar one-soliton solutions of the complex short-pulse equation. In this case, the vector solution of the ccSPE is given by:
\begin{subequations}\label{e:21}
\begin{gather}
\mathbf{u}(\xi,t) = \frac{i \nu_1}{|k_1|^2} e^{-i(\phi_1(\xi,t) - 2 \mathrm{arg}\, k_1)}\mathrm{sech}[\zeta_1(\xi,t) - x_0] \frac{\pmb{\gamma}^*}{||\pmb{\gamma}||},\label{e:21a} \\
x = \xi + \frac{2 \nu_1}{|k_1|^2} \frac{1}{1 + e^{2 (\zeta_1 - x_0)}}, \label{e:21b}
\end{gather}
\end{subequations}
where
\begin{gather*}
\zeta_1(\xi,t) = 2\nu_1(\xi + t/4|k_1|^2),\quad \phi_1(\xi,t) = 2\eta_1\xi - \eta_1t/2|k_1|^2, \quad x_0 = \log\frac{||\pmb{\gamma}||}{2\nu_1}.
\end{gather*}
If $C_1$ is a $2 \times 2$ rank-1 matrix with its columns being proportional to each other, say, $C_1=(\mu \pmb{\gamma}, \kappa \pmb{\gamma})$ for some multiplicative constants $\kappa, \mu \in \Complex$, the corresponding solution is a fundamental breather, which is a superposition of two orthogonally polarized fundamental solitons, with the same amplitude and velocity but different carrier frequencies. In this case, the vector solution is given by:
\begin{subequations}
\begin{gather}\label{e:fundbreather}
\mathbf{u}(\xi,t) = \frac{i \nu_1}{|k_1|^2 \sqrt{|\mu|^2 + |\kappa|^2}}\mathrm{sech}[\zeta_1(\xi,t) - x_0] \times \notag \\ \times \left[
e^{-i(\phi_1(\xi,t) - 2 \mathrm{arg} k_1)} \mu^* \frac{\pmb{\gamma}^*}{||\pmb{\gamma}||} + e^{i(\phi_1(\xi,t) - 2 \mathrm{arg} k_1)} \kappa \frac{(\pmb{\gamma}^*)^\perp}{||\pmb{\gamma}||} \right],\label{e:fundbreathera}\\
x = \xi + \frac{2 \nu_1}{|k_1|^2} \frac{1}{1 + e^{2(\zeta_1 - x_0)}}, \quad x_0 = \log\left[\frac{||\pmb{\gamma}|| \sqrt{|\mu|^2 + |\kappa|^2}}{2\nu_1}\right],\label{e:fundbreatherb}
\end{gather}
\end{subequations}
$\pmb{\gamma}^\perp=(-\beta_1^*,\alpha_1^*)^T$ is such that $\pmb{\gamma}^\dagger \pmb{\gamma}^\perp=0$, and the quantities $\zeta_1$ and $\phi_1$ are the same as in the fundamental soliton case. Notice that equation \eref{e:fundbreathera} implies that a generic fundamental breather solution of the ccSPE is governed by the vector ${\pmb{\gamma}}$, which controls the polarization vectors of the two orthogonally polarized solitons, and by the multiplicative constants $\mu$ and $\kappa$. The fundamental breather in Eq.~\eref{e:fundbreathera} reduces to a fundamental soliton by setting either $\kappa=0$ and $\mu=1$, or $\kappa=1$ and $\mu=0$. We point out that one between the two constants $\kappa$ and $\mu$ can always be scaled out, and in \cite{ABBA} $\mu=1$ was chosen without loss of generality. However, for the purpose of investigating soliton interactions it is convenient to keep both constants in.

Lastly, one can consider $C_1$ to be a $2 \times 2$ full-rank matrix, and in this case the corresponding solution is a composite breather, which still corresponds to a minimal set of discrete eigenvalues, but is a more complicated superposition of fundamental solitons. The explicit expression of a composite breather is more easily derived using Darboux transformations, and is given in Sec.~6 (see also Sec.~2.1 for the reduction to self-symmetric (composite) solitons).

Finally, one can show that whether the fundamental 1-soliton and fundamental 1-breather solution are smooth solitons or not depends on the location of the discrete eigenvalue $k_1=\eta_1+i\nu_1\in \Complex^+$. Specifically, like in the scalar case:
\begin{itemize}
    \item[1.] if $\nu_1 < |\eta_1|$, then the corresponding fundamental soliton and fundamental breather solution has a smooth envelope as a function of $x,t$ in both components;
    \item[2.] if $\nu_1 > |\eta_1|$, then this leads to a loop in the envelope of each component of the solution (fundamental loop solitons/breathers);
    \item[3.] if $\nu_1 = |\eta_1|$, then this leads to a cusp in the envelope of each component of the solution (fundamental cuspon solitons/breathers).
\end{itemize}
Note that in the case of a generic composite breather, it was shown in \cite{ABBA} that the condition $\nu_1<|\eta_1|$ is not sufficient to guarantee smoothness of the solution in terms of $x,t$, and to the best of our knowledge, a regularity condition, which will necessarily have to involve the norming constant, is presently not known.

\subsection{Self-symmetric discrete eigenvalues and associated one-soliton solutions}

If an eigenvalue $k_j=i\nu_j$ is purely imaginary, then $-k_j^*=k_j$ and we refer to these eigenvalues as self-symmetric eigenvalues. In this case, the quartet of eigenvalues reduces to a pair $k_j,k_j^*$. The coalescence of the discrete eigenvalues $k_j$ and $-k_j^*$ and $k_j^*$ and $-k_j$ induces a coalescence of the corresponding norming constants,
namely, one needs to have:
\begin{gather}
\label{e:eq1}
C_j = \tilde{C}_j, \quad \bar{C}_j = \bar{\tilde{C}}_j.
\end{gather}
In addition, the symmetry relations \eref{e:eq3} among the norming constants still hold in the case of self-symmetric discrete eigenvalues, and combining them with \eref{e:eq1} yields a constraint on the entries of the norming constants associated to each self-symmetric eigenvalue, namely:
\begin{equation}
\label{e:self-symm}
C_j=-\sigma_2C_j^*\sigma_2\,.
\end{equation}
One can easily verify that if $C_j$ is a rank-1 matrix, then \eref{e:self-symm} implies  $C_j=0$ and hence in the case of a self-symmetric eigenvalue the only nontrivial solutions are associated to full-rank norming constants. Let
$C_j = \begin{pmatrix}
\alpha_j & \gamma_j\\
\beta_j & \delta_j
\end{pmatrix}$,
where $\alpha_j, \beta_j, \gamma_j, \delta_j \in \mathbb{C}$ and $\det C_j = \alpha_j \delta_j - \beta_j \gamma_j \neq 0$. In this case, imposing the constraint \eref{e:self-symm}, we obtain $\gamma_j = \beta_j^*$ and $\delta_j = - \alpha_j^*$, which then gives
\begin{gather}
\label{e:selfsymmC}
C_j = \begin{pmatrix}
\alpha_j & \beta_j^* \\
\beta_j & -\alpha_j^*
\end{pmatrix}, \quad \bar{C}_j = \begin{pmatrix}
-\alpha_j^* & -\beta_j^*\\
-\beta_j & \alpha_j
\end{pmatrix}.
\end{gather}
In the case of self-symmetric eigenvalues, instead of subtracting two residues in each half plane, we now subtract one residue from both sides of the jump condition \eref{e:sectionfunction} that the matrix function $\breve{\mu}$ satisfies, and therefore this has an immediate consequence in the formal solution of the RHP. For instance, for a pure self-symmetric 1-soliton solution corresponding to a pair of eigenvalues $k_1,k_1^*$ on the imaginary axis and associated norming constants $C_1,\bar{C}_1$ as in \eref{e:selfsymmC}, one can replace the definitions of the residues \eref{e:res} and derive the equations which hold for the upper and lower blocks of $\breve{\mu}_{\pm}$ in the analog of Eqs.~(28) evaluated at $k_1$ and $k_1^*$ respectively, i.e.,
\begin{subequations}\label{e:79}
\begin{gather}
\breve{\mu}_{-,1}^{\mathrm{up}}(k_1^*) = I_2 + \frac{e^{2i \hat{\theta}_1}}{k_1^* - k_1} \breve{\mu}_{+,2}^{\mathrm{up}}(k_1)C_1, \quad \breve{\mu}_{-,1}^{\mathrm{dn}}(k_1^*) = \frac{e^{2i \hat{\theta}_1}}{k_1^* - k_1} \breve{\mu}_{+,2}^{\mathrm{dn}}(k_1)C_1\,,\label{e:79a}\\
\breve{\mu}_{+,2}^{\mathrm{up}}(k_1) = \frac{e^{2i \hat{\theta}_1}}{k_1 - k_1^*} \breve{\mu}_{-,1}^{\mathrm{up}}(k_1^*)\bar{C}_1, \quad
\breve{\mu}_{+,2}^{\mathrm{dn}}(k_1) = I_2 + \frac{e^{2i \hat{\theta}_1}}{k_1 - k_1^*} \breve{\mu}_{-,1}^{\mathrm{dn}}(k_1^*)\bar{C}_1\,\label{e:79b},
\end{gather}
\end{subequations}
where $\hat{\theta}_1=k_1\xi-t/4k_1\equiv i(\nu_1\xi+t/4\nu_1)$.
Replacing the first half of \eref{e:79b} into the first half of \eref{e:79a}, and the second half of \eref{e:79a} into the second half of \eref{e:79b}, we obtain the following expressions for the upper/lower blocks of the eigenfunctions
\begin{subequations}\label{e:80}
\begin{gather}
\breve{\mu}_{-,1}^{\mathrm{up}}(k_1^*) = \left( I_2 + \frac{e^{4i \hat{\theta}_1}}{(k_1^* - k_1)^2} \bar{C}_1 C_1\right)^{-1}\,,\label{e:80a}\\
\breve{\mu}_{+,2}^{\mathrm{up}}(k_1) = \frac{e^{2i \hat{\theta}_1}}{k_1 - k_1^*} \left( I_2 + \frac{e^{4i \hat{\theta}_1}}{(k_1^* - k_1)^2} \bar{C}_1 C_1\right)^{-1} \bar{C}_1\,,\label{e:80b}\\
\breve{\mu}_{+,2}^{\mathrm{dn}}(k_1) = \left( I_2 + \frac{e^{4i \hat{\theta}_1}}{(k_1^* - k_1)^2} C_1 \bar{C}_1 \right)^{-1}\,,\label{e:80c}\\
\breve{\mu}_{-,1}^{\mathrm{dn}}(k_1^*) = \frac{e^{2i \hat{\theta}_1}}{k_1^* - k_1} \left( I_2 + \frac{e^{4i \hat{\theta}_1}}{(k_1^* - k_1)^2} C_1 \bar{C}_1 \right)^{-1} C_1\,.\label{e:80d}
\end{gather}
\end{subequations}
Using the above eigenfunctions in the reconstruction formula \eref{e:reconstruction}, we obtain the expression of the self-symmetric soliton:
\begin{gather}
\mathbf{u}(\xi,t) = -\frac{1}{k_1} \sech[\zeta_1(\xi,t)- x_0]\frac{ \pmb{\gamma^*}}{||
\pmb{\gamma}||}\,, \qquad \zeta_1(\xi,t)=2\nu_1(\xi+t/4\nu_1^2)\,,\\
x = \xi + \frac{2}{\nu_1^2}\frac{1}{1 + e^{2(\zeta_1 - x_0)}}, \quad \pmb{\gamma} = \begin{pmatrix}
\alpha_j\\
\beta_j
\end{pmatrix}, \quad x_0 = \log\frac{||
\pmb{\gamma}||}{2 \nu_1}\,. \notag
\end{gather}
The special cases $\beta=0$ (a diagonal, self-symmetric norming constant) and $\alpha=0$ (an off-diagonal, self-symmetric norming constant) provide scalar solutions with
$u_1\equiv 0$ and $u_2\equiv 0$, respectively.

\section{Real solutions of the ccSPE}

The ccSP equation admits real solutions, and in this case the matrix potential $U$ in \eref{e:U} satisfies the symmetry
\begin{gather}
U^* = - U, \qquad V=-U^T\,.
\end{gather}
This induces an additional symmetry on the Lax operators
\begin{gather}
X^*(x,t,-k^*) = X(x,t,k), \quad T^*(x,t,-k^*) = T(x,t,k),
\end{gather}
and assuming uniqueness of solution of the equations of the Lax pair with prescribed boundary conditions as $x \to \pm \infty$, then the last symmetry implies
\begin{gather}
\Phi_{\pm}^*(x,t,-k^*) = \Phi_{\pm}(x,t,k).
\end{gather}
Recall that $\Phi_{\pm}$ and $\hat{\Phi}_{\pm}$ are related through the gauge transformation \eref{e:gauge1}, and since $P(x,t)$ is unitary, one also has
\begin{gather}
\label{e:symeigen}
\hat{\Phi}_{\pm}^*(x,t,-k^*) =\hat{\Phi}_{\pm}(x,t,k).
\end{gather}
In turn, combining this last symmetry with the definition of the scattering matrix \eref{e:8new}, we obtain
\begin{gather}
S^*(-k^*) = S(k),
\end{gather}
which implies additional symmetries on the scattering data and the reflection coefficients
\begin{subequations} \label{e:symscdata}
\begin{gather}
a^*(-k^*) = a(k), \quad \bar{a}^*(-k^*) = \bar{a}(k),\label{e:symscdata1}\\
b^*(-k^*) = b(k), \quad \bar{b}^*(-k^*) = \bar{b}(k),\label{e:symscdata2}\\
\rho^*(-k^*) = \rho(k), \quad \bar{\rho}^*(-k^*) = \bar{\rho}(k).\label{e:symscdata3}
\end{gather}
\end{subequations}
From \eref{e:symm_data}, it then follows that for real solutions of the ccSPE
\begin{gather}
\label{e:symmrealdeta}
\det a(k) = \det a^*(-k^*),
\end{gather}
but the reduction to real solutions does not induce an additional symmetry for the zeros of $\det a(k)$ and $\det \bar{a}(k)$. Hence, the discrete spectrum consists of either the quartet of points $\{k_n, -k_n^*, k_n^*, -k_n\}$, where $k_n, -k_n^*$ are the zeros of $\det a(k)$ on the UHP of $\mathbb{C}$, and $k_n^*, -k_n$ are the zeros of $\det \bar{a}(k)$ on the LHP of $\mathbb{C}$, or the pair of points $\{k_n, k_n^*\}$, when the discrete eigenvalues are self-symmetric, i.e., purely imaginary.
On the other hand, the reduction to real solutions imposes additional symmetries on the norming constants. Let us assume that $k_n, -k_n^*$ are distinct simple zeros of $\det a(k)$ in $\mathbb{C}^+$, and therefore $-k_n, k_n^*$ are distinct simple zeros of $\det \bar{a}(k)$ in $\mathbb{C}^-$. Since the discrete eigenvalues in $\Complex^+$ coincide with the values of $k$ where the columns $\hat{\Phi}_{-,1}$ and $\hat{\Phi}_{+,2}$ become linearly dependent, we then have
\begin{gather}\label{e:prcon}
\hat{\Phi}_{-,1}(k_n) \alpha(k_n) = \hat{\Phi}_{+,2}(k_n) c_n, \quad \hat{\Phi}_{-,1}(-k_n^*) \alpha(-k_n^*) = \hat{\Phi}_{+,2}(-k_n^*) \tilde{c}_n,
\end{gather}
where $c_n, \tilde{c}_n$ are the corresponding proportionality constants, and $\alpha(k)$ is the cofactor matrix of $a(k)$. Then, combining relations \eref{e:symeigen}, \eref{e:symscdata1} and \eref{e:prcon}, we get
\begin{gather}\label{e:15norming}
\tilde{c}_n = c_n.
\end{gather}
Now, the function $M_{-,1} a^{-1}$ is meromorphic for $k\in\mathbb{C}^+$, with simple poles at the points $k_n, -k_n^*$, and one can define its residues as in \cite{ABBA}, and the norming constants are given by
\begin{gather}\label{e:16normcon}
C_n = \frac{c_n}{\left( \det a \right)'(k_n)}, \quad \tilde{C}_n = \frac{\tilde{c}_n}{\left( \det a \right)'(-k_n^*)}.
\end{gather}
It is easy to check that
\begin{gather}\label{e:17normcon}
\tilde{C}_n = - C_n^*,
\end{gather}
since \eref{e:15norming} holds, and \eref{e:symmrealdeta} implies $\left( \det a \right) ' (-k_n^*) = - \left[ \left( \det a \right) '(k_n)\right]^*$, where prime denotes differentiation with respect to $k$. Similarly, considering $M_{-,2} \bar{a}^{-1}$, which has simple poles at the points $k_n^*, -k_n\in \mathbb{C}^-$, one can show that the associated norming constants satisfy:
\begin{gather}
\label{e:18normcon}
\bar{\tilde{C}}_n = - \bar{C}_n^*.
\end{gather}
When the zeros of $\det a(k)$ and $\det\bar{a}(k)$ are simple, it has been shown in \cite{ABBA} that the norming constants are necessarily rank-one matrices, which implies that $C_n$ has either one column identically zero, or the two columns proportional to each other. Combining \eref{e:17normcon} with the second of \eref{e:eq3},
one can conclude that no non-trivial norming constant exists when one of the columns of $C_n$ is identically zero. This implies that no real fundamental soliton solutions exist.
Now, if the norming constant $C_n$ has proportional columns, i.e., $C_n = ( \pmb{\gamma}, \kappa \pmb{\gamma})$, where $\pmb{\gamma} = (\alpha_n\,, \beta_n)^{T}$ with $\alpha_n, \beta_n, \kappa \in \mathbb{C}$,
then one can prove that Eqs. \eref{e:17normcon} and \eref{e:18normcon} combined give a constraint for the multiplicative constant $\kappa$, i.e., $C_n$ is non-trivial iff $\kappa=\pm i$. The fundamental breather solution is then given by \eref{e:fundbreather} with $\mu=1$ and $\kappa=\pm i$.

Let us now assume that $k_n, -k_n^*$ are distinct double zeros of $\det a(k)$ in $\mathbb{C}^+$, and therefore $k_n^*, -k_n$ are distinct double zeros of $\det \bar{a}(k)$ in $\mathbb{C}^-$. It has been shown that they can still be simple poles for the functions $M_{-,1} a^{-1}$ and $M_{-,2} \bar{a}^{-1}$, and one can define the corresponding residues and norming constants, using similar relations as in \eref{e:16normcon}. The symmetries \eref{e:eq3} and \eref{e:17normcon} hold as well, and in this case the norming constants are either rank-one or full-rank matrices. We already analyzed the cases when $C_n$ is a  rank-one matrix, and therefore let us assume
\begin{gather}\label{e:21normcon}
C_n = \begin{pmatrix}
\alpha_n & \gamma_n\\
\beta_n & \delta_n
\end{pmatrix}, \quad \alpha_n, \beta_n, \gamma_n, \delta_n \in \mathbb{C},
\end{gather}
with $\det C_n \neq 0$. Eqs. \eref{e:17normcon} and \eref{e:18normcon} combined give the following constraint for the entries of $C_n$
\begin{gather}
\delta_n = \alpha_n, \quad \gamma_n = - \beta_n,
\end{gather}
and therefore $C_n$ takes the form
\begin{gather}
C_n = \begin{pmatrix}
\alpha_n & -\beta_n\\
\beta_n & \alpha_n
\end{pmatrix}, \quad \alpha_n, \beta_n \in \mathbb{C}.
\end{gather}

Finally, let us now consider the case of self-symmetric eigenvalues, i.e., an eigenvalue pair $\{ k_n=i\nu_n, k_n^*=-i\nu_n\}$. In this case, there are three symmetries for the norming constants which need to be satisfied simultaneously
\begin{gather}\label{e:25normcon}
\tilde{C}_n = -\sigma_2 C_n^* \sigma_2, \quad \tilde{C}_n = -C_n^*, \quad \tilde{C}_n = C_n\,.
\end{gather}
One can show that no non-trivial norming constants exists when $C_n$ is a rank-one matrix. When $C_n$ is full-rank, as in \eref{e:21normcon}, imposing the symmetries \eref{e:25normcon} requires all entries of $C_n$ to be purely imaginary, i.e., $C_n$ takes the form $C_n = ( \pmb{\gamma},-\pmb{\gamma}^\perp)$, where $\pmb{\gamma}=(i \alpha_n, i \beta_n)$, with $\alpha_n, \beta_n \in \mathbb{R}$.

\section{Pure soliton transmission coefficients}
In order to use Manakov's method to investigate the pairwise interactions of two vector
solitons \cite{Manakov}, one needs the explicit expressions of the (matrix) transmission coefficients corresponding to pure one-soliton solutions, namely
the inverses of the matrices $a(k),\bar{a}(k)$, which correspond to the transmission coefficients from the left, and the inverses of the matrices $c(k), \bar{c}(k)$, which correspond to the transmission coefficients from the right.
Both left and right transmission coefficients corresponding to 1-fundamental soliton, 1-fundamental breather and
1-self-symmetric soliton can be computed as limits of suitable blocks of analytic eigenfunctions as $\xi\rightarrow +\infty$.
We give below the results, the calculations are somewhat involved, and the details are provided in App.~C.
\paragraph{Transmission coefficients of a 1-fundamental soliton.} In the case of one fundamental solitons, the (inverses of the) left and right transmission coefficients are given by
\begin{subequations}\label{e:scdata}
\begin{gather}
a_j(k) = \mathrm{diag}\left( \frac{k_j^*}{k_j} \frac{k-k_j}{k-k_j^*}\,, \frac{k_j}{k_j^*}\frac{k+k_j^*}{k+k_j}\right), \label{e:scdataa}\\
\bar{c}_j(k) = \mathrm{diag}\left( \frac{k_j}{k_j^*} \frac{k-k_j^*}{k-k_j}\,, \frac{k_j^*}{k_j}\frac{k+k_j}{k+k_j^*}\right), \label{e:scdatab}\\
\bar{a}_j(k) = I_2 + \frac{k}{k_j^*}\frac{k_j^* - k_j}{k - k_j} \frac{C_j \bar{C}_j}{||\pmb{\gamma}_j||^2} + \frac{k}{k_j}\frac{k_j - k_j^*}{k + k_j^*} \frac{\tilde{C}_j \bar{\tilde{C}}_j}{||\pmb{\gamma}_j||^2} , \label{e:scdatac}\\
c_j(k) = I_2 + \frac{k}{k_j}\frac{k_j - k_j^*}{k - k_j^*} \frac{C_j \bar{C}_j}{||\pmb{\gamma}_j||^2} + \frac{k}{k_j^*}\frac{k_j^* - k_j}{k + k_j}\frac{\tilde{C}_j \bar{\tilde{C}}_j}{||\pmb{\gamma}_j||^2},\label{e:scdatad}
\end{gather}
\end{subequations}
where the norming constants correspond to a fundamental soliton, and hence $C_j$ has the form $C_j=(\pmb{\gamma}_j\,, \mathbf{0})$, and the other ones are obtained from \eref{e:eq3}. Introducing the unit-norm polarization vectors
\begin{gather}
\label{e:69new}
\mathbf{p}_j = \frac{{\pmb{\gamma}}^*_j}{||{\pmb{\gamma}}_j||}, \quad \hat{\mathbf{p}}_j = \frac{({\pmb{\gamma}}_j^*)^\perp}{||{\pmb{\gamma}}_j||},
\end{gather}
we can rewrite Eqs. \eref{e:scdatac}-\eref{e:scdatad} as follows
\begin{subequations}
\begin{gather}
\bar{a}_j(k, \mathbf{p}_j) = I_2 - \frac{k}{k_j^*}\frac{k_j^* - k_j}{k - k_j} \mathbf{p}_j^* \mathbf{p}_j^T - \frac{k}{k_j}\frac{k_j - k_j^*}{k + k_j^*} \hat{\mathbf{p}}_j^* \hat{\mathbf{p}}_j^T,\\
c_j(k, \mathbf{p}_j) = I_2 - \frac{k}{k_j}\frac{k_j - k_j^*}{k - k_j^*} \mathbf{p}_j^* \mathbf{p}_j^T - \frac{k}{k_j^*}\frac{k_j^* - k_j}{k + k_j} \hat{\mathbf{p}}_j^* \hat{\mathbf{p}}_j^T,
\end{gather}
\end{subequations}
which, using the property $\mathbf{p}_j^* \mathbf{p}_j^T + \hat{\mathbf{p}}_j^* \hat{\mathbf{p}}_j^T = I_2$, can be simplified to
\begin{subequations}\label{e:74}
\begin{gather}
\bar{a}_j(k,\mathbf{p}_j) = \frac{k_j^*}{k_j}\frac{k + k_j}{k + k_j^*} \left[ I_2 + \frac{k^2\left(k_j^2 - (k_j^*)^2\right)}{(k_j^*)^2 \left( k^2 - k_j^2\right)} \mathbf{p}_j^* \mathbf{p}_j^T \right] , \label{e:74a}\\
c_j(k,\mathbf{p}_j) = \frac{k_j}{k_j^*}\frac{k + k_j^*}{k + k_j}\left[ I_2 + \frac{k^2\left((k_j^*)^2 - k_j^2\right)}{k_j^2 \left( k^2 - (k_j^*)^2\right)}
\mathbf{p}_j^* \mathbf{p}_j^T \right]. \label{e:74b}
\end{gather}
\end{subequations}
\paragraph{Transmission coefficients of a 1-fundamental breather.} In the case of one fundamental breather, the (inverses of the) left and right transmission coefficients are given by Eqs.~\eref{e:trans_fb}, where the norming constant corresponds to a fundamental breather, i.e., $C_j=(\mu_j \pmb{\gamma}_j\,, \kappa_j\, \pmb{\gamma}_j)$, for $\mu_j, \kappa_j \in \mathbb{C}$. Note that the products of the norming constants in Eqs.~\eref{e:trans_fb} can be written as follows
\begin{subequations}\label{e:149new}
\begin{gather}
\bar{C}_j C_j = - ||\pmb{\gamma}_j||^2 \pmb{\delta}_j \pmb{\delta}_j^\dagger, \quad \bar{\tilde{C}}_j \tilde{C}_j = - ||\pmb{\gamma}_j||^2  \pmb{\hat{\delta}}_j\pmb{\hat{\delta}}_j^\dagger, \quad \pmb{\delta}_j = \begin{pmatrix}
\mu_j^*\\
\kappa_j^*
\end{pmatrix}, \quad \pmb{\hat{\delta}}_j = \pmb{\delta}_j^\perp,\label{e:149newa}\\
C_j \bar{C}_j = - ||\pmb{\delta}_j||^2\pmb{\gamma}_j \pmb{\gamma}_j^\dagger, \quad \tilde{C}_j \bar{\tilde{C}}_j = -  ||\pmb{\delta}_j||^2 \hat{\pmb{\gamma}_j}^\dagger \hat{\pmb{\gamma}_j}^\dagger, \quad  \hat{\pmb{\gamma}}_j =\pmb{\gamma}_j^\perp,\label{e:149newb}
\end{gather}
\end{subequations}
and therefore the transmission coefficients can be simplified into
\begin{subequations}\label{e:fbreatherdata}
\begin{gather}
a_j(k,\pmb{\delta}_j) = I_2 - \frac{k}{k_j}\frac{k_j - k_j^*}{k-k_j^*}\frac{\pmb{\delta}_j \pmb{\delta}_j^\dagger}{||\pmb{\delta}_j||^2} + \frac{k}{k_j^*}\frac{k_j-k_j^*}{k+k_j}\frac{\pmb{\hat{\delta}}_j \pmb{\hat{\delta}}_j^\dagger}{||\pmb{\delta}_j||^2},\label{e:fbreatherdataa}\\
\bar{c}_j(k,\pmb{\delta}_j) = I_2 - \frac{k}{k_j^*}\frac{k_j^*-k_j}{k-k_j} \frac{\pmb{\delta}_j \pmb{\delta}_j^\dagger}{ ||\pmb{\delta}_j||^2} + \frac{k}{k_j}\frac{k_j^* - k_j}{k+k_j^*}\frac{\pmb{\hat{\delta}}_j \pmb{\hat{\delta}}_j^\dagger}{ ||\pmb{\delta}_j||^2},\label{e:fbreatherdatab}\\
c_j(k,\pmb{\gamma}_j) = I_2 - \frac{k}{k_j}\frac{k_j-k_j^*}{k-k_j^*}\frac{\pmb{\gamma}_j \pmb{\gamma}_j^\dagger}{||\pmb{\gamma}_j||^2} + \frac{k}{k_j^*}\frac{k_j-k_j^*}{k+k_j}\frac{\hat{\pmb{\gamma}_j} \hat{\pmb{\gamma}_j}^\dagger}{||\pmb{\gamma}_j||^2},\label{e:fbreatherdatac}\\
\bar{a}_j(k,\pmb{\gamma}_j) = I_2 - \frac{k}{k_j^*}\frac{k_j^*-k_j}{k-k_j}\frac{\pmb{\gamma}_j \pmb{\gamma}_j^\dagger}{||\pmb{\gamma}_j||^2} + \frac{k}{k_j}\frac{k_j^* - k_j}{k+k_j^*}\frac{\hat{\pmb{\gamma}_j} \hat{\pmb{\gamma}_j}^\dagger}{||\pmb{\gamma}_j||^2}.\label{e:fbreatherdatad}
\end{gather}
\end{subequations}

\noindent Eqs. \eref{e:fbreatherdata} can be further simplified to
\begin{subequations}\label{e:162new}
\begin{gather}
a_j(k,\mathbf{q}_j) = \frac{k_j}{k_j^*}\frac{k+k_j^*}{k+k_j}\left[ I_2 + \frac{k^2 \left( (k_j^*)^2-k_j^2\right)}{k_j^2\left( k^2 - (k_j^*)^2 \right)} \mathbf{q}_j^* \mathbf{q}_j^T\right],\label{e:162newa}\\
\bar{c}_j(k,\mathbf{q}_j) = \frac{k_j^*}{k_j}\frac{k+k_j}{k+k_j^*}\left[ I_2 + \frac{k^2 \left(k_j^2-(k_j^*)^2\right)}{k_j^2\left( k^2 - k_j^2 \right)} \mathbf{q}_j^* \mathbf{q}_j^T \right],\label{e:162newb}\\
c_j(k, \mathbf{p}_j) = \frac{k_j(k+k_j^*)}{k_j^*(k+k_j)}\left[ I_2 + \frac{k^2\left( (k_j^*)^2 - k_j^2 \right)}{k_j^2\left( k^2 - (k_j^*)^2\right)} \mathbf{p}_j^* \mathbf{p}_j^T \right],\label{e:162newc}\\
\bar{a}_j(k, \mathbf{p}_j) = \frac{k_j^*(k+k_j)}{k_j(k+k_j^*)}\left[ I_2 + \frac{k^2\left( k_j^2 - (k_j^*)^2 \right)}{(k_j^*)^2\left( k^2 - k_j^2\right)} \mathbf{p}_j^* \mathbf{p}_j^T\right],\label{e:162newd}
\end{gather}
\end{subequations}
where $\mathbf{p}_j=\pmb{\gamma}_j^*/||\pmb{\gamma}_j||$ and $\mathbf{q}_j=\pmb{\delta}_j^*/||\pmb{\delta}_j||$, and we used once again the properties $\mathbf{p}_j^* \mathbf{p}_j^T + \hat{\mathbf{p}}_j^* \hat{\mathbf{p}}_j^T = I_2$ and $\mathbf{q}_j^* \mathbf{q}_j^T + \hat{\mathbf{q}}_j^* \hat{\mathbf{q}}_j^T = I_2$, to eliminate $\hat{\pmb{\gamma}}_j$ and $\hat{\pmb{\delta}}_j$, respectively. The above equations show that the expressions for $\bar{a}_j$ and $c_j$ are exactly the same both for a fundamental soliton and a fundamental breather solution. Also, note that if we set $\kappa_j=0$, then the expressions of the transmission coefficients $a_j$ and $\bar{c}_j$ reduce to those which hold for a fundamental soliton. Moreover, one can easily compute the determinants of $a_j$ and $\bar{a}_j$, for either a fundamental soliton or a fundamental breather, and obtain the following
\begin{gather}\label{e:determ}
\det a_j(k) = \frac{(k-k_j)(k+k_j^*)}{(k+k_j)(k-k_j^*)}, \quad \det \bar{a}_j(k) = \frac{(k-k_j^*)(k+k_j)}{(k+k_j^*)(k-k_j)},
\end{gather}
as expected.
Also, using \eref{e:determ} is easy to verify the symmetry
\begin{gather}
\det \bar{a}_j(k) = \left( \det a_j(k^*) \right)^*, \quad k \in \mathbb{C}^{-} \cup \mathbb{R},
\end{gather}
for both a fundamental soliton and a fundamental breather.
\noindent
\paragraph{Transmission coefficients of a self-symmetric soliton.} In the case of one self-symmetric composite breather, i.e., when the discrete spectrum consists of a purely imaginary discrete eigenvalue and the associated norming constant $C_j$ is of the form $C_j = (\pmb{\gamma}_j\,, -\pmb{\gamma}_j^\perp)$, with $\pmb{\gamma}_j = (\alpha_j\,,\beta_j)^T$, and $\alpha_j, \beta_j \in \mathbb{C}$, the (inverse of the) left and right transmission coefficients are given by:
\begin{gather}\label{e:ss_transmission}
a_j(k) =\frac{k_j - k}{k_j + k} I_2, \quad \bar{a}_j(k) = \frac{k_j^* - k}{k_j^* + k} I_2, \quad c_j(k) \equiv a_j(k), \quad \bar{c}_j(k) \equiv \bar{a}_j(k),
\end{gather}
which are all diagonal (in fact, proportional to the identity matrix) and independent of the norming constant. Note also that their determinants have double zeros at the discrete eigenvalues, as expected. Also, $a_j(k_j)= c_j(k_j) = 0$, and similarly for $\bar{a}, \bar{c}$ at $k_j^*$.

\section{Fundamental soliton interactions (Manakov's method)}

Let us consider a 2-soliton solution of the ccSPE which corresponds to a pair of discrete eigenvalues $k_j = \eta_j + i \nu_j$ and associated norming constants $C_j$ for $j=1,2$. The problem of interacting solutions can be investigated by looking at the asymptotic states of the solution as $t \to \pm \infty$. We assume that if the individual solitons travel at different velocities, in the backward, i.e., as $t \to - \infty$, and forward, i.e., as $t \to + \infty$, long-time limits, a 2-solution breaks up into two individual solitons, i.e.
\begin{gather}
\label{e:2}
\mathbf{u}(\xi,t)\sim \mathbf{u}^{\pm}(\xi,t)=\mathbf{u}_1^{\pm}(\xi,t) + \mathbf{u}_2^{\pm}(\xi,t), \quad t \to \pm \infty,
\end{gather}
where $\mathbf{u}$ is the $2$-soliton solution of the ccSPE, and $\mathbf{u}_j^{\pm}$ are $1$-soliton solutions. If the two solitons are both fundamental solitons, we expect
$\mathbf{u}_j^{\pm}(\xi,t)=\mathbf{p}_j^\pm u_j^\pm(\xi,t)$ where $u_j^\pm(\xi,t)$ are 1-soliton solutions of the scalar cSPE, and the (constant) unit-norm vectors $\mathbf{p}_j^\pm$ are the asymptotic polarizations of the solitons (cf. \eref{e:21}).

Without loss of generality, we choose the discrete eigenvalues $k_j$ such that soliton-1 moves faster than soliton-2, i.e., we assume  $|\mathrm{v}_1| > |\mathrm{v}_2|$  (recall that the soliton velocity $\mathrm{v}_j \equiv -1/4|k_j|^2$, for $j=1,2$, and so solitons always move to the left). For each $j=1,2$, let $S_j^{-}$ be a $2\times 2$ complex matrix that, together with the discrete eigenvalue $k_j$ determines $\mathbf{u}_j^{-}$ as $t \to - \infty$, and let us denote the corresponding matrices as $t \to + \infty$ by $S_j^{+}$.  In other words, $S_j^\pm$ play the role of the norming constant $C_j$ in determining in the 2-soliton solution the form of soliton $j$ in the limits $t\rightarrow \pm \infty$, according to Eqs.~\eref{e:21} and \eref{e:fundbreather}. To investigate the result of the interaction, we trace the passage of the eigenfunctions through the asymptotic states, following the method developed by Manakov in \cite{Manakov} (see also \cite{APT}). Assuming $|t|$ is large enough so that the solitons are well-separated,  we denote the centers of the two solitons by $x_j$. For the given choice of the soliton velocities, if $t \to -\infty$ then $x_2 \ll x_1$, and the order is reversed as $t \to + \infty$.
In the following, we will assume that both solitons are fundamental solitons. We then have the following: as $t \to - \infty$ and starting from $x \ll x_2$, the eigenfunction $\hat{\Phi}_{-,1}(k_j)$ has the form
\begin{gather}
\label{e:64phi}
\hat{\Phi}_{-,1}(k_j) \sim e^{-i k_j \xi} \begin{pmatrix}
I_2 \\
0
\end{pmatrix}, \quad x \ll x_2,
\end{gather}
for $j=1,2$, where $\xi=\xi(x,t)$ is defined in \eref{e:theta,xi}.
After passing through soliton-2 (the leftmost soliton as $t\to -\infty$), since the corresponding state is a bound state for $k=k_2$ but not for $k=k_1$, then $\hat{\Phi}_{-,1}$ evaluated at $k=k_1$, is given by \eref{e:64phi} multiplied by the corresponding soliton transmission coefficient, while $\hat{\Phi}_{-,1}$ evaluated at $k=k_2$ behaves like $\hat{\Phi}_{-,1}(k_2)$ multiplied by $S_2^-$, because $k_2$ is exactly the value where $\hat{\Phi}_{-,1}$ and $\hat{\Phi}_{+,2}$ become linearly dependent, cf. \eref{e:res}. We therefore obtain
\begin{subequations} \label{e:125}
\begin{gather}
\hat{\Phi}_{-,1}(k_1) \sim e^{-i k_1 \xi} \begin{pmatrix}
I_2 \\ \label{e:125a}
0
\end{pmatrix} a_2(k_1), \quad x_2 \ll x \ll x_1,\\[4pt]
\hat{\Phi}_{-,1}(k_2) \sim e^{i k_2 \xi} \begin{pmatrix}
0 \\
I_2
\end{pmatrix} S_2^{-}, \quad x_2 \ll x \ll x_1, \label{e:125b}
\end{gather}
\end{subequations}
where $a_2(k)$ is the scattering coefficient relative to soliton-2, and we have taken into account that for a fundamental soliton the coefficient only depends on the discrete eigenvalue (cf \eref{e:scdataa}). Upon passing through soliton-1, from \eref{e:125a} we find
\begin{gather}
\label{e:126}
\hat{\Phi}_{-,1}(k_1) \sim e^{i k_1 \xi} \begin{pmatrix}
0\\
I_2
\end{pmatrix} S_1^{-} a_2(k_1), \quad x_1 \ll x,
\end{gather}
since the corresponding state is a bound state for $k=k_1$. Now, starting from $x_1 \ll x$ as $t \to -\infty$ and proceeding in a similar way, we find for the eigenfunction $\hat{\Phi}_{+,2}(k_j)$ the following asymptotic behaviors
\begin{gather}
\label{e:127}
\hat{\Phi}_{+,2}(k_j) \sim e^{i k_j \xi} \begin{pmatrix}
0 \\
I_2
\end{pmatrix}, \quad x_1 \ll x,
\end{gather}
for $j=1,2$, and, upon passing through soliton-1, we get
\begin{gather}
\label{e:128}
\hat{\Phi}_{+,2}(k_2) \sim e^{i k_2 \xi} \begin{pmatrix}
0 \\
I_2
\end{pmatrix}c_1(k_2,S_1^{-}), \quad x_2 \ll x \ll x_1,
\end{gather}
where $c_1$ is the scattering coefficient relative to soliton-1, and we have taken into account that according to \eref{e:scdatad} this coefficient depends not just on the discrete eigenvalue $k_1$, but also on the asymptotic polarization of soliton-1. Now recall from \cite{ABBA} that the eigenfunctions $\hat{\Phi}_{-,1}$ and $\hat{\Phi}_{+,2}$ evaluated at a discrete eigenvalue $k_j$ are related as follows
\begin{gather}
\label{e:129}
\hat{\Phi}_{-,1}(k_j) \alpha(k_j) = \hat{\Phi}_{+,2}(k_j) C_j \left( \det a \right)'(k_j),
\end{gather}
where $a(k)$ is the 2-soliton scattering coefficient, and $\alpha(k_j)$ is its cofactor matrix. In general, the explicit expression of $a(k)$ and hence of its cofactor $\alpha(k)$ for a 2-soliton solution can be expected to depend in a nontrivial way
on the scattering coefficients of the individual solitons. If we restrict ourselves to 2 fundamental solitons, however, Eq.~\eref{e:scdataa} shows that the left transmission coefficient is independent of the norming constants/asymptotic polarizations, and moreover it is diagonal. Consequently, one has
$$
a(k)=\prod_{j=1}^{2}\mathrm{diag}\left( \frac{k_j^*}{k_j} \frac{k-k_j}{k-k_j^*}\,, \frac{k_j}{k_j^*}\frac{k+k_j^*}{k+k_j}\right)\,,
$$
and therefore
\begin{gather}
\label{e:91new}
\alpha(k)=\prod_{j=1}^{2}\mathrm{diag}\left(  \frac{k_j}{k_j^*}\frac{k+k_j^*}{k+k_j}\,, \frac{k_j^*}{k_j} \frac{k-k_j}{k-k_j^*}\right)\,.
\end{gather}
Now, comparing relations \eref{e:126} and \eref{e:127}, and using Eq.~\eref{e:129} for $j=1$, we find
\begin{gather}
\label{e:133}
S_1^{-} a_2(k_1) \alpha(k_1) = C_1 \left( \det a \right)'(k_1),
\end{gather}
while comparing relations \eref{e:125b} and \eref{e:128}, and using Eq.~\eref{e:129} for $j=2$, we obtain
\begin{gather}
\label{e:132}
S_2^{-} \alpha(k_2) = c_1(k_2,S_1^{-}) C_2 \left( \det a \right)'(k_2).
\end{gather}
We now follow a similar process for $t \to + \infty$, with the order of the soliton centers being reversed. Therefore, as $t \to + \infty$ and starting from $x \ll x_1$  we have
\begin{subequations}
\begin{gather}
\hat{\Phi}_{-,1}(k_j) \sim e^{-i k_j \xi} \begin{pmatrix}
I_2 \\
0
\end{pmatrix}, \quad x \ll x_1,
\end{gather}
\end{subequations}
for $j=1,2$, and, after passing through soliton-1, we get
\begin{gather}\label{e:135}
\hat{\Phi}_{-,1}(k_1) \sim e^{i k_1 \xi} \begin{pmatrix}
0\\
I_2
\end{pmatrix} S_1^{+}, \quad x_1 \ll x \ll x_2,
\end{gather}
while, upon passing soliton-2, we obtain
\begin{gather}
\label{e:136}
\hat{\Phi}_{-,1}(k_2) \sim e^{i k_2 \xi} \begin{pmatrix}
0 \\
I_2
\end{pmatrix} S_2^{+} a_1(k_2), \quad x_2 \ll x.
\end{gather}
On the other hand, starting from $x_2 \ll x$, we get
\begin{gather}
\label{e:137}
\hat{\Phi}_{+,2}(k_j) \sim e^{i k_j \xi} \begin{pmatrix}
0 \\
I_2
\end{pmatrix}, \quad x_2 \ll x,
\end{gather}
for $j=1,2$, and, after passing through soliton-2, we find
\begin{gather}
\label{e:138}
\hat{\Phi}_{+,2}(k_1) \sim e^{i k_1 \xi} \begin{pmatrix}
0 \\
I_2
\end{pmatrix} c_2(k_1, S_2^{+}), \quad x_1 \ll x \ll x_2.
\end{gather}
Comparing relations \eref{e:135} and \eref{e:138} and using Eq.~\eref{e:129} for $j=1$, we find
\begin{gather}
\label{e:141}
S_1^{+} \alpha(k_1) = c_2(k_1,S_2^{+}) C_1 \left( \det a \right)' (k_1),
\end{gather}
while comparing relations \eref{e:136} and \eref{e:137} and using \eref{e:129} for $j=2$, we obtain
\begin{gather}
\label{e:140}
S_2^{+} a_1(k_2) \alpha(k_2) = C_2 \left( \det a \right)' (k_2).
\end{gather}
One can solve Eqs.~\eref{e:133}, \eref{e:132}, \eref{e:141} and \eref{e:140} with respect to the norming constants $C_j$, and derive the following expressions from the analysis as $t \to -\infty$:
\begin{subequations}\label{e:142}
\begin{align}
C_1 & = \frac{1}{\left( \det a\right)'(k_1)}S_1^{-} a_2(k_1)\alpha(k_1),\label{e:142a}\\
C_2 & = \frac{1}{\left( \det a \right)'(k_2)} c_1^{-1}(k_2,S_1^{-})S_2^{-} \alpha(k_2),\label{e:142b}
\end{align}
\end{subequations}
while the asymptotics as $t \to + \infty$ yields:
\begin{subequations}\label{e:new143}
\begin{align}
C_1 & = \frac{1}{\left( \det a \right)'(k_1)} c_2^{-1}(k_1, S_2^{+})S_1^{+} \alpha(k_1),\label{e:new143a}\\
C_2 & = \frac{1}{\left( \det a \right)'(k_2)} S_2^{+} a_1(k_2)\alpha(k_2).\label{e:new143b}
\end{align}
\end{subequations}
Since the norming constants are time-independent, one must have
\begin{subequations}\label{e:143}
\begin{gather}
S_1^{-} a_2(k_1) \alpha(k_1) = c_2^{-1}(k_1,S_2^{+}) S_1^{+} \alpha(k_1)\,,\label{e:143a}\\
S_2^{+}a_1(k_2) \alpha(k_2) = c_1^{-1}(k_2, S_1^{-}) S_2^{-} \alpha(k_2)\,,\label{e:143b}
\end{gather}
\end{subequations}
or equivalently
\begin{subequations}
\label{e:35}
\begin{gather}
\left( S_1^{-} a_2(k_1) - c_2^{-1}(k_1,S_2^{+}) S_1^{+} \right)\alpha(k_1) = 0, \label{e:35a}\\
\left( S_2^{+}a_1(k_2) - c_1^{-1}(k_2, S_1^{-}) S_2^{-} \right) \alpha(k_2) = 0 \label{e:35b}.
\end{gather}
\end{subequations}
Recall that we are considering the interaction of two fundamental solitons, whose norming constants $C_j$ are rank-one matrices with one column identically zero, i.e., $C_j$ has the form $C_j=(\pmb{\gamma}_j\,, \mathbf{0})$, where $\pmb{\gamma}_j = (\alpha_j\,, \beta_j)^{T}$, with $\alpha_j, \beta_j \in \Complex$. Moreover, the matrices $S_1^{-}, S_2^{+}$ coincide with the norming constants, i.e., $S_1^{-} \equiv C_1$ and $S_2^{+} \equiv C_2$, which implies that ${\bf{s}}_1^{-} \equiv \pmb{\gamma}_1$, because soliton-1 is the ``fast'' soliton, and ${\bf{s}}_2^{+} \equiv \pmb{\gamma}_2$, because soliton-2 is the ``slow'' soliton.\footnote{Recall that the soliton velocities are negative, so ``slow'' and ``fast'' here are meant in absolute value.}  Here, and for the rest of the paper, ${\bf{s}}_2^{\pm}$ denotes the first column vector of the matrix $S_2^{\pm}$, and ${\bf {s}}_1^{\pm}$ denotes the first column vector of the matrix $S_1^{\pm}$. From \eref{e:91new} it follows that
\begin{gather}
\alpha(k_j)=\prod_{j=1,2}\mathrm{diag}\left( \frac{k_j+k_j^*}{2k_j^*}\,, 0\right),
\end{gather}
which means that only the first column of each of the matrix identities \eref{e:35} is specified. In the fundamental soliton case, the second column of the matrices $S_1^{-}, S_2^{+}$ is identically zero, and assuming this remains valid for $S_1^+, S_2^{-}$, then \eref{e:35} reduce to
\begin{gather}
\label{e:94system}
S_1^{+} = c_2(k_1,S_2^{+})S_1^{-} a_2(k_1), \qquad
S_2^{-} = c_1(k_2, S_1^{-})S_2^{+}a_1(k_2).
\end{gather}
Specifically, from \eref{e:scdataa} it follows that
\begin{gather}\label{e:29}
S_2^{-} = \frac{k_1^*}{k_1}\frac{k_2-k_1}{k_2-k_1^*} c_1(k_2, S_1^{-}) S_2^{+}.
\end{gather}
On the other hand, the matrix $c_1(k_2, S_1^{-})$ is a $2 \times 2$ full-rank matrix (cf. \eref{e:scdatac}), and therefore its product with $S_2^{+}$ from the right results in a $2 \times 2$ matrix whose second column is identically zero. This is consistent with the assumption that soliton-2 remains a fundamental soliton after the interaction, and Eq.~\eref{e:29} can be written in vector form as
\begin{gather}\label{e:30}
{\bf{s}}_2^{-} = \frac{k_1^*}{k_1}\frac{k_2-k_1}{k_2-k_1^*} c_1(k_2, S_1^{-}){\bf{s}}_2^{+}.
\end{gather}
Similarly, under the assumption that $S_1^+$ has the same structure as $S_1^-$ (i.e., the one that pertains to a fundamental soliton), \eref{e:35a} is equivalent to the first of \eref{e:94system}, and the latter can be written in vector form as
\begin{gather}\label{e:31}
{\bf{s}}_1^{+} = \frac{k_2^*}{k_2}\frac{k_1-k_2}{k_1-k_2^*} c_2(k_1,S_2^{+}){\bf{s}}_1^{-},
\end{gather}
where the scalar factor is the $(1,1)$ entry of the matrix $a_2(k_1)$.
Considering the explicit expression for the transmission coefficient $c_j$, one can solve \eref{e:30} for ${\bf{s}}_2^{-}$. Let us introduce the unit-norm vectors
\begin{gather}
{\bf p_j^{\pm}} = \frac{ {\bf (s_j^{\pm})^*}}{|| {\bf {s_j^{\pm}}} ||}, \quad j=1,2,
\end{gather}
where ${\bf{p}}_1^{-} =\pmb \gamma_1^*/||\pmb \gamma_1|| $ and ${\bf{p}}_2^{+} = {\pmb \gamma}_2^*/||{\pmb \gamma}_2||$, and recall the expression \eref{e:74b} for $c_j$.
It is convenient to introduce the quantity
\begin{gather}
\chi^2 = \frac{||{\bf{s}}_2^{-}||^2}{||\pmb{\gamma}_2||^2} = \frac{1}{||\pmb{\gamma}_2||^2} \left( {\bf{s}}_2^{-}\right)^\dagger {\bf{s}}_2^{-},
\end{gather}
which according to \eref{e:30} becomes
\begin{gather}
\chi^2 = \left| \frac{k_1 - k_2}{k_1 - k_2^*} \right|^2 \left( {\bf{p}}_2^{+} \right)^T c_1^\dagger(k_2,{\bf{p}}_1^{-}) c_1(k_2,{\bf{p}}_1^{-})\left({\bf{p}}_2^{+}\right)^*.
\end{gather}
Inserting equation \eref{e:74b} into the last one, and observing that
\begin{gather}
\label{e:53}
\left( {\bf{p}}_2^{+}\right)^T \left({\bf{p}}_1^{-}\right)^* \left( {\bf{p}}_1^{-}\right)^T \left({\bf{p}}_2^{+}\right)^* = \left| {\bf{p}}_1^{-} \cdot \left({\bf{p}}_2^{+}\right)^* \right|^2,
\end{gather}
and the fact that ${\bf{p}}_j^{\pm}$ are unit-norm vectors, we obtain that $\chi^2$ is explicitly given by
\begin{gather}
\label{e:38}
\chi^2 = \left| \frac{(k_1 - k_2)(k_1 + k_2^*)}{(k_1 - k_2^*)(k_1 + k_2)}\right|^2 \left\{ 1 + \frac{(k_1^2 - (k_1^*)^2)(k_2^2 - (k_2^*)^2)}{|k_1 - k_2^*|^2|k_1+k_2^*|^2} \left| {\bf{p}}_1^{-} \cdot \left({\bf{p}}_2^{+}\right)^* \right|^2\right\},
\end{gather}
where $\cdot$ denotes the scalar product between vectors, namely the sum of the products of the components. One can also check that the following holds for $\chi^2$,
\begin{gather}
\chi^2 = \frac{||{\bf{s}}_1^{+}||^2}{||\pmb{\gamma}_1||^2}.
\end{gather}
Dividing relation \eref{e:30} by $||{\bf{s}}_2^{-}||$ and relation \eref{e:31} by $||{\bf{s}}_1^{+}||$, and using the definitions of $\chi$ and ${\bf{p}}_j^{\pm}$, we find
\begin{subequations}\label{e:39}
\begin{gather}
{\bf{p}}_2^{-} = \frac{1}{\chi} \left( \frac{k_1}{k_1^*} \frac{k_2^* - k_1^*}{k_2^* - k_1} \right) \left(c_1(k_2,{\bf{p}}_1^{-})\right)^*{\bf{p}}_2^{+}, \label{e:39a}\\
{\bf{p}}_1^{+} = \frac{1}{\chi} \left( \frac{k_2}{k_2^*} \frac{k_1^*- k_2^*}{k_1^*- k_2} \right) \left(c_2(k_1,{\bf{p}}_2^{+})\right)^*{\bf{p}}_1^{-}. \label{e:39b}
\end{gather}
\end{subequations}
Substituting the expression \eref{e:74b} for $j=1$ and $k=k_2$ into Eq.~\eref{e:39a}, and for $j=2$ and $k=k_1$ into Eq.~\eref{e:39b}, and using the relations
\begin{subequations}
\begin{gather}
\left({\bf{p}}_1^{-}\right)^*\left( {\bf{p}}_1^{-} \right)^T \left({\bf{p}}_2^{+}\right)^* = \left({\bf{p}}_1^{-} \cdot \left({\bf{p}}_2^{+}\right)^*\right) \left({\bf{p}}_1^{-}\right)^*,\\
\left({\bf{p}}_2^{+}\right)^*\left( {\bf{p}}_2^{+} \right)^T \left({\bf{p}}_1^{-}\right)^* = \left( {\bf{p}}_2^{+} \cdot \left({\bf{p}}_1^{-}\right)^* \right) \left({\bf{p}}_2^{+}\right)^*
\end{gather}
\end{subequations}
respectively, we obtain the expressions of the polarization vector for soliton-1 and 2 after the interaction
\begin{subequations}\label{e:112n}
\begin{gather}
{\bf{p}}_1^{+} = \frac{1}{\chi} \frac{(k_1^*-k_2^*)(k_1^* + k_2)}{(k_1^*-k_2)(k_1^*+k_2^*)}\left\{ {\bf{p}}_1^{-} + \frac{(k_1^*)^2\left( k_2^2 - (k_2^*)^2 \right)}{(k_2^*)^2 \left( (k_1^*)^2 - k_2^2 \right)} \left(\left( {\bf{p}}_2^{+}\right)^* \cdot {\bf{p}}_1^{-} \right) {\bf{p}}_2^{+} \right\},\\
{\bf{p}}_2^{-} = \frac{1}{\chi} \frac{( k_2^* - k_1^*)(k_2^* + k_1)}{(k_2^* - k_1)(k_2^*+ k_1^*)}\left\{ {\bf{p}}_2^{+} + \frac{(k_2^*)^2}{(k_1^*)^2}\frac{ \left(k_1^2 -(k_1^*)^2 \right)}{\left((k_2^*)^2 - k_1^2 \right)}\left( \left( {\bf{p}}_1^{-} \right)^* \cdot {\bf{p}}_2^{+} \right) {\bf{p}}_1^{-} \right\},
\end{gather}
\end{subequations}
in terms of the initial polarization vectors. Eqs.~\eref{e:112n} are the analog of the well-known Manakov formulas for the coupled NLS equation, and show that interacting fundamental solitons exhibit a polarization
shift, i.e., a redistribution of energy between the components, unless the initial soliton polarizations are either parallel or orthogonal (in which case one can easily verify that the soliton polarization vectors remain the same, upon interaction, up to an overall phase shift).
Notice that Eqs.~\eref{e:112n} are symmetric with respect to interchanging indices $1$ and $2$, and limits $t \to \pm \infty$, i.e., to obtain the expression for ${\bf{p}}_1^{+}$ we simply need to swap the indices $1$ and $2$ and the limits $t \to \pm \infty$ in the corresponding relation for ${\bf{p}}_2^{-}$.

The above construction can be suitably modified to show that self-symmetric solitons always interact in a trivial way (because all the corresponding transmission coefficients are proportional to the identity, cf. \eref{e:ss_transmission}), and to characterize the interaction of a self-symmetric soliton and a fundamental soliton.
However, as we show in Fig.~\ref{fig:3}, fundamental solitons do not remain such in general upon interacting with a fundamental breather. Therefore, Manakov's method is not effective to deal with more complicated soliton interactions, since, on one hand, one has to make an a priori assumption on the nature of the solitons after the interaction and, on the other hand, if all transmission coefficients depend in a nontrivial way on the asymptotic polarizations of the solitons, there is no guarantee one would be able to find a closed form solution. Instead, we will resort to a different approach: we consider the Darboux matrices corresponding to the various types of solitons, and study soliton interactions by combining refactorization problems on generators of certain rational loop groups, and long-time asymptotics of these generators. As we will show in Sec.~6, this allows us to completely characterize all types of soliton interactions, without making any a priori assumption on the nature of the solitons after the interaction.
\begin{figure}[h!]
\centering
\begin{subfigure}{0.25\textwidth}
\centering
\includegraphics[height=1.4in,width=1.4in]{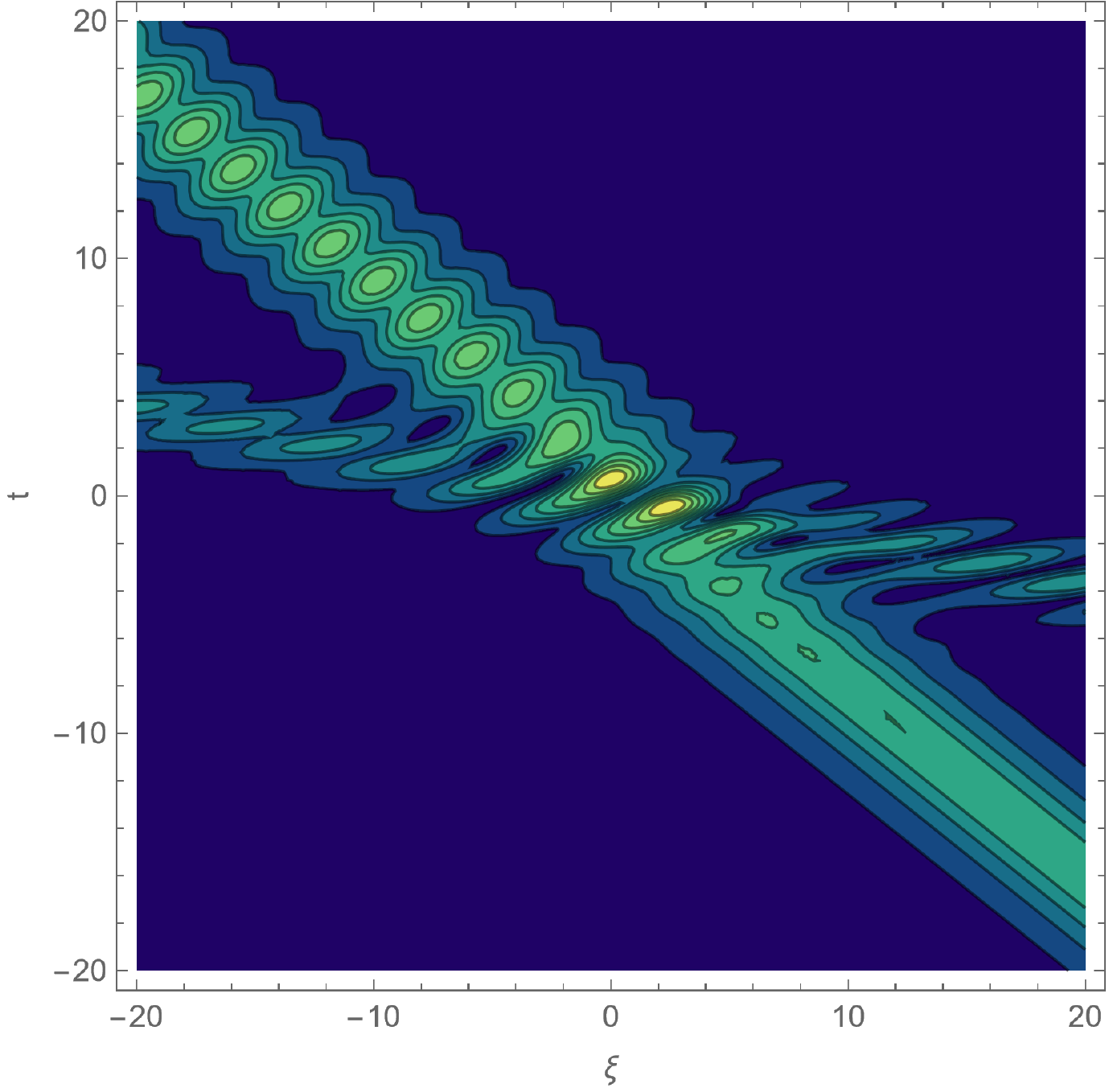}
\caption{} \label{fig:3a}
\end{subfigure}
\hfill
\begin{subfigure}{0.25\textwidth}
\centering
\includegraphics[height=1.4in,width=1.4in]{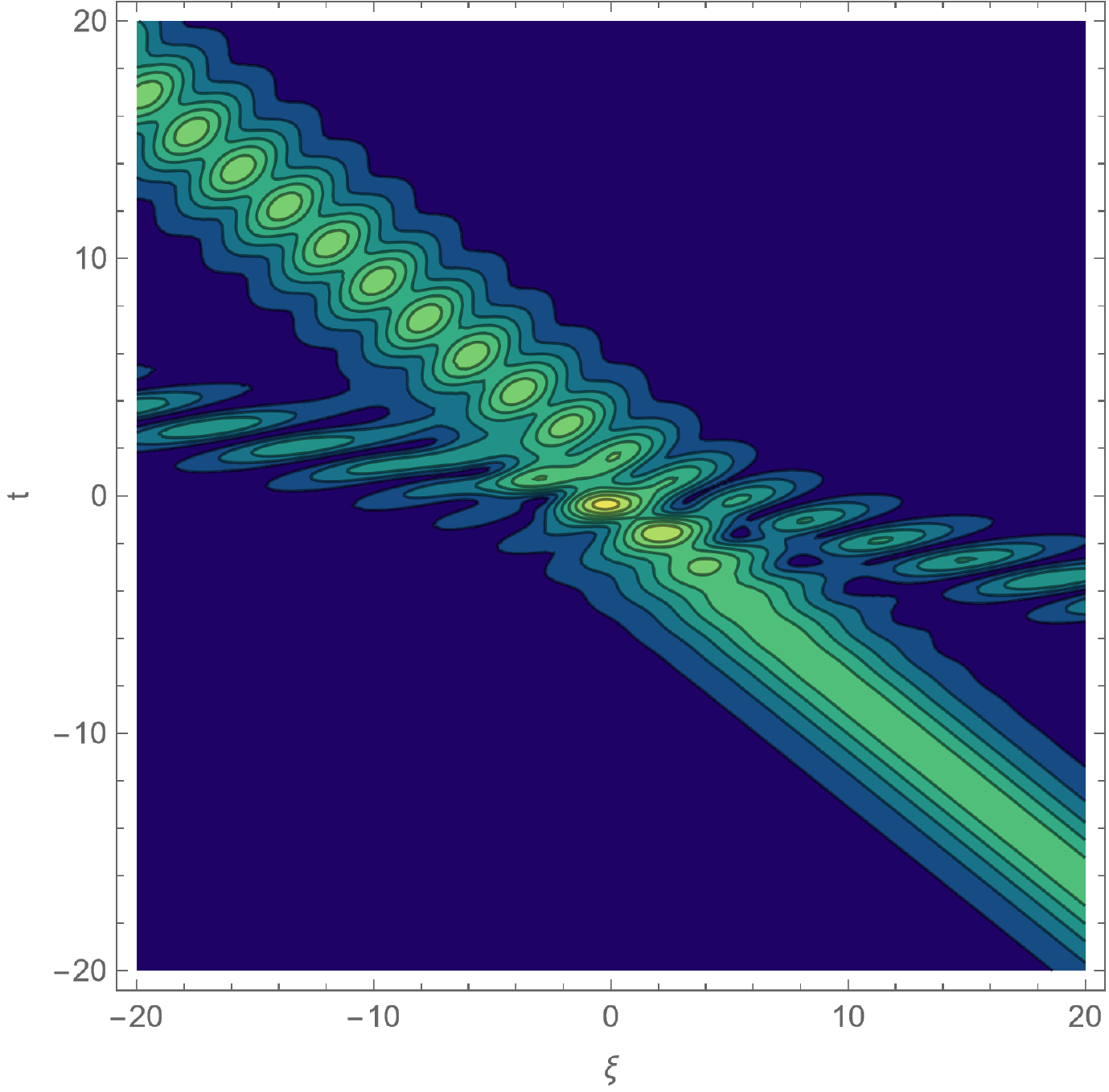}
\caption{} \label{fig:3b}
\end{subfigure}
\hfill
\begin{subfigure}{0.25\textwidth}
\centering
\includegraphics[height=1.4in,width=1.4in]{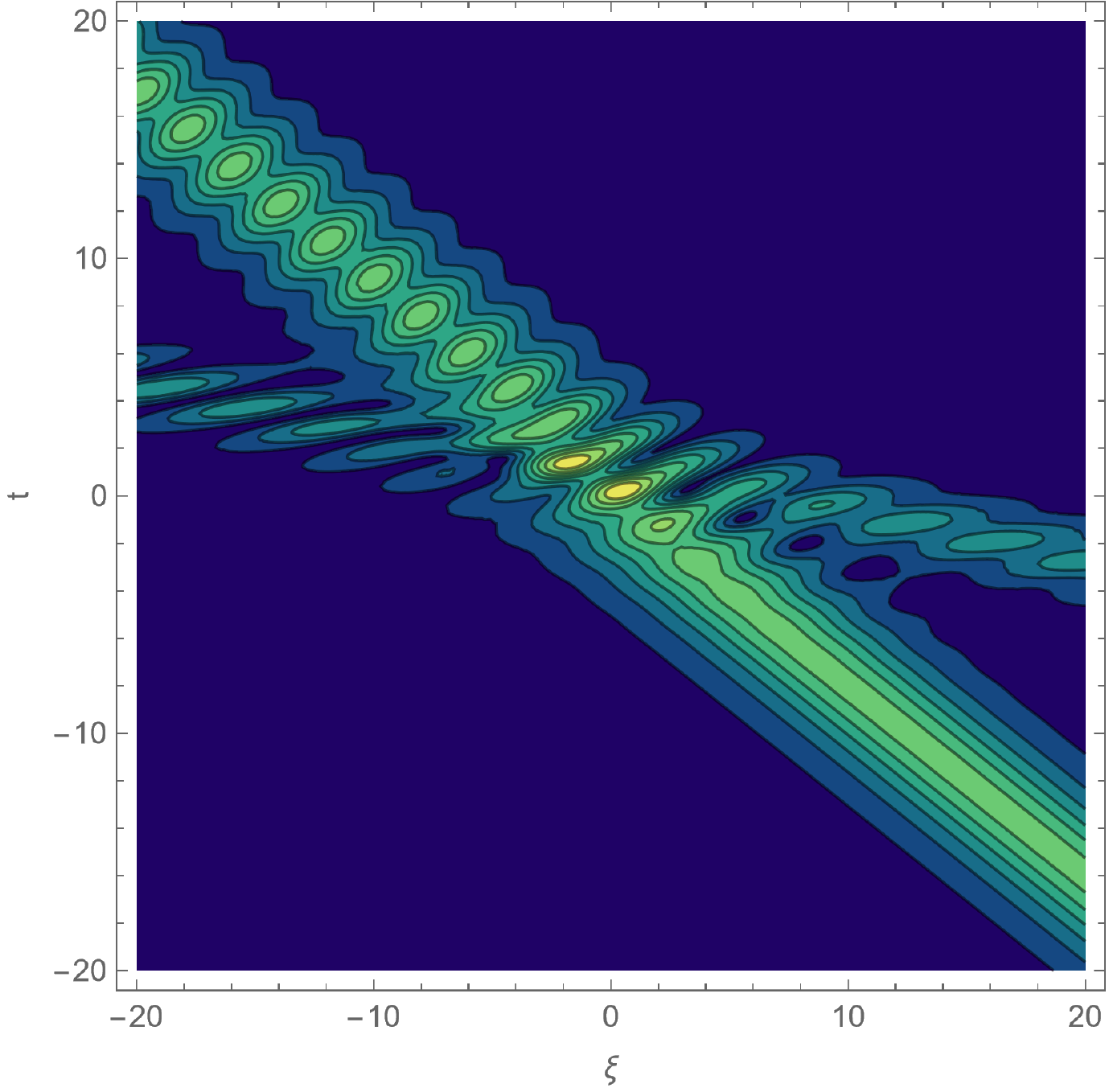}
\caption{} \label{fig:3c}
\end{subfigure}
\caption{Examples of a fundamental soliton interacting with a fundamental breather. From left to right, the three panels show the exact 2-soliton solution for three different examples. In the first example, the initial polarization vectors of the fundamental soliton and of the fundamental breather are chosen to be orthogonal to each other; in the second example they are chosen to be parallel; and in the third example they are neither orthogonal nor parallel. Clearly, the initial fundamental soliton does not maintain its structure upon interacting with the fundamental breather. Here, the soliton parameters are: $k_1 = 1/2+i/4,\,\, k_2=1+i/2,\,\, \alpha_1=e^{i},\,\,\beta_1=1,\,\,\alpha_2 = 1,\,\,\beta_2 =-e^{i}$, for the first example, $k_1=1/2+i/4,\,\,k_2=1+i/2,\,\,\alpha_1=e^{i},\,\,\beta_1=1,\,\,\alpha_2=e^{2i},\,\,\beta_2=e^{i}$, for the second example, and $k_1=1/2+i/4,\,\,k_2=1+i/2,\,\,\alpha_1=e^{i/2},\,\,\beta_1=1,\,\,\alpha_2=e^{i},\,\,\beta_2=4$, for the third example, $\kappa_1=0$ and $\kappa_2=1$ in all the examples, which denotes that soliton-1 is a fundamental soliton and soliton-2 is a fundamental breather before the interaction.}\label{fig:3}
\end{figure}

\section{Yang-Baxter maps and general soliton interactions}

In this section, we study soliton interactions in the ccSPE from the point of view of refactorization problems for B\"acklund-Darboux matrices and Yang-Baxter (YB) maps. We reproduce the results obtained previously with Manakov's method, but also extend more systematically the analysis to all possible two-soliton interactions in the models. This is done by identifying and constructing the possible B\"acklund-Darboux matrices corresponding to fundamental solitons, fundamental breathers, composite breathers and self-symmetric solitons.
With these elementary building blocks at our disposal, the description of soliton interactions is best obtained by combining abstract results on refactorization problems on generators of certain rational loop groups and long time-asymptotics of these generators. The underlying procedure is the so-called dressing method, which goes back to \cite{ZS}. We review the main idea in the next section.

It turns out that the equivalent Lax pair description for the ccSPE introduced in \cite{FL22} is more convenient for the purposes of the beginning of this section. We first review it, and then use it to derive some general results on dressing factors and Yang-Baxter maps. Then, from Section \ref{longtime} onwards, we translate the results back into our original spectral parametrization and go on to derive the Yang-Baxter maps acting on the (polarization) vectors characterizing the various types of solitons of the model.

\subsection{Review of Feng--Ling Lax pair formulation and dressing method}\label{review_dressing}

In \cite{FL22}, the authors consider the following Lax pair obtained from \eref{e:1.3b} by the hodograph transformation \eref{e:4differential} from the coordinates $(x,t)$ to the coordinates $(\xi,t)$ and with the change of spectral parameter $\lambda=1/k$:
\begin{eqnarray}
\label{aux-pb}
\begin{cases}
	\Phi_{\xi}(\xi,t,\lambda)={\cal F}(\xi,t,\lambda)\,\Phi(\xi,t,\lambda)\,,\\
	\Phi_{t}(\xi,t,\lambda)={\cal V}(\xi,t,\lambda)\,\Phi(\xi,t,\lambda)\,,
\end{cases}
\end{eqnarray}
where\footnote{Note that the sign of $\cal{ V}$ has been reversed with respect to \cite{FL22} to account for the sign difference in the way the ccSPE equation is written.}
\begin{eqnarray}
	{\cal F}(\xi,t,\lambda)=-i\frac{\rho(\xi,t)}{\lambda}\Sigma_3-\frac{1}{\lambda}\Sigma_3 \partial_\xi V_{0}(\xi,t)\,,~~
	{\cal V}(\xi,t,\lambda)=\frac{i}{4}\lambda\Sigma_3+\frac{i}{2}V_0(\xi,t)
\end{eqnarray}
with
\begin{equation}
	V_0(\xi,t)=\begin{pmatrix}
		0 & U(\xi,t)\\
		U^\dagger(\xi,t) & 0
	\end{pmatrix}\,.
\end{equation}
The advantage of this formulation is that ${\cal V}(\xi,t,\lambda)$ has the same structure as that of the well-known Lax matrix for a multicomponent nonlinear Schr\"odinger equation with $2\times 2$ matrix potential $U(\xi,t)$. However, the structure of ${\cal F}(\xi,t,\lambda)$ shows that the ccSPE corresponds to a negative flow in that hierarchy. The Lax pair possesses the familiar symmetry
\begin{eqnarray}
	\label{redsym1}
	{\cal F}^\dagger(\xi,t,\lambda^*) = -{\cal F}(\xi,t,\lambda)\,,~~{\cal V}^\dagger(\xi,t,\lambda^*) = -{\cal V}(\xi,t,\lambda),
\end{eqnarray}
but also, and crucially for the ccSPE, the additional symmetry
\begin{eqnarray}
		\label{redsym2}
\Lambda\,{\cal F}^*(\xi,t,-\lambda^*)\,\Lambda^{-1}={\cal F}(\xi,t,\lambda)\,,~~	\Lambda\,{\cal V}^*(\xi,t,-\lambda^*)\,\Lambda^{-1}={\cal V}(\xi,t,\lambda),
\end{eqnarray}
with
\begin{equation}
\Lambda = \diag\left(i \sigma_2\,, i \sigma_2\right), \quad  \sigma_2 = \begin{pmatrix}
0 & -i\\
i & 0
\end{pmatrix}.
\end{equation}
The summary of the ideas and results of the dressing method as introduced in \cite{ZS}, adapted to the present setting, are as follows. We also refer the reader to \cite{Gerd} for a nice review of these ideas focused on multicomponent NLS systems which can be useful for some aspects of the present work. Starting from the potential $V_0=0$, with the goal of constructing pure soliton solutions, one considers two matrix-valued rational functions $G_{\pm}(\xi,t,\lambda)$ which are analytic, respectively, in the upper and lower half-plane for each $\xi,t$, satisfy
$G_{+} G_{-}=I$ for $\lambda \in \mathbb{R}$, $\displaystyle \lim_{\lambda\to\infty}G_\pm(\xi,s;\lambda)=I$, and are degenerate at a finite number of prescribed values of $\lambda$ in their respective domain of analyticity. The structure of their degeneracies is determined by fixing degeneracy spaces corresponding to the norming constants in the IST. An important result of \cite{ZS} is the following. The $(\xi,t)$ dependence of those degeneracy spaces is controlled by a solution of \eref{aux-pb} for $V_0$, and the degeneracy spaces uniquely fix certain projectors entering in the elementary factors contained in $G_\pm$ (matrix Blaschke factors). Below we discuss these aspects in more details.

Since we consider $V_0=0$, the degeneracy space  of $G_+$ at $\lambda_j$ is of the form
\begin{equation}
	\label{kernelcondition}
	\ker G_+(\xi,t,\lambda_j)={\rm span}\Big\lbrace e^{(\frac{i}{4}\lambda_jt-\frac{i}{\lambda_j}\xi)\Sigma_3}W_j \Big\rbrace,
\end{equation}
where $W_j$ is a rectangular matrix whose columns are constant linearly independent vectors fixing the norming constants. As we will show below, the number of these vectors determines the rank of the corresponding projector.
With ${\cal V}(\xi,t,\lambda)=\frac{i}{4}\lambda\Sigma_3+\frac{i}{2}V_0(\xi,t)$ then, in general, a new solution $\widetilde{{\cal V}}(\xi,t,\lambda)=\frac{i}{4}\lambda\Sigma_3+\frac{i}{2}\widetilde{V}_0(\xi,t)$ is obtained by the formula\footnote{It is a standard result that one could equally use $G_-$.}
\begin{equation}
	\label{dressing-formula}
	\widetilde{{\cal V}}(\xi,t,\lambda)=\left(\partial_tG_+(\xi,t,\lambda)+G_+(\xi,t,\lambda) {\cal V}(\xi,t,\lambda)\right)G_+(\xi,t,\lambda)^{-1}.
\end{equation}
The key for explicit formulae is to characterize $G_+$. The {\it generic} result of \cite{ZS} is that $G_+$ is a product of dressing factors (also known as matrix Blaschke factors or B\"acklund-Darboux matrices) of the form
\begin{eqnarray}
	I+\left(\frac{\lambda-\lambda_j^+}{\lambda-\lambda_j^-}-1\right)\Pi_j(\xi,t)\,,
\end{eqnarray}
where $\Pi_j(\xi,t)$ are appropriate projectors.
In the case where $G_+$ is a product of $N$ such factors, $j=1,\dots,N$, we have
\begin{equation}\label{e:n106}
	G_+(\xi,t,\lambda)=I-\frac{1}{\lambda}\sum_{j=1}^N(\lambda_j^+ - \lambda_j^-)\Pi_j(\xi,t)+O\left(\frac{1}{\lambda^2}\right)\,.
\end{equation}
 Inserting Eq.~\eref{e:n106} in \eref{dressing-formula}, recalling that
\begin{gather}
{\cal V}(\xi,t,\lambda)=\frac{i}{4}\lambda\Sigma_3+\frac{i}{2}V_0(\xi,t), \quad \widetilde{{\cal V}}(\xi,t,\lambda)=\frac{i}{4}\lambda\Sigma_3+\frac{i}{2}\widetilde{V}_0(\xi,t),
\end{gather}
we obtain the following relation between the new solution and the old solution by comparing the constant terms in the expansion in $1/\lambda$
\begin{eqnarray}
	\label{dressedsolution}
	\widetilde{V}_0(\xi,t)=V_0(\xi,t)-\sum_{j=1}^N\frac{\lambda_j^+-\lambda_j^-}{2}\left[\Pi_j(\xi,t),\Sigma_3\right]\,.
\end{eqnarray}
In our case, recall that we take $V_0(\xi,t)=0$ as the trivial seed solution.

The exact details on the location of $\lambda^\pm$ and on the properties of the projector $\Pi_j$ depend on the particular model under consideration. We will study this below in conjunction with the two symmetries \eref{redsym1}-\eref{redsym2}. The important message is that the order in which these factors appear in $G_+$ is irrelevant for the final result. This is the important observation that gives rise to Yang-Baxter maps via refactorization problems for the dressing factors. In the sequel, we will follow closely the methodology of \cite{CZ} to combine the refactorization problem with long-time asymptotics in order to derive explicitly Yang-Baxter maps acting on the parameters (norming constants) characterizing the solitons within a multisoliton solution.

\subsection{Refactorization and elementary B\"acklund-Darboux matrices for the ccSPE}

The previous discussion motivates the study of certain refactorization problems for matrix-valued rational functions of the form \begin{equation}\label{e:n111}
	g_{\alpha,\beta,\Pi}(\lambda)=I+\left(\frac{\lambda-\alpha}{\lambda-\beta}-1\right)\Pi,
\end{equation}
 where $\Pi$ is a projector and $\alpha\neq \beta$ are fixed complex numbers. The pointwise dependence on $(\xi,t)$ can be dropped at this stage to formulate the general results. It can be reinstated later on when using these results within the dressing method thanks to the property \eref{kernelcondition} of the degeneracy spaces.

The rigorous study of refactorization problems in conjunction with the dressing method has been undertaken in detail e.g. in \cite{TU}. When forgetting about the pointwise dependence on the independent variables which are present in the dressing methods, such results form part of general structural properties of certain rational loop groups which were studied in detail more recently in \cite{WGW}. The symmetry \eref{redsym1} implies that we can use directly the results of Sec.~6 in \cite{TU}. In particular, symmetry \eref{redsym1} implies that dressing factors must satisfy the symmetry
$g_{\alpha, \beta, \Pi}^\dagger (\xi,t,\lambda^*) = g_{\alpha, \beta, \Pi}^{-1}(\xi,t,\lambda)$.
Eq.~\eref{e:n111} then implies the constraints $\beta=\alpha^*$ and $\Pi^\dagger=\Pi$, and therefore \eref{e:n111} becomes
\begin{equation}\label{e:128n}
	g_{\alpha,\Pi}(\lambda)=I-\frac{\alpha-\alpha^*	}{\lambda-\alpha^*}\Pi\,,~~\alpha \in \Complex\setminus\Real\,,~~\Pi^2=\Pi^\dagger=\Pi\,.
\end{equation}
An elementary factor $g_{\alpha,\Pi}$ is characterized by a complex number with nonzero imaginary part and by an Hermitian projector $\Pi$ which we call its data. We then have the following result on the refactorization of two such factors (Sec.~6 of \cite{TU}).
\begin{theorem}\label{Threfac}
	Let $\Pi_1$, $\Pi_2$ be two orthogonal projectors and $\alpha_1,\alpha_2 \in \mathbb{C}\setminus\Real$ be such that $\alpha_1\neq \alpha_2$ and $\alpha_1\neq\alpha_2^*$. The unique solution of the refactorization problem
	\begin{gather}\label{e:124refac}
		\left(I-\frac{\alpha_2-\alpha_2^*}{\lambda-\alpha_2^*}\Pi_2\right)
\left(I-\frac{\alpha_1-\alpha_1^*}{\lambda-\alpha_1^*}\Pi_1\right)
		=
		\left(I-\frac{\alpha_1-\alpha_1^*}{\lambda-\alpha_1^*}P_1\right)
\left(I-\frac{\alpha_2-\alpha_2^*}{\lambda-\alpha_2^*}P_2\right)
	\end{gather}
	is given by
	\begin{gather}\label{e:10}
		P_j=\phi^{-1}\Pi_j\phi\,,~~j=1,2\,,
	\end{gather}
	where
	\begin{gather}\label{e:11}
	\phi=(\alpha_1^*-\alpha_2)I+(\alpha_1-\alpha_1^*)\Pi_1+(\alpha_2-\alpha_2^*)\Pi_2\,.
	\end{gather}
	Moreover $P_j$, $j=1,2$ are also orthogonal projectors.	
\end{theorem}
Note that the rank of the projectors is preserved under refactorization. The most difficult part in this result is to establish that the matrix $\phi$ is always invertible. This result shows that the refactorization problem induces a (parametric) map on pairs of orthogonal projectors\footnote{The attentive reader will notice that we have included the transposition $(P_1,P_2)\mapsto (P_2,P_1)$ in the definition of $R$ so that it satisfies the Yang-Baxter equation, see corollary below. Without this, $R$ would satisfy the so-called braided version of the Yang-Baxter equation. }
\begin{equation}
	R_{12}(\alpha_1,\alpha_2):~~	(\Pi_1,\Pi_2)\mapsto (P_2,P_1).
\end{equation}
The next result implies that this map is a (parametric) Yang-Baxter (YB) map.
\begin{lemma}
	Let $\alpha_j \in \mathbb{C}\setminus\Real$, $j=1,2,3$ be as in Theorem~\ref{Threfac}, i.e., $\alpha_i\neq \alpha_j$ and $\alpha_i\neq\alpha_j^*$, $i\neq j$.  If
	\begin{equation}
		g_{\alpha_3,\Pi_3}(\lambda)\,g_{\alpha_2,\Pi_2}(\lambda)\,g_{\alpha_1,\Pi_1}(\lambda)=g_{\alpha_3,P_3}(\lambda)\,g_{\alpha_2,P_2}(\lambda)\,g_{\alpha_1,P_1}(\lambda)\,,
	\end{equation}
	Then $\Pi_j=P_j$, $j=1,2,3$.	
\end{lemma}
\begin{proof}
	The proof is the same as in Prop. 3.10 of \cite{Li}.
\end{proof}

As a consequence, we get the following result which goes back at least to \cite{KP}.
\begin{corollary}
	The map $R_{12}(\alpha_1,\alpha_2)$ satisfies the Yang-Baxter equation
	\begin{equation}
\label{e:YB}
		R_{12}(\alpha_1,\alpha_2)\,R_{13}(\alpha_1,\alpha_3)\,R_{23}(\alpha_2,\alpha_3)=R_{23}(\alpha_2,\alpha_3)\,R_{13}(\alpha_1,\alpha_3)\,R_{12}(\alpha_1,\alpha_2).
	\end{equation}
\end{corollary}
The meaning of this equation and our notations are as follows. For instance, on the triplet $(\Pi_1,\Pi_2,\Pi_3)$, the LHS of \eref{e:YB} acts as
\begin{eqnarray*}
	&&R_{12}(\alpha_1,\alpha_2)\,R_{13}(\alpha_1,\alpha_3)\,R_{23}(\alpha_2,\alpha_3)\,(\Pi_1,\Pi_2,\Pi_3)\\
	&=&	R_{12}(\alpha_1,\alpha_2)\,R_{13}(\alpha_1,\alpha_3)\,(\Pi_1,\Pi_3^{(i)},\Pi_2^{(i)})\\
	&=&R_{12}(\alpha_1,\alpha_2)\,(\Pi_3^{(ii)},\Pi_1^{(i)},\Pi_2^{(i)})\\
	&=&(\Pi_3^{(ii)},\Pi_2^{(ii)},\Pi_1^{(ii)})\equiv(P_3,P_2,P_1)\,.
\end{eqnarray*}
The YB equation means that the RHS of \eref{e:YB} produces exactly the same final result, although the YB maps act in a different order and the intermediate projectors are not equal:
\begin{eqnarray*}
	&&R_{23}(\alpha_2,\alpha_3)\,R_{13}(\alpha_1,\alpha_3)\,R_{12}(\alpha_1,\alpha_2)\,(\Pi_1,\Pi_2,\Pi_3)\\
	&=&	R_{23}(\alpha_2,\alpha_3)\,R_{13}(\alpha_1,\alpha_3)\,(\Pi_2^{(1)},\Pi_1^{(1)},\Pi_3)\\
	&=&R_{23}(\alpha_2,\alpha_3)\,(\Pi_2^{(1)},\Pi_3^{(1)},\Pi_1^{(2)})\\
	&=&(\Pi_3^{(2)},\Pi_2^{(2)},\Pi_1^{(2)})=(P_3,P_2,P_1)\,.
\end{eqnarray*}
In our case, the model also enjoys the additional symmetry \eref{redsym2}, which has nontrivial consequences on the structure of the elementary factors that can generate a rational loop group element $G(\lambda)$. For such a loop group element, the symmetry reads
\begin{gather}\label{Lambdasym}
	G^*(-\lambda^*)=\Lambda\,G(\lambda)\,\Lambda^{-1}\,.
\end{gather}
The following two results show that the symmetry \eref{Lambdasym} requires to consider $g_{\alpha,\Pi}(\lambda)$ with specific data but also a product of two such elementary factors with specific data in order to exhaust the possible generators of the known types of solitons in the ccSPE. Before we proceed, note that because the ccSPE is described by $4\times 4$ Lax pair, the projectors in the elementary factors acts on $\Complex^4$. A nontrivial projector is therefore of rank $1,2,3$. Due to the relation
\begin{equation}
	g_{\alpha,\Pi}(\lambda^*)^\dagger=	g_{\alpha,\Pi}(\lambda)^{-1}=\frac{\lambda-\alpha^*}{\lambda-\alpha}g_{\alpha,\Pi^\perp}(\lambda)\,,
\end{equation}
where $\Pi^\perp=I-\Pi$ is the orthogonal complement projector of $\Pi$, we see that we can work equally with $\Pi$ or $\Pi^\perp$. Obviously, when $\Pi$ has rank $3$,  $\Pi^\perp$ has rank $1$ and vice-versa, which shows that the rank $1$ and rank 3 cases are equivalent. Hence, in the following we will only consider the rank $1$ and $2$ cases.
\begin{lemma}
	\label{selfsymmetric:lemma}
	Let $\alpha \in \mathbb{C} \setminus \mathbb{R}$. The elementary factor $g_{\alpha,\Pi}(\lambda)$ satisfies the first symmetry in \eref{Lambdasym} if and only if $\alpha=ia$, $a\in \mathbb{R}$, and
	\begin{eqnarray}
		\label{formPi}
		\Pi=\Pi_1+\Lambda \,\Pi_1^*\,\Lambda^{-1}\,,
	\end{eqnarray}
	where $\Pi_1$ is a rank-$1$ projector. In that case, $\Pi$ is a rank-$2$ projector.
\end{lemma}
\begin{proof} If $\alpha$ and $\Pi$ are as stated, it is straightforward to check that $g_{\alpha,\Pi}(\lambda)$ satisfies ${g}^*_{\alpha,\Pi}(-\lambda^*) = \Lambda g_{\alpha,\Pi}(\lambda) \Lambda^{-1}$. Conversely, suppose that ${g}^*_{\alpha,\Pi}(-\lambda^*) = \Lambda g_{\alpha,\Pi}(\lambda) \Lambda^{-1}$. This yields $(\lambda-\alpha^*)\Pi^*=(\lambda+\alpha)\Lambda\,\Pi\,\Lambda^{-1}$, which in turn implies $\alpha=-\alpha^*$ and $\Pi^*=\Lambda\,\Pi\,\Lambda^{-1}$. Since $\Pi$ is a Hermitian projector, we can write $\Pi=U\,U^\dagger$, where $U$ is the matrix whose columns vectors form an orthonormal basis of $\Im\Pi$. In particular $U^\dagger U=I_r$, where $r$ is the rank of $\Pi$. We have
	\begin{gather}\label{pistar}
		\Pi^*=(\Pi^*)^2=\Lambda\,\Pi\,\Lambda^{-1}\,\Pi^*=\Lambda\,U\,U^\dagger\,\Lambda^{-1}\,U^*\,U^T.
	\end{gather}
	If $r=1$, then $U=(u_1)$ is a column vector and we have $U^\dagger\,\Lambda^{-1}\,U^*=0$. This is due to the form of $\Lambda$ which implies that for any vector $v\in\Complex^4$
	\begin{equation}
		\label{e:orthog}
		v^\dagger\,\Lambda^{-1}\,v^*=0=	v^T\,\Lambda^{-1}\,v\,.
	\end{equation}
	Hence $\Pi=0$, which is a contradiction so this case is not possible.
Suppose now $r=2$ with $U=(u_1,u_2)$. We have
		\begin{gather}\label{expressionPi}
			\Pi=\Lambda\,U^*\,M^*\,U^\dagger, \quad M\equiv U^\dagger\,\Lambda^{-1}\,U^*=i\alpha_{12}\sigma_2, \quad \alpha_{12} = u_1^\dagger \Lambda^{-1} u_2^*.
	\end{gather}	
	Inserting into $\Pi U=U$, yields
	$$u_2=\alpha_{12}^*\Lambda u_1^*\,,~~u_1=-\alpha_{12}^*\Lambda u_2^*\,.$$
	For consistency we must have $|\alpha_{12}|=1$.
	Using the relation to eliminate $u_2$ in the expression \eref{expressionPi} for $\Pi$ we find	
	\begin{gather*}
		\Pi=-\alpha_{12}^*\Lambda u_2^* u_1^\dagger+\alpha_{12}^*\Lambda u_1^* u_2^\dagger=u_1u_1^\dagger+\Lambda u_1^* u_1^T\Lambda^{-1}
		=\Pi_1+\Lambda\, \Pi_1^* \,\Lambda^{-1}
	\end{gather*}
	where $\Pi_1=u_1u_1^\dagger$ is a rank-$1$ projector. This completes the proof.	
\end{proof}

\medskip

This shows that a single elementary factor always corresponds to a so-called self-symmetric (purely imaginary) eigenvalue, and can only be built from a special rank-$2$ projector. In particular, it is not possible to create solitons for non self-symmetric zeros with only one elementary factors because of the symmetry \eref{Lambdasym}. For this, we need a product of two elementary factors, which we study below.
\begin{lemma}\label{lemmasym}
	Let $\alpha_1,\alpha_2 \in \mathbb{C}\setminus \mathbb{R}$ be such that $\alpha_1\neq\alpha_2$ and $\alpha_1\neq \alpha_2^*$. If
\begin{gather}\label{e:matrixT}
T(\lambda)=g_{\alpha_1,\Pi_1}(\lambda)g_{\alpha_2,\Pi_2}(\lambda),
\end{gather}
satisfies \eref{Lambdasym}, where $g_{\alpha_j,\Pi_j}(\lambda)$ is as in \eref{e:128n}, then either :
\begin{enumerate}
		\item $\alpha_j=-\alpha_j^*$, $j=1,2$, and each elementary factor is of the self-symmetric type classified in Lemma \ref{selfsymmetric:lemma}, or;		
		\item $\alpha_2=-\alpha_1^*$,  and $\Pi_2=\Lambda\,\Pi_1^*\,\Lambda^{-1}$.
		
		
	\end{enumerate}
\end{lemma}
The proof of this lemma is given in Appendix~C.

Several comments are in order.
	\begin{remark}
		For case 2. in Lemma \ref{lemmasym}, it is interesting to note that we can obtain the case of self-symmetric solitons of Lemma~\ref{selfsymmetric:lemma} when the projector $\Pi_1$ is of rank $1$ and by allowing $\alpha_1$ to be purely imaginary, even though this is not strictly allowed under the assumptions on $\alpha_1$ and $\alpha_2$ in Lemma~\ref{lemmasym}. Indeed, a direct calculation shows that $T(\lambda)$ then exactly reduces to $g_{\alpha_1,\Pi}$ with $\Pi$ as in \eref{formPi}, of rank $2$. Thus,  we can use the factor $T(\lambda)$ described in Lemma \ref{lemmasym} and recover the self-symmetric case as a limiting case if desired, instead of treating it as a separate case. This means that $T(\lambda)$ can be used to deal with all possible cases of solitons in the model, allowing for the various possibilities for the rank of $\Pi_1$ and for the nature of $\alpha_1$ (purely imaginary or not). From now on, we will write the ``elementary'' factor in Lemma~\ref{lemmasym} as
		\begin{eqnarray}
			T_{\alpha,\Pi}(\lambda)\equiv g_{\alpha,\Pi}(\lambda)g_{-\alpha^*,\widetilde{\Pi}}(\lambda)
		\end{eqnarray}
	with the understanding that $\widetilde{\Pi}$ is determined by $\Pi$ according to $\widetilde{\Pi}=\Lambda\,\Pi^*\,\Lambda^{-1}$. As before, $\alpha$ and $\Pi$ will be called the data for $T_{\alpha,\Pi}(\lambda)$.
	\end{remark}

\begin{remark}
In the rank-$1$ case, the factor $T_{\alpha,\Pi}(\lambda)$ reproduces the Darboux matrix of Theorem 1 in \cite{FL22}. This can be checked by a direct calculation, observing that in the rank-$1$ case, $\Pi$ and $\widetilde{\Pi}$ satisfy $\widetilde{\Pi}\,\Pi=0=\Pi\,\widetilde{\Pi}$, so that
	\begin{eqnarray}
				T_{\alpha,\Pi}(\lambda)&=& g_{\alpha,\Pi}(\lambda)g_{-\alpha^*,\widetilde{\Pi}}(\lambda)\nonumber\\
				&=&g_{-\alpha^*,\widetilde{\Pi}}(\lambda)g_{\alpha,\Pi}(\lambda)\nonumber\\
				&=&I-\left(\frac{\alpha-\alpha^*	}{\lambda-\alpha^*}-1\right)\Pi-\left(\frac{\alpha-\alpha^*}{\lambda-\alpha^*}-1\right)\Lambda\,\Pi^*\,\Lambda^{-1}.
	\end{eqnarray}
This factor produces generically the fundamental breather solution, as will be clear below. As explained above, taking $\alpha$ to be pure imaginary produces the special case of a self-symmetric eigenvalue solution. We will see below that choosing the degeneracy space of $\Pi$ to be generated by a special vector yields the case of a fundamental soliton as a special case.
\end{remark}

\begin{remark}\label{Ridentity}
Property \eref{commute1}, which is easily checked for any Hermitian projector $\Pi_1$, ensures that the two factors in $T_{\alpha,\Pi}(\lambda)$ commute
		\begin{eqnarray}
		T_{\alpha,\Pi}(\lambda)= g_{\alpha,\Pi}(\lambda)g_{-\alpha^*,\widetilde{\Pi}}(\lambda)=g_{-\alpha^*,\widetilde{\Pi}}(\lambda)g_{\alpha,\Pi}(\lambda)\,.
	\end{eqnarray}
As a consequence, the YB map $R(\alpha_1,\alpha_2)$ reduces to the identity map when $\alpha_2=-\alpha_1^*$ and when it acts on pairs on projectors $(\Pi_1,\Pi_2)$ such that $\Pi_2=\Lambda\,\Pi^*\,\Lambda^{-1}$.
\end{remark}

In view of our results and the previous remarks, for our model the refactorization problem of Theorem~\ref{Threfac} should now be posed in terms of the binary elementary factor $T_{\alpha,\Pi}(\lambda)$:
	\begin{equation}\label{refacT}
		T_{\alpha_1,\Pi_1}(\lambda)\,T_{\alpha_2,\Pi_2}(\lambda)=T_{\alpha_2,P_2}(\lambda)\,T_{\alpha_1,P_1}(\lambda)\,.
	\end{equation}
All the possible subcases arising from specializing the data of $T_{\alpha,\Pi}(\lambda)$ will describe all the possible interactions between the various types of solitons in the model. This is studied in detail in the next section. For now, let us first note that the refactorization problem \eref{refacT} is well-defined in the sense that, if the LHS is a product of binary elementary factors of type $T_{\alpha,\Pi}(\lambda)$, then the RHS is also such a product. In other words, when considering $T_{\alpha_1,\Pi_1}(\lambda)\,T_{\alpha_2,\Pi_2}(\lambda)$ as a product of 4 elementary factors of type $g_{\alpha,\Pi}(\lambda)$
\begin{equation}
	T_{\alpha_1,\Pi_1}(\lambda)\,T_{\alpha_2,\Pi_2}(\lambda)= g_{\alpha_1,\Pi_1}(\lambda)\,g_{-\alpha_1^*,\widetilde{\Pi}_1}(\lambda)\,g_{\alpha_2,\Pi_2}(\lambda)\,g_{-\alpha_2^*,\widetilde{\Pi_2}}(\lambda),
\end{equation}
and refactorizing into
\begin{equation}
	g_{\alpha_2,P_2}(\lambda)\,g_{-\alpha_2^*,P^{'}_2}(\lambda)\, g_{\alpha_1,P_1}(\lambda)\,g_{-\alpha_1^*,P^{'}_1}(\lambda),
\end{equation}
we have necessarily $P^{'}_j=\Lambda\,P_j^*\Lambda^{-1}=\widetilde{P}_j$, $j=1,2$. Thus, we can consistently interpret the result as the product $T_{\alpha_2,P_2}(k)\,T_{\alpha_1,P_1}(k)$. To see this, one uses again the symmetry \eref{Lambdasym}, a comparison of the poles on each side of the equality, Liouville theorem and Lemma~\ref{lemmasym} to obtain the desired relation $P^{'}_j=\Lambda\,P_j^*\,\Lambda^{-1}$.

The refactorization problem \eref{refacT}  yields a new parametric YB map which can be seen as composite of the basic map $R_{12}(\alpha_1,\alpha_2)$, restricted to special data. Notice that each basic map $R_{12}(\alpha_1,\alpha_2)$ accounts for a single refactorization, and therefore the refactorization problem \eref{refacT} amounts to considering the following composition of four $R$-maps
\begin{eqnarray}
	\label{contraints}
	R_{14}(\alpha_1,\alpha_4)\,R_{24}(\alpha_2,\alpha_4)\,R_{13}(\alpha_1,\alpha_3)\,R_{23}(\alpha_2,\alpha_3)
\end{eqnarray}
acting on the quadruplet of projectors $(\Pi_1,\Pi_2,\Pi_3,\Pi_4)$ with the restrictions that
$$\Pi_2=\widetilde{\Pi}_1\,,~~\Pi_4=\widetilde{\Pi}_3\,,~~\alpha_2=-\alpha_1^*\,,~~\alpha_4=\alpha_3^*\,.$$
With this in mind, it is convenient and suggestive to denote the first space as $1$, the second space as $\tilde{1}$, the third space as $2$ and the fourth space as $\tilde{2}$. Thus, dropping the parameters for conciseness, we write
\begin{eqnarray}
	\label{tildenotation}
	R_{14}\,R_{24}\,R_{13}\,R_{23}\,(\Pi_1,\Pi_2,\Pi_3,\Pi_4)\equiv 	R_{1\tilde{2}}\,R_{\tilde{1}\tilde{2}}\,R_{12}\,R_{\tilde{1}2}\,(\Pi_1,\widetilde{\Pi}_1,\Pi_2,\widetilde{\Pi}_2).
\end{eqnarray}
Due to the above observation on the consistency of the refactorization for the binary elementary factors, the result of this operation is given by
\begin{eqnarray}\label{e:new134}
	R_{1\tilde{2}}\,R_{\tilde{1}\tilde{2}}\,R_{12}\,R_{\tilde{1}2}\,(\Pi_1,\widetilde{\Pi}_1,\Pi_2,\widetilde{\Pi}_2)=(P_2,\widetilde{P}_2,P_1,\widetilde{P}_1)\,.
\end{eqnarray}
As for $T_{\alpha,\Pi}(\lambda)$, $\widetilde{\Pi}_j$, $j=1,2$ (resp., $\widetilde{P}_j$, $j=1,2$) are completely determined by $\Pi_j$, $j=1,2$ (resp., $P_j$, $j=1,2$). Thus they are redundant and we can project the map onto the first and third entries in the quadruplet to obtain a map acting on $(\Pi_1,\Pi_2)$ to produce $(P_2,P_1)$:
	\begin{equation}
		S_{{\pmb 1}{\pmb 2}}(\alpha_1,\alpha_2): (\Pi_1,\Pi_2)\mapsto (P_2,P_1)\,.
	\end{equation}
Note that the indices are in bold to remember the difference between $S$ and $R$. It is a direct consequence of the YB equation for $R$ that this map also satisfies the YB  equation. Indeed, in view of the definition of $S_{{\pmb 1}{\pmb 2}}(\alpha_1,\alpha_2)$ from \eref{e:new134}, to show that the following holds
	\begin{equation}
	S_{{\pmb 1}{\pmb 2}}(\alpha_1,\alpha_2)\,S_{{\pmb 1}{\pmb 3}}(\alpha_1,\alpha_3)\,S_{{\pmb 2}{\pmb 3}}(\alpha_2,\alpha_3)=S_{{\pmb 2}{\pmb 3}}(\alpha_2,\alpha_3)\,S_{{\pmb 1}{\pmb 3}}(\alpha_1,\alpha_3)\,S_{{\pmb 1}{\pmb 2}}(\alpha_1,\alpha_2),
\end{equation}
it is sufficient to show that
\begin{eqnarray}
	\label{YBE_4R}
	&&R_{1\tilde{2}}\,R_{\tilde{1}\tilde{2}}\,R_{12}\,R_{\tilde{1}2}\,R_{1\tilde{3}}\,R_{\tilde{1}\tilde{3}}\,R_{13}\,R_{\tilde{1}3}\,
	R_{2\tilde{3}}\,R_{\tilde{2}\tilde{3}}\,R_{23}\,R_{\tilde{2}3}\notag\\
	&&=R_{2\tilde{3}}\,R_{\tilde{2}\tilde{3}}\,R_{23}\,R_{\tilde{2}3}\,R_{1\tilde{3}}\,R_{\tilde{1}\tilde{3}}\,R_{13}\,R_{\tilde{1}3}\,R_{1\tilde{2}}\,R_{\tilde{1}\tilde{2}}\,R_{12}\,R_{\tilde{1}2}\,,
\end{eqnarray}
holds. Now, \eref{YBE_4R} can be seen to hold using the YB equation for $R$ repeatedly for various combinations  of indices $1,\tilde{1}$, $2,\tilde{2}$, $3,\tilde{3}$ and remembering Remark~\ref{Ridentity} which tells us that $R_{j\tilde{j}}$, $j=1,2,3$ is the identity map.

The map $S$ is the map underlying all the maps describing the interactions of solitons derived in the next section. Its YB property ensures that these various maps consistently extend to multisoliton solutions with more than 2 solitons. This is similar to what is explained in detail in \cite{CZ}, with the important difference that in the latter work one could only have one type of solitons (vector NLS solitons of rank $1$) whereas in the present work, we have various types of solitons corresponding to the various allowed cases for $T$: rank $1$ or $2$, self-symmetric zeros or not. Note that the rank-$1$ case degenerates further into fundamental solitons and fundamental breathers, but we find that the rank-$2$ case also degenerates into various subcases described below. Indeed, as we will explain in Sec.~\ref{e:sub6.3.2}, to construct a full-rank soliton solution, one needs to consider four $2 \times 1$ vectors, and whether two out of the four vectors are linearly independent or not produces various subcases that need to be discussed separately. To our knowledge, this is the first time that YB maps are presented in a multicomponent integrable system with such a variety of soliton interactions. As we will see below, this gives rise to different explicit forms of the map $S$, depending on the type of solitons which interact. Our results could perhaps give a realization in the context of soliton interactions of the idea of entwining YB maps as discussed for instance in \cite{KP2,Kas,KRP}. However, this is beyond the scope of the present work.

\subsection{Long-time asymptotic analysis and map on soliton data}\label{longtime}

To derive the induced YB map on the quantities controlling the solitons in a multisoliton solution and interpret the soliton collisions from this point of view, we need to relate the previous abstract results to objects appearing in the IST. This is the essence of the dressing method as introduced in \cite{ZS}. In view of our discussion in Sec.~\ref{review_dressing} and in particular formula \eref{kernelcondition}, all we need to do is to fix the matrices $W_j$ corresponding to the eigenvalue $\lambda_j$. For a solution corresponding to a projector of rank $1$, this matrix is just a column vector ($4\times 1$ matrix), and the rank-$2$ case corresponds to a $4\times 2$ matrix (see Sec.~6.3.2 for details).

\subsubsection{Rank-1 case: fundamental breathers and fundamental solitons}

We first review the construction for $N=1$ to derive the rank-$1$ one-soliton solution (either fundamental soliton or fundamental breather), and fix notations of the relevant parameters. Then, we will study the case $N=2$ in view of the refactorization results of the previous section, and deduce the map induced on the parameters characterizing each soliton.

\paragraph{Case $N=1$.} In this case, fix $\lambda_1=a_1+ib_1$ with $b_1<0$ (recall that $\lambda_1=1/k_1 \equiv k_1^*/|k_1|^2$, and $k_1\in\mathbb{C}^{+}$), and we have
\begin{eqnarray}
G_+(\xi,t;\lambda)&=&T_{\lambda_1,\Pi_1}(\lambda)=\left(I-\frac{\lambda_1-\lambda_1^*}{\lambda-\lambda_1^*}\Pi_1\right)\left(I-\frac{\lambda_1-\lambda_1^*}{\lambda+\lambda_1}\Lambda\Pi_1^*\Lambda^{-1}\right)\nonumber\\
&=&I-\frac{\lambda_1-\lambda_1^*}{\lambda-\lambda_1^*}\Pi_1-\frac{\lambda_1-\lambda_1^*}{\lambda+\lambda_1}\Lambda\Pi_1^*\Lambda^{-1},
\end{eqnarray}
where $\Pi_1$, which now depends on $(\xi,t)$ even though we did not show it explicitly for conciseness, is determined by the condition $\ker G_+(\xi,t;\lambda_1)={\rm span}\Big\lbrace e^{(\frac{i}{4}\lambda_1 t-\frac{i}{\lambda_1}\xi)\Sigma_3}\Phi_1 \Big\rbrace$. Introducing the notations
\begin{gather}
\Phi_1(\xi,t)=e^{(\frac{i}{4}\lambda_1 t-\frac{i}{\lambda_1}\xi)\Sigma_3}M_1=e^{(\frac{i}{4}\lambda_1 t-\frac{i}{\lambda_1}\xi)\Sigma_3}\begin{pmatrix}
	{\pmb \delta}_1\\
	{\pmb \gamma}_1
\end{pmatrix}\,~~{\pmb \delta}_1,{\pmb \gamma}_1\in \mathbb{C}^2\,,
\end{gather}
and recalling the definition of $\Pi_1(\xi,t) = \frac{\Phi_1(\xi,t) \Phi_1^\dagger(\xi,t)}{\Phi_1^\dagger(\xi,t) \Phi_1(\xi,t)}$, this yields
\begin{equation}
	\label{formPi1}
	\Pi_1(\xi,t)=\frac{1}{D(\xi,t)}\begin{pmatrix}
		e^{\frac{b_1}{2\mathrm{v}_1}(\xi-\mathrm{v}_1t)}{\pmb \delta}_1{\pmb \delta}_1^\dagger & e^{\frac{ia_1}{2\mathrm{v}_1}(\xi+\mathrm{v}_1t)}{\pmb \delta}_1{\pmb \gamma}_1^\dagger\\
		e^{-\frac{ia_1}{2\mathrm{v}_1}(\xi+\mathrm{v}_1t)}{\pmb \gamma}_1{\pmb \delta}_1^\dagger		& e^{-\frac{b_1}{2\mathrm{v}_1}(\xi-\mathrm{v}_1t)}{\pmb \gamma}_1{\pmb \gamma}_1^\dagger
	\end{pmatrix},
\end{equation}
where $\mathrm{v}_1=-\frac{a_1^2+b_1^2}{4}=-\frac{|\lambda_1|^2}{4}$ represents the velocity of the soliton envelope in the $(\xi,t)$ coordinates, and we have introduced $D(\xi,t)=\Phi_1^\dagger(\xi,t) \Phi_1(\xi,t)$. It is convenient to set
\begin{equation}
	\label{new:param}
	\nu_1=\frac{b_1}{4\mathrm{v}_1}\,,~~\eta_1=-\frac{a_1}{4\mathrm{v}_1}\,.
\end{equation}
Then,
\bea
D(\xi,t)&=&e^{2\nu_1(\xi-\mathrm{v}_1 t)}||{\pmb \delta}_1||^2+e^{-2\nu_1(\xi-\mathrm{v}_1t)}||{\pmb \gamma}_1||^2\nonumber\\
&=&2||{\pmb \delta}_1||||{\pmb \gamma}_1||\cosh\left[2\nu_1(\xi-\mathrm{v}_1t)+\delta_1-\gamma_1\right],
\eea
where $e^{\delta_1}=||{\pmb \delta}_1||$ and $e^{\gamma_1}=||{\pmb \gamma}_1||$. Note that the quantity $D$ is real. Inserting in \eref{dressedsolution}, we find that
\begin{gather*}
\widetilde{V}_0(\xi,t)= \begin{pmatrix}
0_2 & U(\xi,t)\\
U(\xi,t)^\dagger & 0_2
\end{pmatrix}\,,
\end{gather*}
with the one-soliton solution
\begin{gather}
U(\xi,t) = 4i\nu_1\mathrm{v}_1\sech\left[2 \nu_1(\xi-\mathrm{v}_1 t)+\delta_1-\gamma_1\right]\times \nonumber \\
\times \left( e^{-2i \eta_1(\xi +\mathrm{v}_1 t)}
\frac{{\pmb \delta}_1}{||{\pmb \delta}_1||}\frac{{\pmb \gamma}_1^\dagger}{||{\pmb \gamma}_1||} + e^{2i \eta_1(\xi + \mathrm{v}_1 t)} \sigma_2 \frac{ {\pmb \delta}_1^*}{|| {\pmb \delta}_1||} \frac{{\pmb \gamma}_1^T}{||{\pmb \gamma}_1||}\sigma_2 \right).
\end{gather}
Noting that $(i \sigma_2{\bf v}^*)^\dagger {\bf v}=0$ and $||i \sigma_2{\bf v}^*||=||{\bf v}||$ for any ${\bf v}\in \mathbb{C}^2$, it is then  convenient to introduce the shorthand notation $i \sigma_2 {\bf v}^*={\bf v}^\perp$, and write:
\begin{gather}\label{e:141New}
U(\xi,t) = 4i\nu_1\mathrm{v}_1\sech\left[2 \nu_1(\xi-\mathrm{v}_1 t)+\delta_1-\gamma_1\right]\times \nonumber \\
\times \left( e^{-2i \eta_1(\xi + \mathrm{v}_1 t)}
\frac{{\pmb \delta}_1}{||{\pmb \delta}_1||}\frac{{\pmb \gamma}_1^\dagger}{||{\pmb \gamma}_1||} + e^{2i \eta_1(\xi + \mathrm{v}_1 t)} \frac{ {\pmb \delta}_1^\perp}{|| {\pmb \delta}_1||} \frac{({\pmb \gamma}_1^\perp)^\dagger}{||{\pmb \gamma}_1||}\right).
\end{gather}
Thus, the two vectors ${\pmb \gamma}_1$ and ${\pmb \delta}_1$ completely control the multicomponent structure of the solution, and we refer to them as the polarization vectors of the one-soliton solution. Recalling that
\be
U=\begin{pmatrix}
	-i u_1 & -i u_2\\
	i u_2^* & -i u_1^*
\end{pmatrix}\,,
\ee
one can easily extract the solution for the components $u_j$, for $j=1,2$ and recover \eref{e:fundbreathera}. However, doing so breaks the natural symmetric role played by the vectors ${\pmb \gamma}_1$ and ${\pmb \delta}_1$ and makes the upcoming analysis of the YB map acting on the polarization vectors less transparent. For this reason, we prefer to continue working with $U$ and the vectors ${\pmb \gamma}_1$ and ${\pmb \delta}_1$.

At this point, having shown how to determine an explicit solutions from the structure of the degeneracy space of $G_+$ and having already introduced the parameters $\eta_1, \nu_1$ in \eref{new:param}, we will now revert to the spectral parameter $k$ and the corresponding eigenvalue $k_1$:
\begin{equation}
k=\frac{1}{\lambda}\,,~~k_1=\eta_1+i\nu_1=\frac{1}{\lambda_1}\,.
\end{equation}
Of course, all the general results on refactorization and YB maps are unchanged by this redefinition. This allows us to compare more easily the forthcoming results with those already obtained using Manakov's method in Sec.~5	.

The polarization vectors also determine the phase shifts, and in particular the norms of ${\pmb \delta}_1$ and ${\pmb \gamma}_1$ control the position shift. The discrete eigenvalue $k_1$ controls the velocity $\mathrm{v}_1$. It is natural to denote this one-soliton solution as
$	U^{sol}\left(\xi,t;k_1,{\pmb \gamma}_1, {\pmb \delta}_1\right)$
to emphasize the quantities controlling its properties (we may drop the $(\xi,t)$ dependence for conciseness when this does not lead to confusion).

Having identified the quantities determining one-soliton solutions of rank $1$, we examine next the case of a two-soliton solutions built on rank-$1$ projectors, and investigate how the refactorization problem on the projectors discussed above induces a scattering map on the polarization vectors. To this end, we analyze the asymptotic behavior of the projector $\Pi_1(\xi,t)$ as $|t|\to\infty$ along rays $\xi - \mathrm{v} t = C$ ($C$ constant) for $\mathrm{v}<\mathrm{v}_1$ and $\mathrm{v}>\mathrm{v}_1$.

\underline{Case $\mathrm{v}<\mathrm{v}_1$.}
In view of \eref{formPi1}, and recalling that $\nu_1>0$, then the dominant term is
$e^{-2\nu_1(\mathrm{v}-\mathrm{v}_1)t}$ as $t \to \infty$, and we find that
\begin{equation}
	\label{Pi:asympt}
	\Pi_1(\xi,t)\sim\begin{pmatrix}
		0 & 0\\
		0 & \frac{{\pmb \gamma}_1{\pmb \gamma}_1^\dagger}{{\pmb \gamma}_1^\dagger{\pmb \gamma}_1}
	\end{pmatrix}+O \left(e^{2 \nu_1(\mathrm{v} - \mathrm{v}_1)t}\right)\,.
\end{equation}
On the other hand, as $t\to-\infty$, the dominant term is $e^{2\nu_1(\mathrm{v} - \mathrm{v}_1)t}$ and we obtain
\begin{equation}
	\Pi_1(\xi,t)\sim\begin{pmatrix}
		\frac{{\pmb \delta}_1{\pmb \delta}_1^\dagger}{{\pmb \delta}_1^\dagger{\pmb \delta}_1} & 0\\
		0		& 0
	\end{pmatrix}+O \left(e^{-2 \nu_1(\mathrm{v} - \mathrm{v}_1)t}\right)\,.
\end{equation}
As a consequence, denoting by $\pi_\gamma$ (resp., $\pi_\delta$) the 2-dimensional Hermitian projector $\frac{{\pmb \gamma}_1{\pmb \gamma}_1^\dagger}{{\pmb \gamma}_1^\dagger{\pmb \gamma}_1}$ (resp., $\frac{{\pmb \delta}_1{\pmb \delta}_1^\dagger}{{\pmb \delta}_1^\dagger{\pmb \delta}_1}$), we obtain
\begin{align}
	\label{T:asympt1}
	T_{k_1,\Pi_1}(k) &\sim I_4 - \frac{k}{k_1}\frac{k_1^* - k_1}{k_1^* - k}\begin{pmatrix}
		0 & 0\\
		0 & \pi_\gamma
	\end{pmatrix}
	-\frac{k}{k_1^*}\frac{k_1^*-k_1}{k_1+k}\begin{pmatrix}
	0 & 0\\
	0 & \sigma_2 \pi_\gamma^*\sigma_2
\end{pmatrix}\nonumber\\&+O \left(e^{2\nu_1(\mathrm{v} - \mathrm{v}_1)t}\right), \quad t \to \infty,
\end{align}
and
\begin{align}
	\label{T:asympt2}
	T_{k_1,\Pi_1}(k) &\sim I_4 - \frac{k}{k_1}\frac{k_1^* - k_1}{k_1^* - k}\begin{pmatrix}
		\pi_\delta & 0\\
		0 & 0
	\end{pmatrix}
	-\frac{k}{k_1^*}\frac{k_1^*-k_1}{k_1+k}\begin{pmatrix}
	\sigma_2 \pi_\delta^*\sigma_2 & 0\\
	0 & 0
\end{pmatrix}\nonumber\\&+O \left(e^{-2\nu_1(\mathrm{v} - \mathrm{v}_1)t}\right), \quad t \to -\infty.
\end{align}

\medskip

\underline{ Case $\mathrm{v}>\mathrm{v}_1$.} In this case, the role of the dominant terms as $t\to\pm\infty$ is swapped compared to the previous case, so we immediately deduce
\begin{align}
	\label{T:asympt3}
	T_{k_1,\Pi_1}(k) &\sim I_4 - \frac{k}{k_1}\frac{k_1^* - k_1}{k_1^* - k}\begin{pmatrix}
		\pi_\delta & 0\\
		0 & 0
	\end{pmatrix}
	-\frac{k}{k_1^*}\frac{k_1^*-k_1}{k_1+k}\begin{pmatrix}
	\sigma_2 \pi_\delta^*\sigma_2 & 0\\
	0 & 0
\end{pmatrix}\nonumber\\&+O \left(e^{-2\nu_1(\mathrm{v} - \mathrm{v}_1)t}\right), \quad t \to \infty,
\end{align}
and
\begin{align}
	\label{T:asympt4}
	T_{k_1,\Pi_1}(k) &\sim I_4 - \frac{k}{k_1}\frac{k_1^* - k_1}{k_1^* - k}\begin{pmatrix}
		0 & 0\\
		0 & \pi_\gamma
	\end{pmatrix}
	-\frac{k}{k_1^*}\frac{k_1^*-k_1}{k_1+k}\begin{pmatrix}
	0 & 0\\
	0 & \sigma_2 \pi_\gamma^*\sigma_2
\end{pmatrix}\nonumber\\&+O \left(e^{2\nu_1(\mathrm{v} - \mathrm{v}_1)t}\right), \quad t \to -\infty.
\end{align}

\medskip

\paragraph{Case $N=2$.} Here, we have
\begin{eqnarray}
	\label{refacG}
	G_+(k)=T_{k_2,\Pi_2}(k)\,T_{k_1,\Pi_1}(k)=T_{k_1,P_1}(k)\,T_{k_2,P_2}(k),
\end{eqnarray}
where we recall that, in the parametrization with the spectral parameter $k=\frac{1}{\lambda}$,
\be
T_{k_1,\Pi_1}(k)=\left(I_4 -\frac{k}{k_1}\frac{k_1^*-k_1}{k_1^*-k}\Pi_1\right)
	\left(I_4-\frac{k}{k_1^*}\frac{k_1^*-k_1}{k_1+k}\Lambda \Pi_1^* \Lambda^{-1}\right),
\ee
and similarly for the other factors.
The projectors $(\Pi_1,\Pi_2)$ and $(P_1,P_2)$ are related by the map arising from the refactorization discussed previously.
Now, \eref{dressedsolution} gives
\begin{subequations}
\begin{align}
\begin{pmatrix}
0_2 & U(\xi,t)\\
U(\xi,t)^\dagger & 0_2
\end{pmatrix} &= \frac{k_1^*-k_1}{2|k_1|^2}\left[\Pi_1+\widetilde{\Pi}_1,\Sigma_3\right] + \frac{k_2^*-k_2}{2|k_2|^2}\left[\Pi_2+\widetilde{\Pi}_2,\Sigma_3\right] \label{e:193a}\\
&= \frac{k_1^*-k_1}{2|k_1|^2}\left[P_1+\widetilde{P}_1,\Sigma_3\right] + \frac{k_2^*-k_2}{2|k_2|^2}\left[P_2+\widetilde{P}_2,\Sigma_3\right], \label{e:193b}
\end{align}
\end{subequations}
when we take $V_0(\xi,t)=0$ as the trivial seed solution. The projectors are found from conditions \eref{kernelcondition}.  Using the fact that $[\Pi_j,\Lambda\,\Pi_j^*\,\Lambda^{-1}]=0$, these give
\begin{eqnarray}
&&	\Im \Pi_1
={\rm span}
	\left\lbrace \begin{pmatrix}
		e^{-i(k_1 \xi - t/4k_1)}{\pmb \delta}_1\\
				e^{i(k_1 \xi - t/4k_1)}{\pmb \gamma}_1
	\end{pmatrix} \right\rbrace\,,\\
\label{Im:Pi2}
&&	\Im \Pi_2={\rm span}
	\left\lbrace  T_{k_1,\Pi_1}(k_2)\begin{pmatrix}
		e^{-i(k_2 \xi - t/4k_2)}{\pmb \delta}_2\\
		e^{i(k_2 \xi - t/4k_2)}{\pmb \gamma}_2
	\end{pmatrix} \right\rbrace\,.
\end{eqnarray}
Similarly, in view of \eref{refacG}, we find
\begin{eqnarray}\label{Im:P1}
&&	\Im P_1
={\rm span}
	\left\lbrace T_{k_2,P_2}(k_1)
	\begin{pmatrix}
	 e^{-i(k_1 \xi - t/4k_1)}{\pmb \delta}_1\\
e^{i(k_1 \xi - t/4k_1)}{\pmb \gamma}_1
	\end{pmatrix} \right\rbrace\,,\\
	&&	\Im P_2={\rm span}
	\left\lbrace
	\begin{pmatrix}
		e^{-i(k_2 \xi - t/4k_2)}{\pmb \delta}_2\\
		e^{i(k_2 \xi - t/4k_2)}{\pmb \gamma}_2
	\end{pmatrix} \right\rbrace\,.
\end{eqnarray}
We are now in a position to discuss the map on the solitons polarizations arising from soliton interactions, following the arguments given in Sec.~2.4 of \cite{CZ}.
\begin{proposition}
Suppose without loss of generality that $k_1$ and $k_2$ are such that $\mathrm{v}_2 > \mathrm{v}_1$. Then, up to exponentially small terms, as $t\to\pm\infty$ the two-soliton solution is the sum of two one-soliton solutions
 \begin{eqnarray}
	\label{Q:asympt}
	U(\xi,t) \sim U^{sol}\left( \xi,t;k_1,{\pmb \gamma}_1^\pm,{\pmb \delta}_1^\pm\right) + U^{sol}\left( \xi,t;k_2,{\pmb \gamma}_2^\pm,{\pmb \delta}_2^\pm\right),
\end{eqnarray}
where the relations between the ``incoming'' polarization vectors $({\pmb \gamma}_j^-,{\pmb \delta}_j^-)$, $j=1,2$ and the ``outgoing'' polarization vectors $({\pmb \gamma}_j^+,{\pmb \delta}_j^+)$, $j=1,2$ are given by:
\begin{subequations}\label{e:196n}
\begin{gather}
{\pmb \gamma}^{+}_1 = \frac{k_2}{k_2^*}\frac{k_1+k_2^*}{k_1+k_2} \left[ I_2 + \frac{k_1^2}{k_2^2}\frac{(k_2^*)^2-k_2^2}{k_1^2 - (k_2^*)^2} \frac{{\pmb \gamma}^{+}_2 ({\pmb \gamma}^{+}_2)^\dagger}{({\pmb \gamma}^{+}_2)^\dagger {\pmb \gamma}^{+}_2} \right]{\pmb \gamma}^{-}_1,\label{e:196na}\\
{\pmb \gamma}^{-}_2 = \frac{k_1}{k_1^*}\frac{k_2 + k_1^*}{k_2 + k_1} \left[ I_2 + \frac{k_2^2}{k_1^2} \frac{(k_1^*)^2 - k_1^2}{k_2^2 - (k_1^*)^2} \frac{{\pmb \gamma}^{-}_1 ({\pmb \gamma}^{-}_1)^\dagger}{({\pmb \gamma}^{-}_1)^\dagger{\pmb \gamma}^{-}_1} \right] {\pmb \gamma}^{+}_2,\label{e:196nb}\\
{\pmb \delta}^{-}_1 = \frac{k_2}{k_2^*}\frac{k_1+k_2^*}{k_1+k_2} \left[ I_2 + \frac{k_1^2}{k_2^2}\frac{(k_2^*)^2-k_2^2}{k_1^2 - (k_2^*)^2}\frac{{\pmb \delta}^{-}_2 ({\pmb \delta}^{-}_2)^\dagger}{({\pmb \delta}^{-}_2)^\dagger {\pmb \delta}^{-}_2} \right]{\pmb \delta}^{+}_1,\label{e:196nc}\\
{\pmb \delta}^{+}_2 =  \frac{k_1}{k_1^*}\frac{k_2 + k_1^*}{k_2 + k_1} \left[ I_2 + \frac{k_2^2}{k_1^2} \frac{(k_1^*)^2 - k_1^2}{k_2^2 - (k_1^*)^2}  \frac{{\pmb \delta}^{+}_1 ({\pmb \delta}^{+}_1)^\dagger}{({\pmb \delta}^{+}_1)^\dagger {\pmb \delta}^{+}_1} \right] {\pmb \delta}^{-}_2,\label{e:196nd}
\end{gather}
\end{subequations}
and the relations between these vectors and the ``true'' norming constants defining the degeneracy spaces are
\begin{gather}
{\pmb \gamma}^{-}_1 = {\pmb \gamma}_1, \quad {\pmb \gamma}^{+}_2 = {\pmb \gamma}_2, \quad {\pmb \delta}^{+}_1 = {\pmb \delta}_1, \quad {\pmb \delta}^{-}_2 = {\pmb \delta}_2.
\end{gather}
\end{proposition}
\begin{proof}
First, we study the behavior as $t\to\infty$. The argument will be similar for $t\to-\infty$. The idea is to use the collection of estimates \eref{Pi:asympt}--\eref{T:asympt4} appropriately applied to $\Pi_1$ and $\Pi_2$, or $P_1$ and $P_2$, recalling that $\mathrm{v}_2 > \mathrm{v}_1$. Let $\xi= \mathrm{v}_1 t+C$ and let us use \eref{e:193b} for the expression of the solution.  As $t\to\infty$, we find that $P_2$ is block diagonal, hence so is $\widetilde{P}_2$ and therefore
\begin{gather}
\begin{pmatrix}
0_2 & U(\xi,t)\\
U(\xi,t)^\dagger & 0_2
\end{pmatrix} \sim \frac{k_1^*-k_1}{2|k_1|^2}\left[P_1+\widetilde{P}_1,\Sigma_3\right],
\end{gather}
where now, using \eref{T:asympt1} applied to $T_{k_2,P_2}(k_1)$ and writing accordingly $\frac{{\pmb \gamma}_2{\pmb \gamma}_2^\dagger}{{\pmb \gamma}_2^\dagger{\pmb \gamma}_2}=\pi_{\gamma_2}$, in view of \eref{Im:P1} we get that $\Im P_1$ coincides with the span of
$$
\left[I_4 - \frac{k_1}{k_2}\frac{k_2^* - k_2}{k_2^* - k_1}\begin{pmatrix}
		0 & 0\\
		0 & \pi_{\gamma_2}
	\end{pmatrix}
	- \frac{k_1}{k_2^*}\frac{k_2^*-k_2}{k_2+k_1}\begin{pmatrix}
		0 & 0\\
		0 & \sigma_2 \pi_{\gamma_2}^*\sigma_2
	\end{pmatrix}\right]\begin{pmatrix}
	 e^{-i(k_1 \xi - t/4k_1)}{\pmb \delta}_1\\
e^{i(k_1 \xi - t/4k_1)}{\pmb \gamma}_1
	\end{pmatrix},
$$
i.e.
$$
\Im P_1\sim {\rm span}
\left\lbrace \begin{pmatrix}
	 e^{-i(k_1 \xi - t/4k_1)}{\pmb \delta}^{+}_1\\
e^{i(k_1 \xi - t/4k_1)}{\pmb \gamma}^{+}_1
	\end{pmatrix} \right\rbrace.
$$
In turn, this gives
\begin{gather}
{\pmb \gamma}^{+}_1 = \left[ I_4 - \frac{k_1}{k_2}\frac{k_2^*-k_2}{k_2^*-k_1}\pi_{\gamma_2} - \frac{k_1}{k_2^*}\frac{k_2^*-k_2}{k_2+k_1}\sigma_2 \pi_{\gamma_2}^*\sigma_2 \right]{\pmb \gamma}_1, \quad {\pmb \delta}^{+}_1 = {\pmb \delta}_1.
\end{gather}	
Therefore, along $\xi=\mathrm{v}_1 t + C$ and as $t\to\infty$, $U(\xi,t)$ behaves like
$U^{sol}\left(\xi,t;k_1,{\pmb \gamma}_1^{+},{\pmb \delta}_1^{+}\right).$
Now we analyze the solution along $\xi=\mathrm{v}_2 t + C$ as $t\to\infty$, following the same reasoning applied to \eref{e:193a}. Now $\Pi_1$ becomes block diagonal so that it does not contribute to the solution, which now behaves like 	
\begin{gather}
\begin{pmatrix}
0_2 & U(\xi,t)\\
U(\xi,t)\dagger & 0_2
\end{pmatrix} \sim \frac{k_2^*-k_2}{2|k_2|^2}\left[\Pi_2+\widetilde{\Pi}_2,\Sigma_3\right].
\end{gather}	
Using \eref{T:asympt3} for $T_{k_1,\Pi_1}(k_2)$ and \eref{Im:Pi2}, we get that $\Im \Pi_2$ coincides with the span of	
$$ \left[I_4 - \frac{k_2}{k_1}\frac{k_1^* - k_1}{k_1^*- k_2}\begin{pmatrix}
		 \pi_{\delta_1} & 0\\
		0 & 0
	\end{pmatrix}
	- \frac{k_2}{k_1^*}\frac{k_1^*-k_1}{k_1+k_2}\begin{pmatrix}
		\sigma_2 \pi_{\delta_1}^*\sigma_2 & 0\\
		0 & 0
	\end{pmatrix}\right]\begin{pmatrix}
	 e^{-i(k_2 \xi - t/4k_2)}{\pmb \delta}_2\\
e^{i(k_2 \xi - t/4k_2)}{\pmb \gamma}_2
	\end{pmatrix},$$
or equivalently
\begin{gather*}
\Im \Pi_2 \sim {\rm span}
\left\lbrace \begin{pmatrix}
	 e^{-i(k_2 \xi - t/4k_2)}{\pmb \delta}_2^{+}\\
e^{i(k_2 \xi - t/4k_2)}{\pmb \gamma}_2^{+}
	\end{pmatrix} \right\rbrace,
\end{gather*}
whence
\begin{gather*}
{\pmb \gamma}^{+}_2 = {\pmb \gamma}_2, \quad {\pmb \delta}^{+}_2 = \left[ I_4 - \frac{k_2}{k_1}\frac{k_1^*-k_1}{k_1^*-k_2} \pi_{\delta_1} - \frac{k_2}{k_1^*}\frac{k_1^*-k_1}{k_1+k_2}\sigma_2 \pi_{\delta_1}^*\sigma_2 \right]{\pmb \delta}_2, \quad \pi_{\delta_1} = \frac{{\pmb \delta}_1 {\pmb \delta}_1^\dagger}{{\pmb \delta}_1^\dagger {\pmb \delta}_1}.
\end{gather*}
Therefore, along $\xi=\mathrm{v}_2t+C$ and as $t\to\infty$, $U(\xi,t)$ behaves like
$U^{sol}\left(\xi,t;k_1,{\pmb \gamma}_2^{+},{\pmb \delta}_2^{+}\right).$
This proves the claim \eref{Q:asympt} as $t\to\infty$. The argument can be repeated as $t \to-\infty$, with similar results. For brevity, we only provide the formulas for the asymptotic polarization vectors:
\begin{gather*}
{\pmb \gamma}^{-}_1 = {\pmb \gamma}_1, \quad {\pmb \delta}^{-}_1 = \left[ I_4 - \frac{k_1}{k_2}\frac{k_2^*-k_2}{k_2^*-k_1}\pi_{\delta_2} - \frac{k_1}{k_2^*}\frac{k_2^*-k_2}{k_2+k_1}\sigma_2 \pi_{\delta_2}^* \sigma_2 \right]{\pmb \delta}_1, \quad \pi_{\delta_2} = \frac{{\pmb \delta}_2 {\pmb \delta}_2^\dagger}{{\pmb \delta}_2^\dagger {\pmb \delta}_2},
\end{gather*}
and
\begin{gather*}
{\pmb \gamma}^{-}_2 = \left[ I_4 - \frac{k_2}{k_1}\frac{k_1^*-k_1}{k_1^*-k_2}\pi_{\gamma_1} - \frac{k_2}{k_1^*}\frac{k_1^*-k_1}{k_1 + k_2} \sigma_2 \pi_{\gamma_1}^* \sigma_2 \right] {\pmb \gamma}_2, \quad {\pmb \delta}^{-}_2 = {\pmb \delta}_2, \quad \pi_{\gamma_1} = \frac{{\pmb \gamma}_1 {\pmb \gamma}_1^\dagger}{{\pmb \gamma}_1^\dagger {\pmb \gamma}_1}.
\end{gather*}
Summarizing our previous results, we have the following relations:
\begin{subequations}
\begin{gather}
{\pmb \gamma}^{+}_1 = \left[ I_4 - \frac{k_1}{k_2}\frac{k_2^*-k_2}{k_2^*-k_1}\frac{{\pmb \gamma}_2^{+} ({\pmb \gamma}_2^{+})^\dagger}{({\pmb \gamma}_2^{+})^\dagger {\pmb \gamma}_2^{+}} - \frac{k_1}{k_2^*}\frac{k_2^*-k_2}{k_2+k_1}\sigma_2\frac{({\pmb \gamma}_2^{+})^* ({\pmb \gamma}_2^{+})^T}{({\pmb \gamma}_2^{+})^T ({\pmb \gamma}_2^{+})^*}\sigma_2 \right]{\pmb \gamma}_1^{-},\\
{\pmb \gamma}^{-}_2 = \left[ I_4 - \frac{k_2}{k_1}\frac{k_1^*-k_1}{k_1^*-k_2}\frac{{\pmb \gamma}_1^{-} ({\pmb \gamma}_1^{-})^\dagger}{({\pmb \gamma}_1^{-})^\dagger {\pmb \gamma}_1^{-}} - \frac{k_2}{k_1^*}\frac{k_1^*-k_1}{k_1 + k_2} \sigma_2 \frac{({\pmb \gamma}_1^{-})^* ({\pmb \gamma}_1^{-})^T}{({\pmb \gamma}_1^{-})^T ({\pmb \gamma}_1^{-})^*} \sigma_2\right] {\pmb \gamma}_2^{+},\\
{\pmb \delta}^{-}_1 = \left[ I_4 - \frac{k_1}{k_2}\frac{k_2^*-k_2}{k_2^*-k_1}\frac{{\pmb \delta}_2^{-} ({\pmb \delta}_2^{-})^\dagger}{({\pmb \delta}_2^{-})^\dagger {\pmb \delta}_2^{-}} - \frac{k_1}{k_2^*}\frac{k_2^*-k_2}{k_2+k_1}\sigma_2 \frac{({\pmb \delta}_2^{-})^* ({\pmb \delta}_2^{-})^T}{({\pmb \delta}_2^{-})^T ({\pmb \delta}_2^{-})^*}\sigma_2 \right]{\pmb \delta}_1^{+},\\
{\pmb \delta}^{+}_2 = \left[ I_4 - \frac{k_2}{k_1}\frac{k_1^*-k_1}{k_1^*-k_2} \frac{{\pmb \delta}_1^{+} ({\pmb \delta}_1^{+})^\dagger}{({\pmb \delta}_1^{+})^\dagger {\pmb \delta}_1^{+}} - \frac{k_2}{k_1^*}\frac{k_1^*-k_1}{k_1+k_2}\sigma_2\frac{({\pmb \delta}_1^{+})^* ({\pmb \delta}_1^{+})^T}{({\pmb \delta}_1^{+})^T ({\pmb \delta}_1^{+})^*}\sigma_2\right]{\pmb \delta}_2^{-}.
\end{gather}
\end{subequations}
Using the property
\begin{gather}
\sigma_2 \frac{(\pmb{\gamma}_j^{\pm})^* (\pmb{\gamma}_j^{\pm})^{T}}{(\pmb{\gamma}_j^{\pm})^{T} (\pmb{\gamma}_j^{\pm})^*}\sigma_2 = I_2 - \frac{\pmb{\gamma}_j^{\pm} (\pmb{\gamma}_j^{\pm})^\dagger}{(\pmb{\gamma}_j^{\pm})^\dagger \pmb{\gamma}_j^{\pm}}, \quad \sigma_2 \frac{ (\pmb{\delta}_j^{+})^* (\pmb{\delta}_j^{\pm})^{T}}{(\pmb{\delta}_j^{\pm})^{T} (\pmb{\delta}_j^{\pm})^*}\sigma_2 = I_2 - \frac{\pmb{\delta}_j^{\pm} (\pmb{\delta}_j^{\pm})^\dagger}{(\pmb{\delta}_j^{\pm})^\dagger \pmb{\delta}_j^{\pm}},
\end{gather}
for $j=1,2$, one can reduce the last equations into \eref{e:196n}, where we essentially express the polarization vectors after the interaction in terms of parameters characterizing the initial solitons. Notice that since soliton-1 is faster than soliton-2, then the states before/after the interaction coincide with the limits $t \to \pm \infty$, while for soliton-2 the state before the interaction coincides with the limit $t \to + \infty$, and the state after the interaction coincides with the limit $t \to - \infty$.
\end{proof}

\begin{remark}
Note that one can easily rewrite eqs \eref{e:196n} more explicitly as a scattering map $({\pmb \gamma}_1^-,{\pmb \delta}_1^-,{\pmb \gamma}_2^-,{\pmb \delta}_2^-)  \mapsto  ({\pmb \gamma}_1^+,{\pmb \delta}_1^+,{\pmb \gamma}_2^+,{\pmb \delta}_2^+) $ (which is its intended meaning) but then \eref{e:196na} and \eref{e:196nd} become rather lengthy since one has to replace ${\pmb \gamma}_2^+$ and ${\pmb \delta}_1^+$ by their expressions in terms of ${\pmb \gamma}_1^-,{\pmb \delta}_1^-,{\pmb \gamma}_2^-,{\pmb \delta}_2^-$ obtained from \eref{e:196nb} and \eref{e:196nc}. This is the reason why we chose to write the map more compactly as in \eref{e:196n}.
\end{remark}

\paragraph{Self-symmetric case.} Note that from our general results on the structure of the dressing factors, the self-symmetric case is simply a limit of the above results when $k_1$ and/or $k_2$ become purely imaginary. It is easy to see that for instance if $k_2=-k^*_2$ then the maps \eref{e:196na} and \eref{e:196nc} become trivial in the sense that the polarization vectors after the interaction are proportional to those before the interaction. If $k_1$ remains generic then \eref{e:196nb} and \eref{e:196nd} retains their matrix structure. If both zeros are self-symmetric then the the map is trivial, again in the sense that only the norm of the vectors is changed upon interaction, not their direction. This only produces a phase and position shift. This gives a description of the interaction between a self-symmetric soliton and a fundamental breather or between two self-symmetric solitons, respectively.

\subsubsection{Rank-2 case: composite breathers}\label{e:sub6.3.2}

We present a formula for the one-soliton solution in the generic rank-2 case. Recall that this means that the dressing factor is
$g_{k_1,\Pi_1}(k)g_{k_2,\Pi_2}(k)$ with $k_2=-k_1^*$ and $\Pi_2=\Lambda\,\Pi_1^*\,\Lambda^{-1}$ where $\Pi_1$ is a Hermitian projector of rank $2$. Let $\phi_1$ and $\phi_2$ be a basis of the image of $\Pi_1$ (at $t=\xi=0$) (we can choose them orthogonal wlog).
Let
\begin{equation}
	\Phi_1(\xi,t)=e^{-i(k_1 \xi - t/4k_1)\Sigma_3}(\phi_1,\phi_2)\,,
\end{equation}
then, dropping again the $(\xi,t)$ dependence for conciseness,
\begin{eqnarray}
	\label{formulaPi1rank2}
	\Pi_1=\Phi_1(\Phi_1^\dagger \Phi_1)^{-1}\Phi_1^\dagger\,.
\end{eqnarray}
It is convenient to write
\begin{eqnarray}
	\phi_1=\begin{pmatrix}
		{\pmb \delta}_1\\
				{\pmb \gamma}_1
	\end{pmatrix}\,,~~
	\phi_2=\begin{pmatrix}
	{\pmb \tau}_1\\
	{\pmb \omega}_1
\end{pmatrix}\,,~~	{\pmb \delta}_1,{\pmb \gamma}_1,{\pmb \tau}_1,{\pmb \omega}_1\in \mathbb{C}^2\,,~~j=1,2\,.
\end{eqnarray}
Then one obtains
\begin{gather}
\Phi_1^\dagger \Phi_1 = \begin{pmatrix}
{\cal F}_{11} & 	{\cal F}_{12}\\
		{\cal F}_{21} & 	{\cal F}_{22}
\end{pmatrix},
\end{gather}
where
\begin{subequations}
\begin{gather}
{\cal F}_{11} = e^{2\nu_1 (\xi + t/4|k_1|^2)}{\pmb \delta}_1^\dagger {\pmb \delta}_1+e^{-2\nu_1 (\xi + t/4|k_1|^2)}{\pmb \gamma}_1^\dagger{\pmb \gamma}_1 ,\\
{\cal F}_{12} = e^{2\nu_1 (\xi + t/4|k_1|^2)}{\pmb \delta}_1^\dagger {\pmb \tau}_1+e^{-2\nu_1 (\xi + t/4|k_1|^2)}{\pmb \gamma}_1^\dagger{\pmb \omega}_1,\\
{\cal F}_{21} = e^{2\nu_1 (\xi + t/4|k_1|^2)}{\pmb \tau}_1^\dagger {\pmb \delta}_1+e^{-2\nu_1 (\xi + t/4|k_1|^2)}{\pmb \omega}_1^\dagger{\pmb \gamma}_1,\\
{\cal F}_{22} =e^{2\nu_1 (\xi + t/4|k_1|^2)}{\pmb \tau}_1^\dagger {\pmb \tau}_1+e^{-2\nu_1 (\xi + t/4|k_1|^2)}{\pmb \omega}_1^\dagger{\pmb \omega}_1.
\end{gather}
\end{subequations}
Then in particular
\begin{eqnarray}
	(\Phi_1^\dagger \Phi_1)^{-1}	=\frac{1}{D(\xi,t)}\begin{pmatrix}
		{\cal F}_{22} & 	-{\cal F}_{12}\\
		-{\cal F}_{21} & 	{\cal F}_{11}
	\end{pmatrix}\,,~~D(\xi,t)={\det}\,\Phi_1^\dagger \Phi_1\,.
\end{eqnarray}
After long but straightforward calculations, we find
\begin{eqnarray}\label{e:171compbr}
	&& U(\xi,t)=\frac{2i \nu_1}{|k_1|^2}\frac{1}{D(\xi,t)}
\left[			e^{ -2i \eta_1(\xi - t/4|k_1|^2)}\begin{pmatrix}
				{\pmb \delta}_1 & {\pmb \tau}_1
			\end{pmatrix}\begin{pmatrix}
			{\cal F}_{22}^* & 	-{\cal F}_{21}^*\\
			-{\cal F}_{12}^* & 	{\cal F}_{11}^*
		\end{pmatrix}\begin{pmatrix}
		{\pmb \gamma}_1^\dagger\\
		{\pmb \omega}_1^\dagger
	\end{pmatrix}\right.\nonumber\\
&&\qquad\qquad\left.	+e^{2i \eta_1(\xi - t/4|k_1|^2)}\sigma_2 \begin{pmatrix}
	{\pmb \delta}_1^* & {\pmb \tau}_1^*
\end{pmatrix}\begin{pmatrix}
	{\cal F}_{22} & 	-{\cal F}_{21}\\
	-{\cal F}_{12} & 	{\cal F}_{11}
\end{pmatrix}\begin{pmatrix}
	{\pmb \gamma}_1^T\\
	{\pmb \omega}_1^T
\end{pmatrix}\sigma_2  \right].
	\end{eqnarray}
Note that ${\cal F}_{11}^*={\cal F}_{11}$, ${\cal F}_{22}^*={\cal F}_{22}$, and ${\cal F}_{12}^*={\cal F}_{21}$. Even though we could normalize the four vectors ${\pmb \delta}_1,{\pmb \gamma}_1,{\pmb \tau}_1,{\pmb \omega}_1$ characterizing the structure of a composite breather solution, it appears less natural to do so in the above expression than in the case of the fundamental breather. Thus we keep them non-normalized, and denote the above one-soliton solution as
$$U^{cb}(\xi,t;k_1,{\pmb \delta}_1,{\pmb \gamma}_1,{\pmb \tau}_1,{\pmb \omega}_1)\,.$$
Two examples of such composite breathers are shown in Fig.~\ref{fig:5}.
\begin{figure}[h!]
\centering
\begin{subfigure}{0.48\textwidth}
\centering
\includegraphics[height=1.5in,width=1.5in]{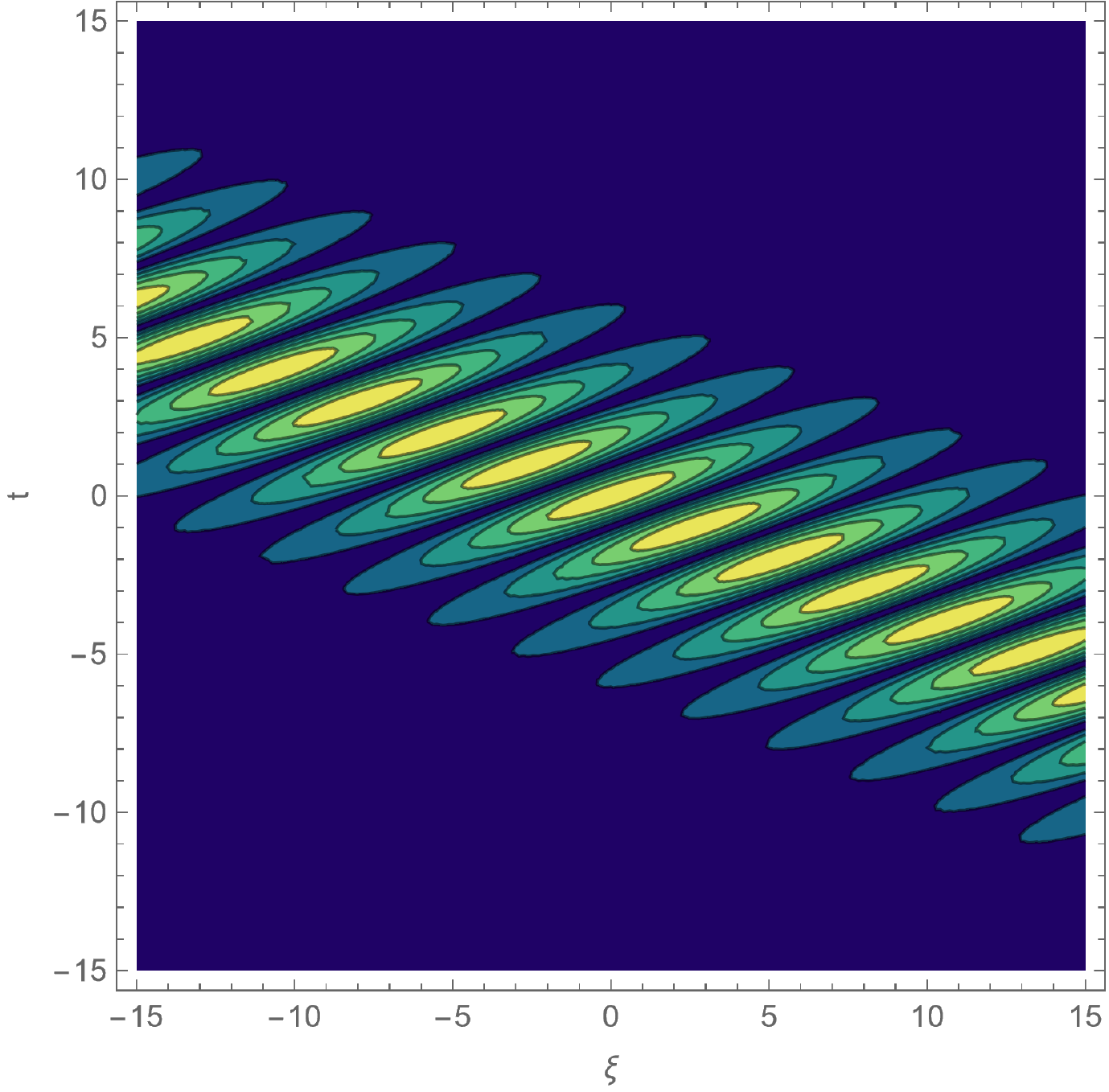}
\caption{} \label{fig:5a}
\end{subfigure}
\hfill
\begin{subfigure}{0.48\textwidth}
\centering
\includegraphics[height=1.5in,width=1.5in]{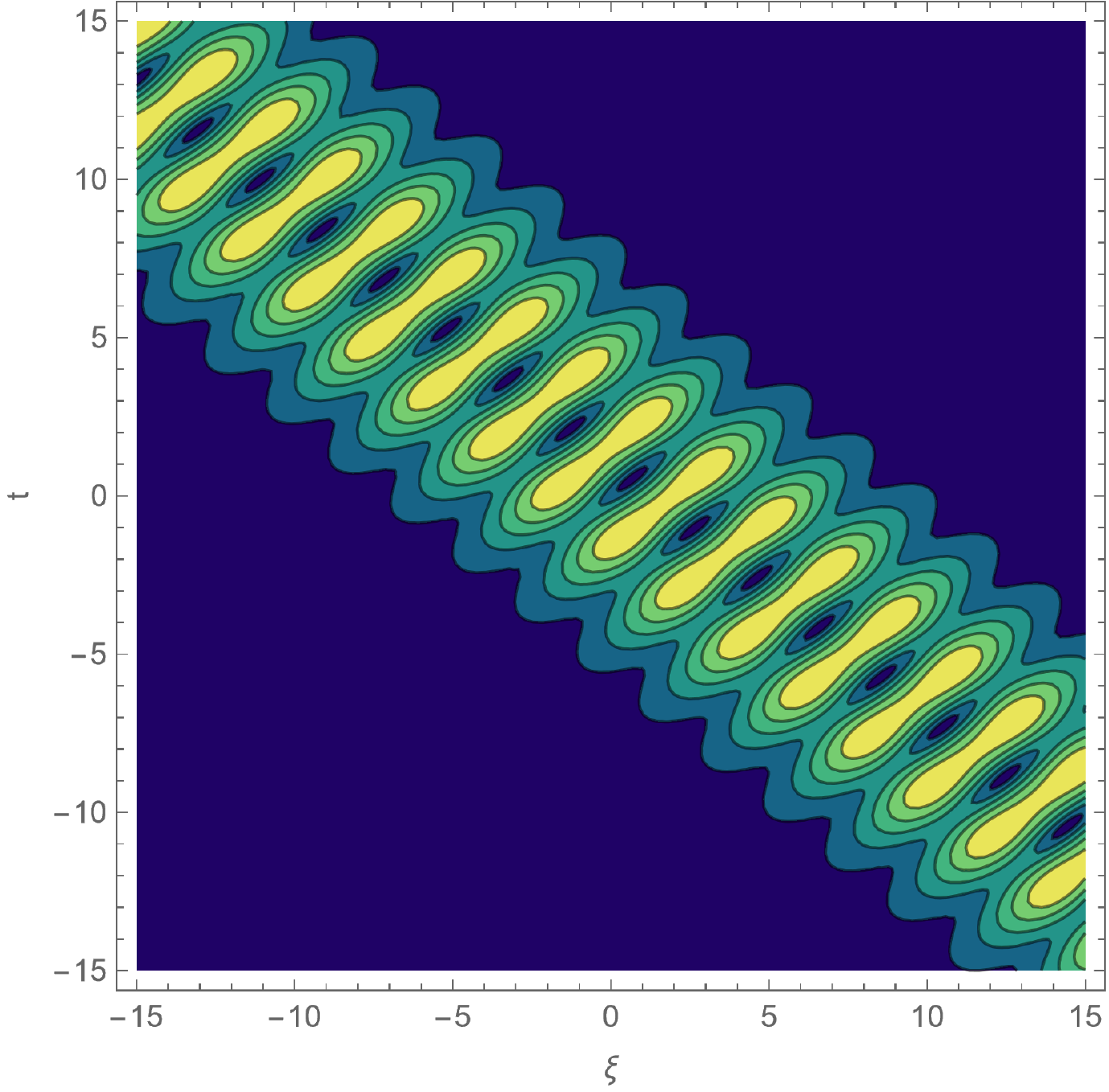}
\caption{} \label{fig:5b}
\end{subfigure}
\hfill
\caption{Contour plots for composite breather solutions. The panels show the magnitude of the component $u_1(\xi,t)$ for two different examples. Here, the soliton parameters are: ${\pmb \delta}_1=(1\,,2)^{T}$, ${\pmb \tau}_1=(2\,,1)^{T}$, ${\pmb \gamma}_1=(10\,,2)^{T}$, ${\pmb \omega}_1=(2\,,e^{i/2})^{T}$ and $k_1=4/5+i/5$ for panel (a); ${\pmb \delta}_1=(1\,,1)^{T}$, ${\pmb \tau}_1=(2\,,3)^{T}$, ${\pmb \gamma}_1=(1\,,2)^{T}$, ${\pmb \omega}_1=(2\,,1)^{T}$ and $k_1=1/2+i/4$ for panel (b).}\label{fig:5}
\end{figure}

Following the same reasoning as in the previous section, our next task is to determine the asymptotic behavior of the projector
$\Pi_1(\xi,t)$ as $|t|\to \infty$ when $\xi-\mathrm{v}t=C$ for $\mathrm{v}<\mathrm{v}_1$ and $\mathrm{v}>\mathrm{v}_1$. Again, the calculations based on the explicit form of $\Phi_1(\xi,t)$ are long but straightforward. Unlike the rank-$1$ case, they show an interesting phenomenon: we need to distinguish several cases depending on whether ${\pmb \delta}_1$ and ${\pmb \tau}_1$ are independent vectors or not, and similarly for ${\pmb \gamma}_1$ and ${\pmb \omega}_1$. We call the case where $\det ({\pmb \delta}_1,{\pmb \tau}_1)\neq 0$ and $\det ({\pmb \gamma}_1,{\pmb \omega}_1)\neq 0$ the {\it generic case}. If either $\det ({\pmb \delta}_1,{\pmb \tau}_1)= 0$ or $\det ({\pmb \gamma}_1,{\pmb \omega}_1)= 0$, we will talk about the {\it non-generic case}.
As we show in Appendix E by computing the associated transmission coefficients, in either of the above non-generic cases the composite breather reduces to a fundamental breather. Moreover, one can see directly from \eref{e:171compbr} that the case where both determinants are zero does not produce a soliton solution. Therefore, in the following we will only consider generic composite breathers.

\paragraph{Generic composite breather ($\det ({\pmb \delta}_1,{\pmb \tau}_1)\neq 0$ and $\det ({\pmb \gamma}_1,{\pmb \omega}_1)\neq 0$).}\ \newline
\underline{Case $\mathrm{v}<\mathrm{v}_1$.}
Recalling that $\nu_1>0$, then the dominant term is
$e^{-2 \nu_1(\mathrm{v}-\mathrm{v}_1)t}$ as $t \to \infty$, and introducing ${\cal N}=({\pmb \gamma}_1~{\pmb \omega}_1)$, we find
\begin{equation}
	\Pi_1(\xi,t)\sim\begin{pmatrix}
		0 & 0\\
		0 & {\cal N}({\cal N}^\dagger{\cal N})^{-1}{\cal N}^\dagger
	\end{pmatrix}+O \left(e^{2 \nu_1(\mathrm{v} - \mathrm{v}_1)t}\right)\,.
\end{equation}
Now since ${\pmb \gamma}_1,{\pmb \omega}_1$ are assumed to be independent, the projector ${\cal N}({\cal N}^\dagger{\cal N})^{-1}{\cal N}^\dagger$ is of full rank equal to $2$ and hence is simply the identity matrix $I_2$. Consequently, we find
\begin{equation}
	\Pi_1(\xi,t)\sim\begin{pmatrix}
		0 & 0\\
		0 & I_2
	\end{pmatrix}+ O \left(e^{2 \nu_1(\mathrm{v} - \mathrm{v}_1)t}\right)\,.
\end{equation}
On the other hand, as $t\to-\infty$, the dominant term is $e^{2\nu_1(\mathrm{v} - \mathrm{v}_1)t}$ and we obtain, since ${\pmb \delta}_1,{\pmb \tau}_1$ are independent,
\begin{equation}
	\Pi_1(\xi,t)\sim\begin{pmatrix}
		I_2 & 0\\
		0		& 0
	\end{pmatrix}+ O \left(e^{-2 \nu_1(\mathrm{v} - \mathrm{v}_1)t}\right)\,.
\end{equation}
As a consequence, we obtain
\begin{align}
	T_{k_1,\Pi_1}(k) \sim   \begin{pmatrix}
		I_2 & 0_2\\
		0_2 & \frac{k-k_1}{k-k_1^*} \frac{k+k_1^*}{k+k_1}I_2
	\end{pmatrix} +O \left(e^{2\nu_1(\mathrm{v} - \mathrm{v}_1)t}\right), \quad t \to \infty,
\end{align}
and
\begin{align}
	T_{k_1,\Pi_1}(k) \sim   \begin{pmatrix}
		\frac{k-k_1}{k-k_1^*} \frac{k+k_1^*}{k+k_1}I_2 & 0_2\\
		0_2 & I_2
	\end{pmatrix}+O \left(e^{-2\nu_1(\mathrm{v} - \mathrm{v}_1)t}\right), \quad t \to -\infty.
\end{align}

\medskip

\noindent \underline{Case $\mathrm{v}>\mathrm{v}_1$.} In this case, the role of the dominant terms as $t\to\pm\infty$ is swapped compared to the previous case so we immediately deduce
\begin{align}
	T_{k_1,\Pi_1}(k) \sim  \begin{pmatrix}
		\frac{k-k_1}{k-k_1^*} \frac{k+k_1^*}{k+k_1}I_2 & 0_2\\
		0_2 & I_2
	\end{pmatrix}+O \left(e^{-2\nu_1(\mathrm{v} - \mathrm{v}_1)t}\right), \quad t \to \infty,
\end{align}
and
\begin{align}
	T_{k_1,\Pi_1}(k) \sim   \begin{pmatrix}
		I_2 & 0_2\\
		0_2 & \frac{k-k_1}{k-k_1^*} \frac{k+k_1^*}{k+k_1}I_2
	\end{pmatrix}+O \left(e^{2\nu_1(\mathrm{v} - \mathrm{v}_1)t}\right), \quad t \to -\infty.
\end{align}

We are now in a position to derive the map induced on the vectors describing a two-soliton solution made of composite breathers. We follow exactly the same reasoning as in the rank-$1$ case, combining the above asymptotic results obtained on the projectors with the refactorization properties that allow to express the solution in two equivalent ways. Thus, one can show that, as $t\to\pm\infty$,
\begin{eqnarray}
	U(\xi,t) \sim U^{cb}(\xi,t;k_1,{\pmb \delta}_1^\pm,{\pmb \gamma}_1^\pm,{\pmb \tau}_1^\pm,{\pmb \omega}_1^\pm)+U^{cb}(\xi,t;k_2,{\pmb \delta}_2^\pm,{\pmb \gamma}_2^\pm,{\pmb \tau}_2^\pm,{\pmb \omega}_2^\pm),
\end{eqnarray}
where the incoming/outgoing polarization vectors ${\pmb \delta}_j^\pm,{\pmb \gamma}_j^\pm,{\pmb \tau}_j^\pm,{\pmb \omega}_j^\pm$, $j=1,2$ are related via the following maps. Note that we make use of the same property already used in the rank-$1$ case to simplify the expressions, namely for any two-dimensional projector $\pi$ of rank $1$ we have
\begin{eqnarray}
	\pi + \sigma_2 \,\pi^*\,\sigma_2=I_2\,.
\end{eqnarray}
As discussed above, we assume that both composite breathers are generic, i.e., $\det ({\pmb \delta}_j,{\pmb \tau}_j)\neq 0$ and $\det ({\pmb \gamma}_j,{\pmb \omega}_j)\neq 0$, $j=1,2$. We obtain the following relations between the incoming/outgoing polarization vectors:
\begin{subequations}
\begin{gather*}
{\pmb \gamma}_1^+ = \frac{k_1-k_2}{k_1-k_2^*}\frac{k_1+k_2^*}{k_1+k_2}{\pmb \gamma}_1^-, \quad {\pmb \delta}_1^- = \frac{k_1-k_2}{k_1-k_2^*}\frac{k_1+k_2^*}{k_1+k_2}{\pmb \delta}_1^+,\\
 {\pmb \omega}_1^+ = \frac{k_1-k_2}{k_1-k_2^*}\frac{k_1+k_2^*}{k_1+k_2}{\pmb \omega}_1^-, \quad {\pmb \tau}_1^- = \frac{k_1-k_2}{k_1-k_2^*}\frac{k_1+k_2^*}{k_1+k_2}{\pmb \tau}_1^+,\\
{\pmb \gamma}_2^- = \frac{k_2-k_1}{k_2-k_1^*}\frac{k_2+k_1^*}{k_2+k_1}{\pmb \gamma}_2^+, \quad {\pmb \delta}_2^+ = \frac{k_2-k_1}{k_2-k_1^*}\frac{k_2+k_1^*}{k_2+k_1}{\pmb \delta}_2^-,\\
{\pmb \omega}_2^- = \frac{k_2-k_1}{k_2-k_1^*}\frac{k_2+k_1^*}{k_2+k_1}{\pmb \omega}_2^+, \quad {\pmb \tau}_2^+ = \frac{k_2-k_1}{k_2-k_1^*}\frac{k_2+k_1^*}{k_2+k_1}{\pmb \tau}_2^-,
\end{gather*}
\end{subequations}
and the following relations between these vectors and the norming constants:
\begin{subequations}
\begin{gather*}
{\pmb \gamma}_1^- = {\pmb \gamma}_1, \quad {\pmb \delta}_1^+ = {\pmb \delta}_1, \quad {\pmb \omega}_1^- = {\pmb \omega}_1, \quad {\pmb \tau}_1^+ = {\pmb \tau}_1,\\
{\pmb \gamma}_2^+ = {\pmb \gamma}_2, \quad {\pmb \delta}_2^- = {\pmb \delta}_2, \quad {\pmb \omega}_2^+ = {\pmb \omega}_2, \quad {\pmb \tau}_2^- = {\pmb \tau}_2.
\end{gather*}
\end{subequations}
We see that this map is trivial in the sense that the vectors characterizing a generic composite breather are only rescaled upon interaction, thus only contributing to a phase shift and position shift. The internal structure of generic composite breathers is not affected by interactions.

\subsubsection{Mixed rank-1 and rank-2 case}

We can also investigate the map on the soliton data in the case where one soliton corresponds to a rank-$1$ projector and the other to a rank-$2$ projector. As a result, below we will be able to discuss the interaction between a fundamental soliton or a fundamental breather with a composite breather.

We already have established the required methodology, so we simply state the results. Suppose that soliton $1$, with associated eigenvalue $k_1$, corresponds to a rank-$1$ projector characterized by ${\pmb \delta}_1$ and ${\pmb \gamma}_1$, and that soliton $2$, with associated eigenvalue $k_2$, corresponds to a rank-$2$ projector characterized by ${\pmb \delta}_2$, ${\pmb \gamma}_2$, ${\pmb \tau}_2$ ${\pmb \omega}_2$ with $\det ({\pmb \delta}_2,{\pmb \tau}_2)\neq 0$ and $\det ({\pmb \gamma}_2,{\pmb \omega}_2)\neq 0$. As before, we assume without loss of generality that $\mathrm{v}_2>\mathrm{v}_1$. As $t\to\pm \infty$,
\begin{eqnarray}
	U(\xi,t) \sim U^{sol}(\xi,t;k_1,{\pmb \delta}_1^\pm,{\pmb \gamma}_1^\pm)+U^{cb}(\xi,t;k_2,{\pmb \delta}_2^\pm,{\pmb \gamma}_2^\pm,{\pmb \tau}_2^\pm,{\pmb \omega}_2^\pm),
\end{eqnarray}
with the following maps among the polarization vectors before/after the soliton interaction:
\begin{subequations}
\begin{gather}
{\pmb \gamma}_1^+=\frac{k_1-k_2}{k_1-k_2^*}\frac{k_1+k_2^*}{k_1+k_2}{\pmb \gamma}_1^-, \quad {\pmb \delta}_1^-=  \frac{k_1-k_2}{k_1-k_2^*}\frac{k_1+k_2^*}{k_1+k_2}{\pmb \delta}_1^+,\\
{\pmb \gamma}_2^-= \frac{k_1}{k_1^*}\frac{k_1^*+k_2}{k_1+k_2} \left[ I_2 + \frac{k_2^2}{k_1^2}\frac{(k_1^*)^2-k_1^2}{k_2^2 - (k_1^*)^2} \frac{{\pmb \gamma}_1^-({\pmb \gamma}_1^-)^\dagger}{({\pmb \gamma}_1^-)^\dagger{\pmb \gamma}_1^-} \right] {\pmb \gamma}_2^+, \\
{\pmb \delta}_2^+= \frac{k_1}{k_1^*}\frac{k_1^*+k_2}{k_1+k_2} \left[ I_2 + \frac{k_2^2}{k_1^2}\frac{k_1^2-(k_1^*)^2}{(k_1^*)^2 - k_2^2} \frac{{\pmb \delta}_1^-({\pmb \delta}_1^-)^\dagger}{({\pmb \delta}_1^-)^\dagger{\pmb \delta}_1^-} \right] {\pmb \delta}_2^-, \\
{\pmb \omega}_2^-= \frac{k_1}{k_1^*}\frac{k_1^*+k_2}{k_1+k_2} \left[ I_2 + \frac{k_2^2}{k_1^2}\frac{(k_1^*)^2-k_1^2}{k_2^2 - (k_1^*)^2} \frac{{\pmb \gamma}_1^-({\pmb \gamma}_1^-)^\dagger}{({\pmb \gamma}_1^-)^\dagger{\pmb \gamma}_1^-} \right] {\pmb \omega}_2^+,\\
{\pmb \tau}_2^+= \frac{k_1}{k_1^*}\frac{k_1^*+k_2}{k_1+k_2} \left[ I_2 + \frac{k_2^2}{k_1^2}\frac{k_1^2-(k_1^*)^2}{(k_1^*)^2 - k_2^2} \frac{{\pmb \delta}_1^-({\pmb \delta}_1^-)^\dagger}{({\pmb \delta}_1^-)^\dagger{\pmb \delta}_1^-} \right] {\pmb \tau}_2^-,
\end{gather}
\end{subequations}
while the relations between these and the norming constants are as follows
\begin{subequations}
\begin{gather}
{\pmb \gamma}_1^- = {\pmb \gamma}_1, \quad {\pmb \delta}_1^+ = {\pmb \delta}_1, \quad {\pmb \gamma}_2^+ = {\pmb \gamma}_2, \quad {\pmb \delta}_2^- = {\pmb \delta}_2, \quad {\pmb \omega}_2^+ = {\pmb \omega}_2, \quad {\pmb \tau}_2^- = {\pmb \tau}_2.
\end{gather}
\end{subequations}
Fig.~\ref{fig:7} shows two examples of a fundamental soliton and a fundamental breather interacting with a generic composite breather, showing that the fundamental soliton retains its nature upon interacting with a generic composite breather.
\begin{figure}[h!]
\centering
\begin{subfigure}{0.48\textwidth}
\centering
\includegraphics[height=1.5in,width=1.5in]{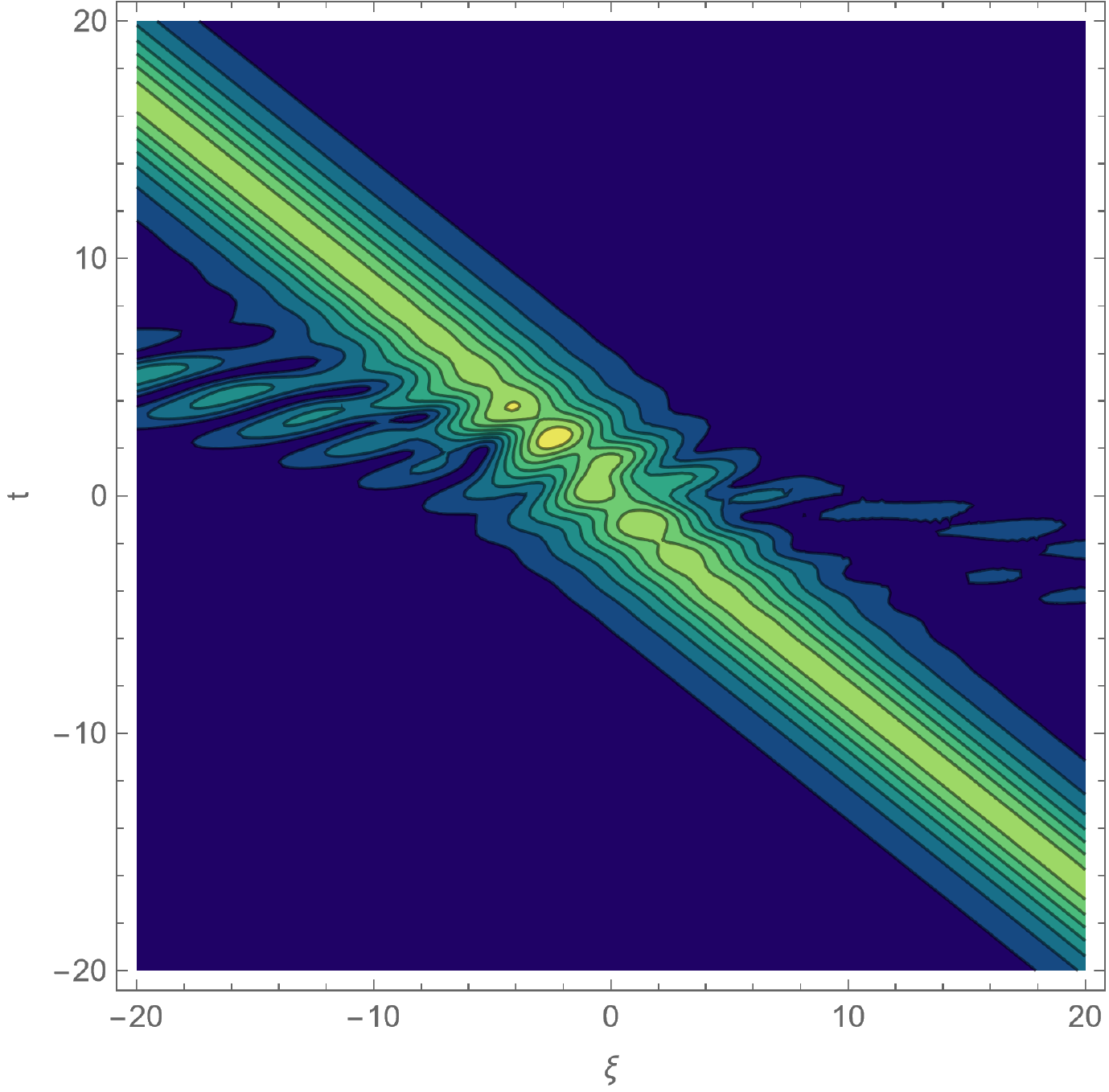}
\caption{} \label{fig:7a}
\end{subfigure}
\hfill
\begin{subfigure}{0.48\textwidth}
\centering
\includegraphics[height=1.5in,width=1.5in]{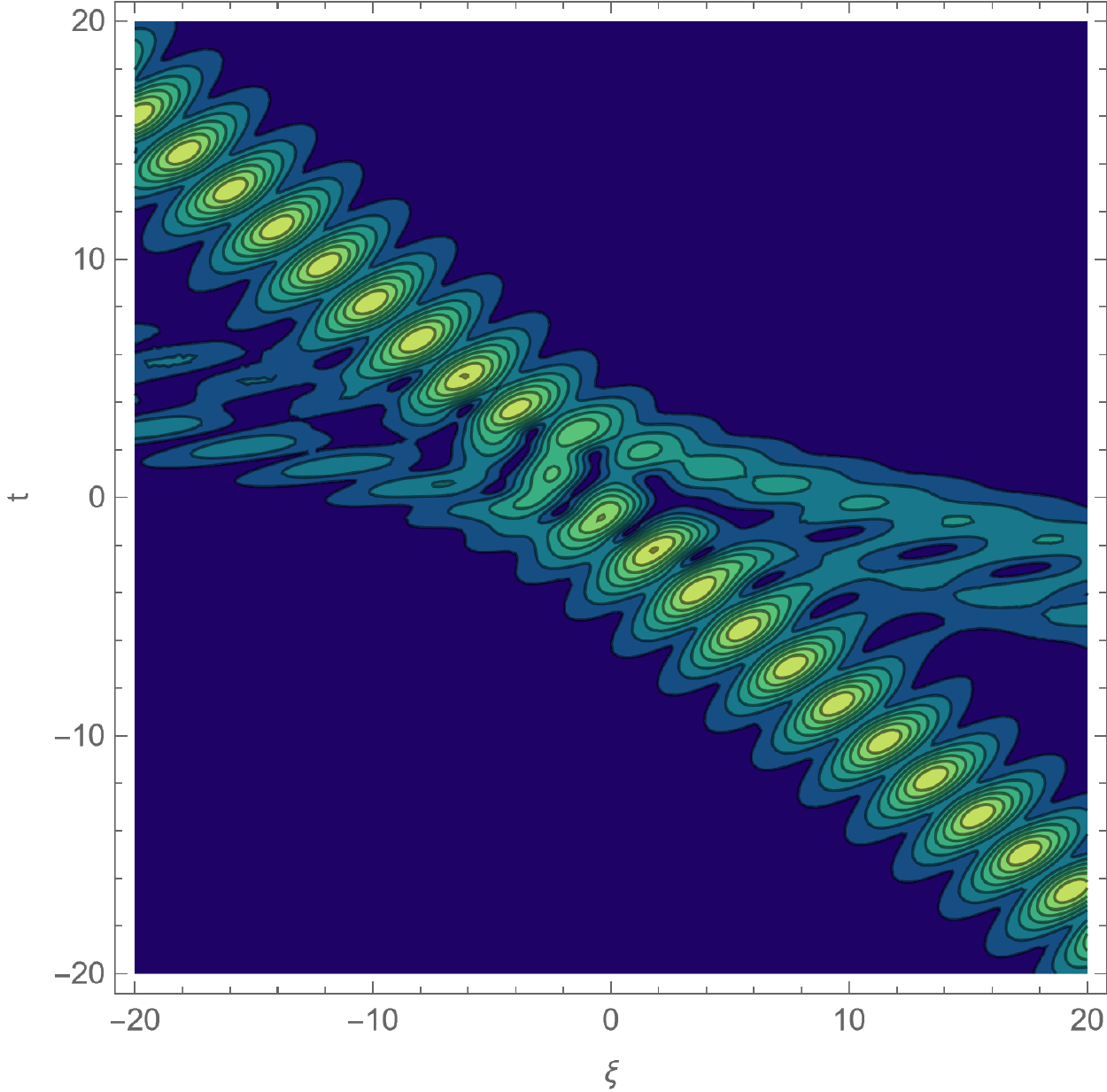}
\caption{} \label{fig:7b}
\end{subfigure}
\caption{A fundamental soliton (left) and fundamental breather (right) interacting with a generic composite breather. The panels show the magnitude of the first component of the vector solution. Panel (a): interaction between a fundamental soliton and a generic composite breather. Panel (b): interaction between a fundamental breather and a generic composite breather. Here, the soliton parameters are: panel (a): $\pmb{\gamma}_1 = \left( 1\,, 1/2\right)^{T}$, $\pmb{\delta}_1 = \left( 1\,,0\right)^{T}$ for soliton-1 being a fundamental soliton, and $\pmb{\delta}_2 = \left( 1\,,1\right)^{T}$, $\pmb{\tau}_2 = \left( 1\,, 2\right)^{T}$, $\pmb{\gamma}_2 = \left( 1,2\right)^{T}$, $\pmb{\omega}_2 = \left( 3,e^{i}\right)^{T}$, for soliton-2 being a generic composite breather. Panel (b): $\pmb{\gamma}_1 = \left( 1\,, 1/2\right)^{T}$, $\pmb{\delta}_1 = \left(1\,,1\right)^{T}$ for soliton-1 being a fundamental breather, and $\pmb{\delta}_2 = \left( 1\,,2\right)^{T}$, $\pmb{\tau}_2 = \left( 2\,, 1\right)^{T}$, $\pmb{\gamma}_2 = \left( 10,2\right)^{T}$, $\pmb{\omega}_2 = \left( e^{i},1\right)^{T}$ for soliton-2 being a generic composite breather. In both examples, the discrete eigenvalues have been chosen as: $k_1=1/2+i/4$, $k_2=1+i/2$.}\label{fig:7}
\end{figure}

Finally, Fig.~\ref{fig:6} shows some examples of mixed rank-1 and rank-2 soliton interactions, as well as interactions between two rank-2 solitons, to illustrate the results we derived in the current and the previous sections.

\begin{figure}[h!]
\centering
\begin{subfigure}{0.45\textwidth}
\centering
\includegraphics[height=1.5in,width=1.5in]{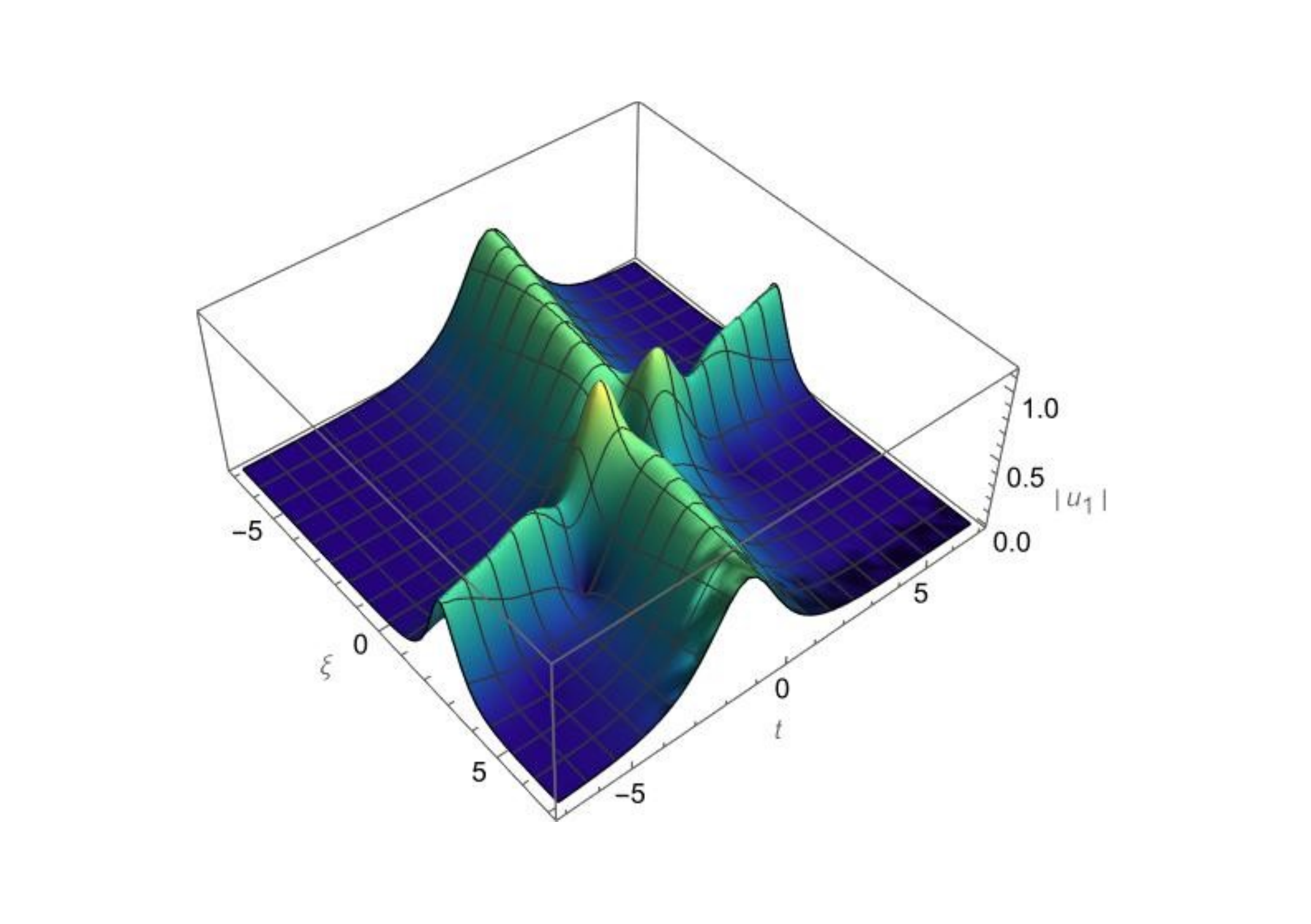}
\caption{} \label{fig:6a}
\end{subfigure}
\hfill
\begin{subfigure}{0.45\textwidth}
\centering
\includegraphics[height=1.5in,width=1.5in]{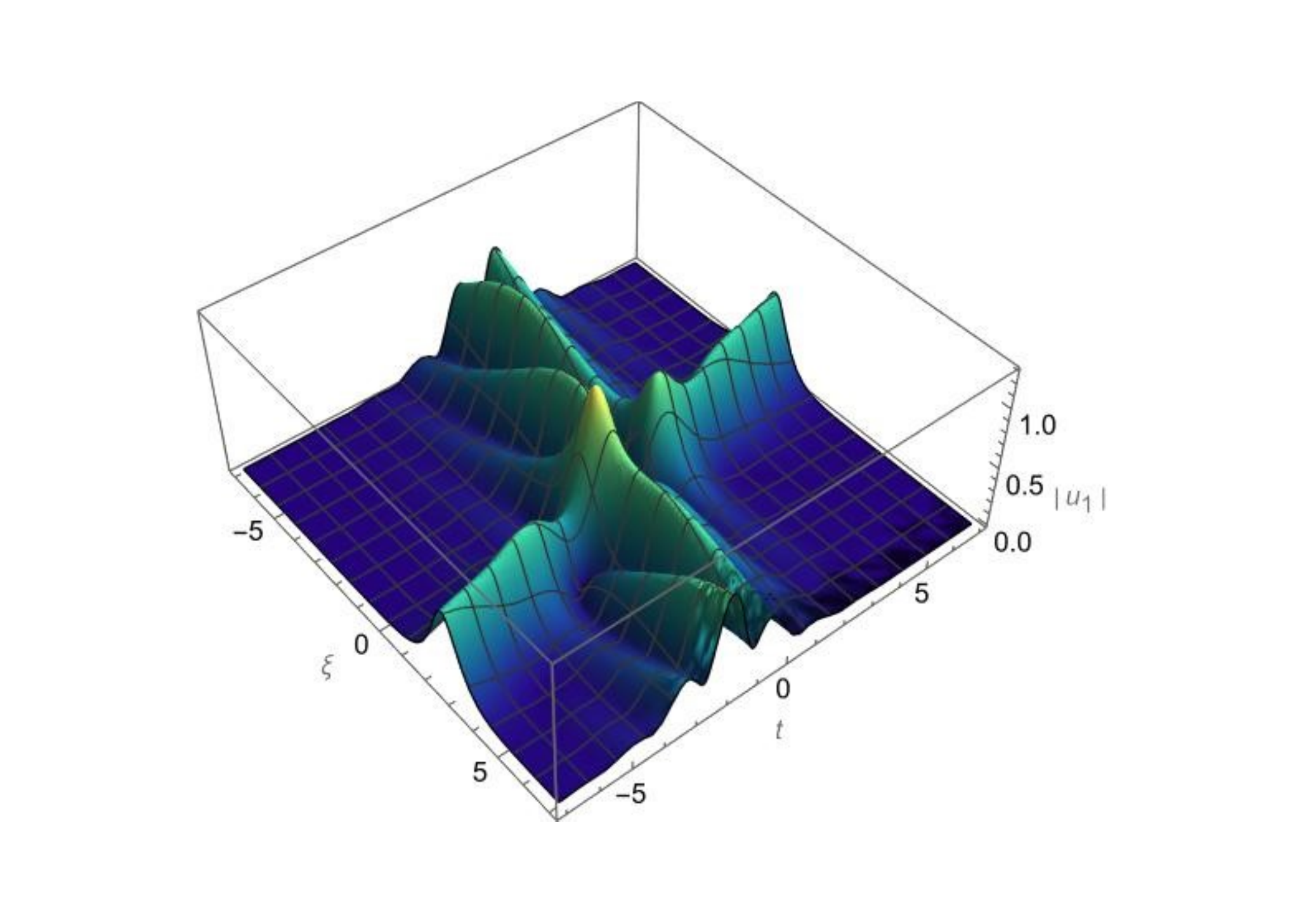}
\caption{} \label{fig:6b}
\end{subfigure}
\hfill
\begin{subfigure}{0.45\textwidth}
\centering
\includegraphics[height=1.5in,width=1.5in]{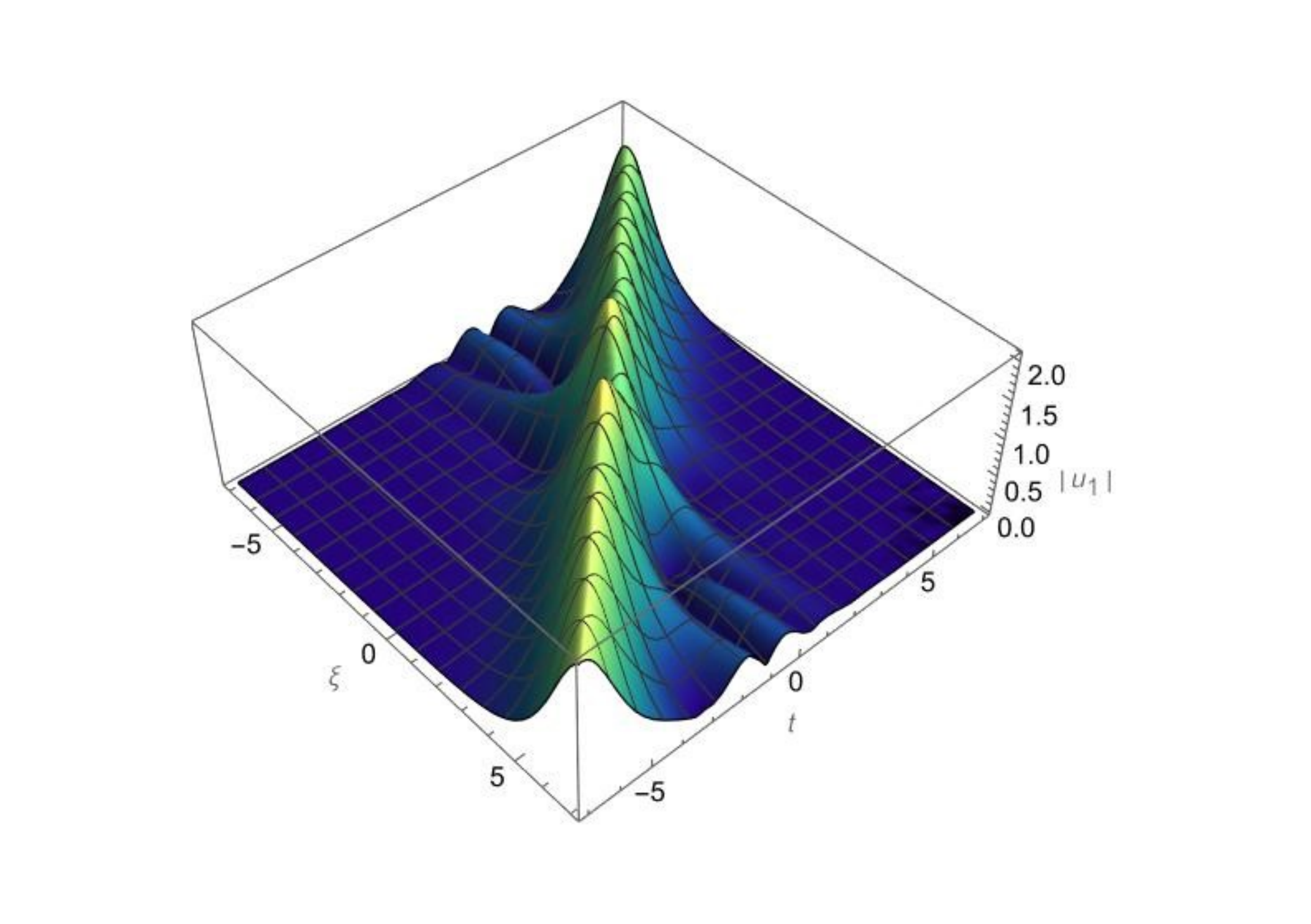}
\caption{} \label{fig:6d}
\end{subfigure}
\hfill
\hfill
\begin{subfigure}{0.45\textwidth}
\centering
\includegraphics[height=1.5in,width=1.5in]{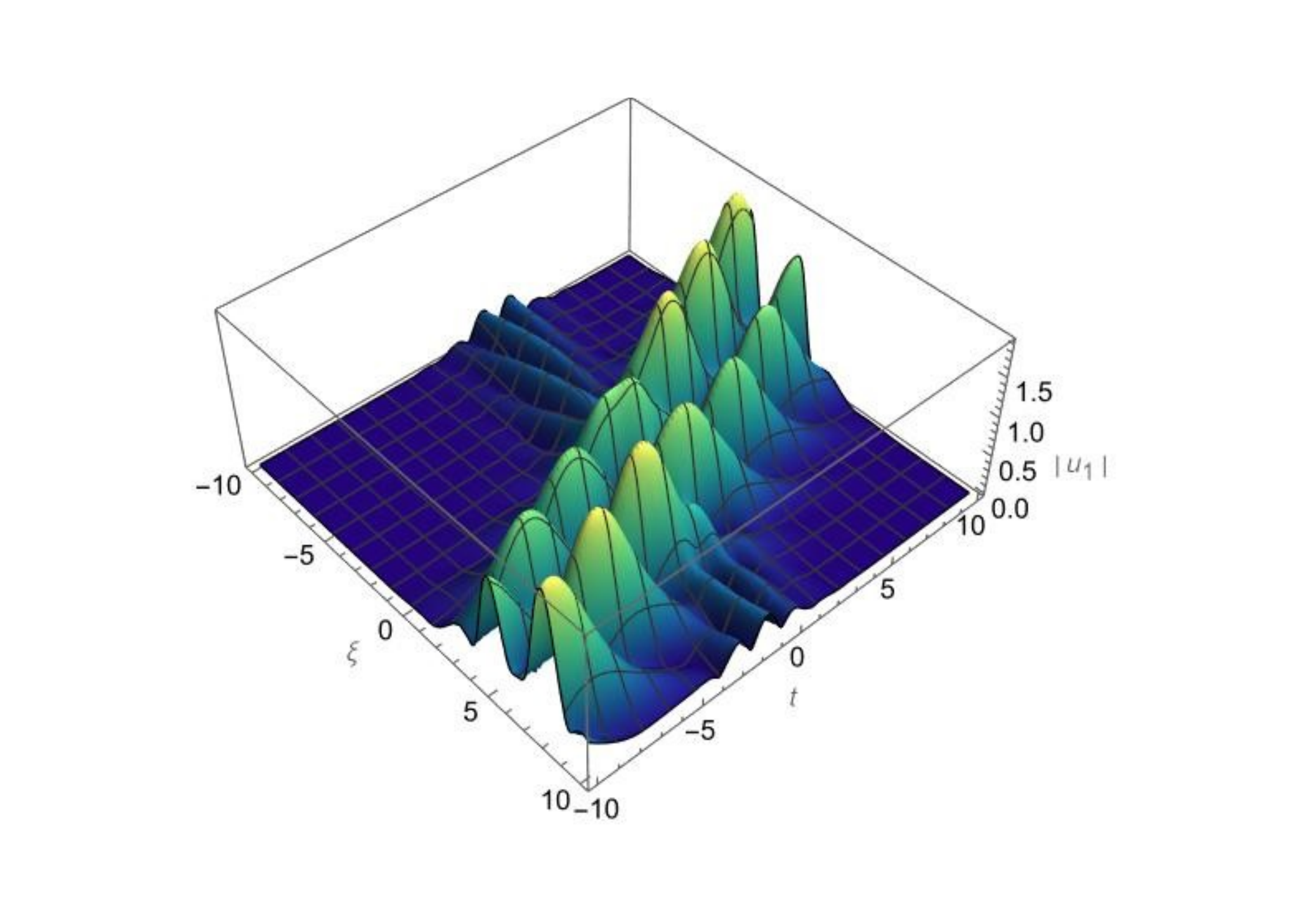}
\caption{} \label{fig:6f}
\end{subfigure}
\caption{Soliton interactions. All the panels show the magnitude of the first component of the vector solution. Panel (a): interaction between a fundamental soliton and a self-symmetric soliton. Panel (b): interaction between a fundamental breather and a self-symmetric soliton. Panel (c): interaction between a self-symmetric soliton and a generic composite breather. Panel (d): interaction between two generic composite breathers. Here, the soliton parameters are: panel (a): $\pmb{\gamma}_1 = \left( 1\,, 1/2\right)^{T}$, $\pmb{\delta}_1 = \left( 1\,, 0\right)^{T}$, $\pmb{\delta}_2 = \left( 1\,,1\right)^{T}$, $\pmb{\tau}_2 = \left( 1\,, 1\right)^{T}$, $\pmb{\gamma}_2 = \left( 10,2\right)^{T}$, $\pmb{\omega}_2 = \left( 2,-10\right)^{T}$ and $k_1=1+i/2$, $k_2=i/4$. Panel (b): $\pmb{\gamma}_1 = \left( 1\,,1/2 \right)^{T}$, $\pmb{\delta}_1 = \left( 1\,,1 \right)^{T}$, $\pmb{\delta}_2 = \left( 1\,,1\right)^{T}$, $\pmb{\tau}_2 = \left( 1\,, 1\right)^{T}$, $\pmb{\gamma}_2 = \left(10,2\right)^{T}$, $\pmb{\omega}_2 = \left( 2,-10\right)^{T}$ and $k_1=1+i/2$, $k_2=i/4$. Panel (c): $\pmb{\delta}_1 = \left( 1\,,1\right)^{T}$, $\pmb{\tau}_1 = \left( 1\,, 1\right)^{T}$, $\pmb{\gamma}_1 = \left( 1,2\right)^{T}$, $\pmb{\omega}_1 = \left( 2, -1\right)^{T}$, $\pmb{\delta}_2 = \left( 1\,,2\right)^{T}$, $\pmb{\tau}_2 = \left( 2\,, 1\right)^{T}$, $\pmb{\gamma}_2 = \left( 1,2\right)^{T}$, $\pmb{\omega}_2 = \left( 3, e^{i}\right)^{T}$ and $k_1=i/2$, $k_2=1+i/2$. Panel (d): $\pmb{\delta}_1 = \left( 1\,,2\right)^{T}$, $\pmb{\tau}_1 = \left( 2\,, 3\right)^{T}$, $\pmb{\gamma}_1 = \left( 10,2\right)^{T}$, $\pmb{\omega}_1 = \left( 2\,, -10\right)^{T}$, $\pmb{\delta}_2 = \left( 1\,,3\right)^{T}$, $\pmb{\tau}_2 = \left( 2\,, 3\right)^{T}$, $\pmb{\gamma}_2 = \left( 1,2\right)^{T}$, $\pmb{\omega}_2 = \left( 3\,, e^{i}\right)^{T}$ and $k_1=1/4+i/5$, $k_2=1+i/2$.}\label{fig:6}
\end{figure}

\subsection{Discussion of soliton interactions using the Yang-Baxter maps}

Thanks to the maps systematically derived above, the phenomenon observed in Fig.~\ref{fig:3} can now be explained. Looking at \eref{e:196n}, we note that the map on the polarizations decouples: the ${\pmb \gamma}$-type polarizations act on each other, and similarly for the ${\pmb \delta}$-type polarizations. From this point of view, the fundamental soliton sector is stable and the map reduces to a trivial map where the polarizations are simply rescaled upon interaction. This is why two fundamental solitons remain fundamental solitons when interacting with each other. However, in general a fundamental soliton interacting with a fundamental breather will become a fundamental breather after the interaction. This is a very interesting phenomenon which, to our knowledge, had never been identified before in studies on multicomponent soliton collisions. We discuss this in more detail below.

\paragraph{Interaction of two fundamental breathers.} A fundamental breather is obtained when both components in ${\pmb \delta}$ are non-zero. To study the interaction between two fundamental breathers, without loss of generality we can consider ${\pmb \delta}^{+}_1 \equiv {\pmb \delta}_1 = (1\,, \kappa^*_1)^{T}$, and ${\pmb \delta}^{-}_2 \equiv {\pmb \delta}_2 = (1\,, \kappa^*_2)^{T}$, $\kappa_j \in \mathbb{C}$, for $j=1,2$. The vectors ${\pmb \gamma}^{-}_1 \equiv {\pmb \gamma}_1$ and ${\pmb \gamma}^{+}_2 \equiv {\pmb \gamma}_2$ are chosen to be arbitrary, and after inserting these into \eref{e:196nc}-\eref{e:196nd} we obtain
\begin{subequations}\label{e:221new}
\begin{gather}
{\pmb \delta}^{-}_1 = \left(d_1\,,e_1\right)^{T},\quad {\pmb \delta}^{+}_2 =  \left(d_2\,,e_2\right)^{T},\\
d_1 = \frac{k_2}{k_2^*} \frac{k_1+k_2^*}{k_1+k_2} \left( 1 + \frac{k_1^2}{k_2^2} \frac{\left( (k_2^*)^2-k_2^2\right) \left( 1 + \kappa_1^* \kappa_2 \right)}{\left( k_1^2 - (k_2^*)^2 \right) \left(1 + |\kappa_2|^2\right)}\right),\\
e_1 = \frac{k_2}{k_2^*} \frac{k_1+k_2^*}{k_1+k_2} \left( \kappa_1^* + \frac{k_1^2}{k_2^2}
 \frac{\left( (k_2^*)^2 - k_2^2 \right) \left( \kappa_2^* + \kappa_1^* |\kappa_2|^2 \right)}{\left( k_1^2 - (k_2^*)^2 \right)\left(1+|\kappa_2|^2 \right)}\right),\\
d_2 = \frac{k_1}{k_1^*} \frac{k_2+k_1^*}{k_2+k_1} \left( 1 + \frac{k_2^2}{k_1^2} \frac{\left( (k_1^*)^2 - k_1^2 \right) \left( 1 + \kappa_1 \kappa_2^*\right)}{\left( k_2^2 - (k_1^*)^2 \right) \left( 1 + |\kappa_1|^2 \right)} \right),\\
e_2 =  \frac{k_1}{k_1^*} \frac{k_2+k_1^*}{k_2+k_1} \left( \kappa_2^* + \frac{k_2^2}{k_1^2} \frac{\left( (k_1^*)^2 - k_1^2 \right) \left( \kappa_1^* + \kappa_2^* | \kappa_1|^2 \right)}{\left( k_2^2 - (k_1^*)^2 \right) \left( 1 + |\kappa_1|^2 \right)} \right).
\end{gather}
\end{subequations}
To determine the total change in the first column of the norming constants of the solitons, we need to multiply relations \eref{e:196na} and \eref{e:196nb} by the scalars $d_1$ and $d_2$, respectively (cf. \eref{e:fundbreathera}). Therefore, we define
\begin{subequations}\label{e:221n}
\begin{gather}
{\bf s}^{+}_1 = d_1 \frac{k_2}{k_2^*}\frac{k_1+k_2^*}{k_1+k_2} \left[ I_2 + \frac{k_1^2}{k_2^2}\frac{(k_2^*)^2-k_2^2}{k_1^2 - (k_2^*)^2} \frac{{\pmb \gamma}_2 ({\pmb \gamma}_2)^\dagger}{({\pmb \gamma}_2)^\dagger {\pmb \gamma}_2} \right]{\pmb \gamma}_1, \label{e:221na}\\
{\bf s}^{-}_2 = d_2 \frac{k_1}{k_1^*}\frac{k_2 + k_1^*}{k_2 + k_1} \left[ I_2 + \frac{k_2^2}{k_1^2} \frac{(k_1^*)^2 - k_1^2}{k_2^2 - (k_1^*)^2} \frac{{\pmb \gamma}_1 ({\pmb \gamma}_1)^\dagger}{({\pmb \gamma}_1)^\dagger{\pmb \gamma}_1} \right] {\pmb \gamma}_2,
\end{gather}
\end{subequations}
and introducing the following notations for the normalization constants
\begin{gather}
\chi_1^2 \equiv \frac{|| \bf{s}^{+}_1||^2}{|| \pmb {\gamma}_1||^2}, \quad \chi_2^2 \equiv \frac{|| \bf {s}^{-}_2||^2}{|| \pmb {\gamma}_2||^2},
\end{gather}
we find
\begin{subequations}\label{e:224new}
\begin{gather}
\chi_1^2 = |d_1|^2 \chi^2, \quad \chi_2^2 = |d_2|^2 \chi^2, \\[2pt] \chi^2 = \left| \frac{k_1+k_2^*}{k_1+k_2}\right|^2 \left\{ 1 + \frac{\left( k_1^2 - (k_1^*)^2 \right) \left( k_2^2 - (k_2^*)^2\right)}{|k_1-k_2^*|^2 |k_1+k_2^*|^2} \left| \bf {{p}}_1^{-} \cdot (\bf {{p}}_2^{+})^* \right|^2\right\},
\end{gather}
\end{subequations}
where $\bf{{p}}_1^{-}=\pmb{\gamma}_1^*/||\pmb{\gamma}_1||$ and $\bf{{p}}_2^{+}=\pmb{\gamma}_2^*/||\pmb{\gamma}_2||$ are unit-norm polarization vectors. Furthermore, the polarization vectors after the interaction are given by:
\begin{subequations}\label{e:219nw}
\begin{gather}
{\bf {p}}^{+}_1 = \frac{d_1^*}{\chi_1} \frac{k_2^*}{k_2}\frac{k_1^*+k_2}{k_1^*+k_2^*} \left\{ {\bf {{p}}_1^{-}} + \frac{(k_1^*)^2}{(k_2^*)^2}\frac{\left( k_2^2-(k_2^*)^2\right)}{\left((k_1^*)^2 - k_2^2\right)} \left( ({\bf {{p}}_2^{+}})^* \cdot {\bf {{p}}_1^{-}} \right) {\bf {{p}}_2^{+}} \right\},\\
{\bf {p}}^{-}_2 = \frac{d_2^*}{\chi_2} \frac{k_1^*}{k_1}\frac{k_2^* + k_1}{k_2^* + k_1^*} \left\{  {\bf {{p}}_2^{+}} + \frac{(k_2^*)^2}{(k_1^*)^2} \frac{\left( k_1^2 - (k_1^*)^2\right)}{\left( (k_2^*)^2 - k_1^2\right)} \left( ({\bf {{p}}_1^{-}})^* \cdot {\bf {{p}}_2^{+}} \right) {\bf {{p}}_1^{-}} \right\},
\end{gather}
\end{subequations}
in terms of parameters characterizing the initial fundamental breathers, where ${\bf {p}}^{+}_1=({\bf s}^{+}_1)^*/||{\bf s}^{+}_1||$ and ${\bf {p}}^{-}_2=({\bf s}^{-}_2)^*/||{\bf s}^{-}_2||$.

Generically speaking, the quantities $e_1$ and $e_2$ in Eqs. \eref{e:221new} are nonzero, which implies that the two-soliton solution is a superposition of two fundamental breathers which retain their nature throughout their interaction. However, there exists a specific value of the multiplicative constant $\kappa_1$, which involves the discrete eigenvalues and the second multiplicative constant $\kappa_2$, and which makes the quantity $e_1$ identically zero
\begin{gather}\label{e:241}
\kappa_1 = \frac{(k_1^*)^2 \left( (k_2^*)^2 - k_2^2 \right) \kappa_2}{(k_2^*)^2 \left( (k_1^*)^2 - k_2^2 \right) + k_2^2 \left( (k_1^*)^2 - (k_2^*)^2 \right) |\kappa_2|^2}.
\end{gather}
In this case, soliton-2 retains its nature, while soliton-1 becomes a fundamental soliton after the interaction. Equivalently, there exists a specific value for the multiplicative constant $\kappa_2$ which makes the quantity $e_2$ identically zero, its expression is given by simply interchanging the indices 1 and 2 in relation \eref{e:241}, and in this case soliton-1 retains its nature, and soliton-2 becomes a breather. Fig.~\ref{fig:8} shows an example of interacting fundamental breathers when the multiplicative constant $\kappa_2$ is arbitrarily assigned, but $\kappa_1$ has been computed via relation \eref{e:241}, and an example of interacting fundamental breathers where both multiplicative constants are arbitrarily assigned.

\begin{figure}[h!]
\centering
\begin{subfigure}{0.48\textwidth}
\centering
\includegraphics[height=1.5in,width=1.5in]{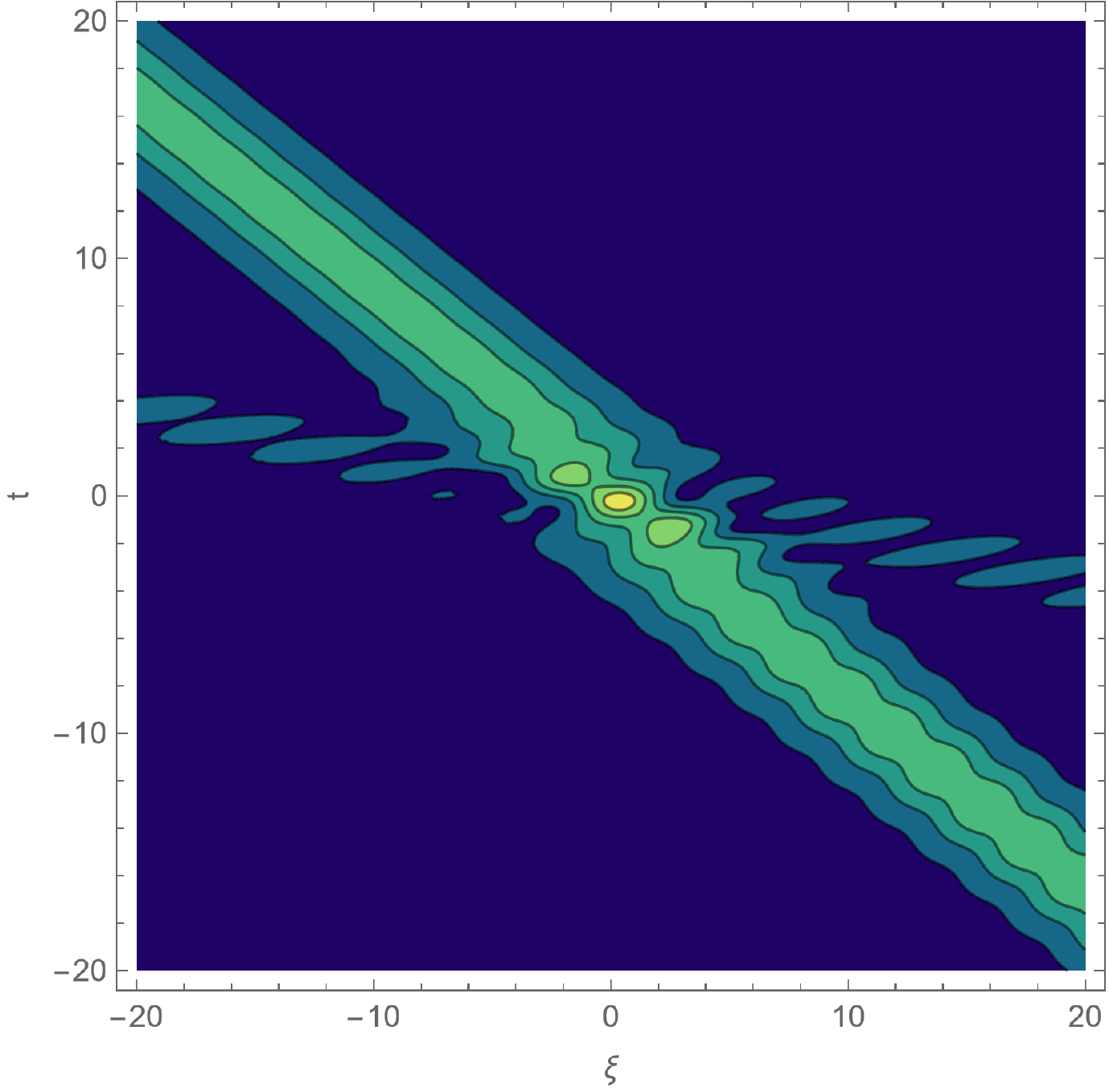}
\caption{} \label{fig:8a}
\end{subfigure}
\hspace*{\fill}
\hfill
\begin{subfigure}{0.48\textwidth}
\centering
\includegraphics[height=1.5in,width=1.5in]{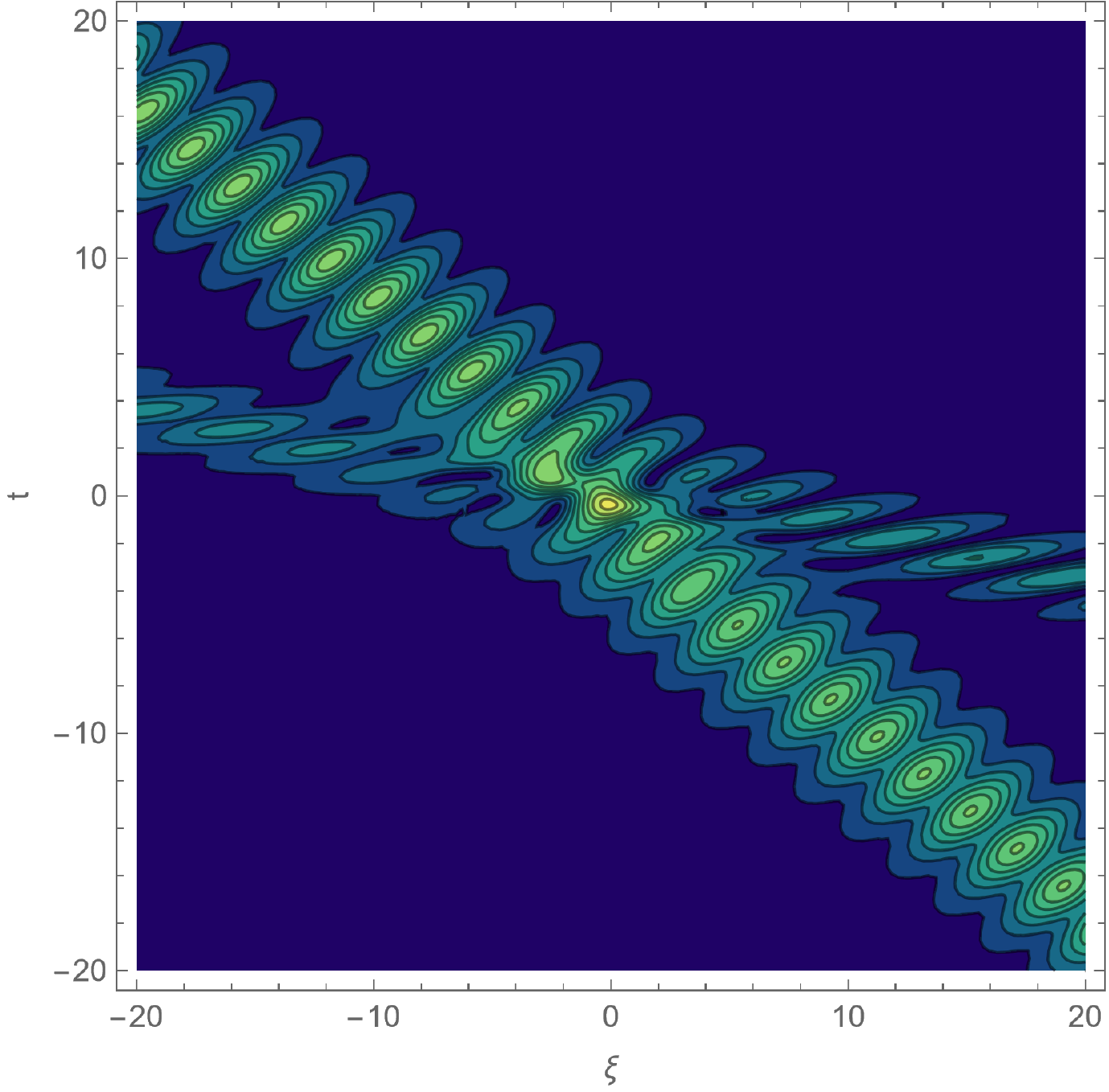}
\caption{} \label{fig:8b}
\end{subfigure}
\hfill
\caption{Fundamental breathers interacting with fundamental breathers. Panel (a): the multiplicative constant $\kappa_2$ is arbitrarily assigned, while $\kappa_1$ is computed via relation \eref{e:241}. Clearly, the initial fundamental breather with multiplicative constant $\kappa_2$ retains its nature, but the initial fundamental breather with multiplicative constant $\kappa_1$ becomes a fundamental soliton after the interaction. Panel (b): both multiplicative constants $\kappa_1$ and $\kappa_2$ have been arbitrarily assigned, and one can observe that both initial fundamental breathers retain their nature upon interacting with each other. The panels show the magnitude of the first component of the vector solution. Here, the soliton parameters are: panel (a): $\pmb{\gamma}_1 = \left( 1\,, 1/2\right)^{T}$, $\pmb{\delta}_1 = \left( 0.18+0.1 i\,, 1\right)^{T}$, $\pmb{\gamma}_2 = \left(1\,,e^{i}\right)^{T}$, $\pmb{\delta}_2 = \left(1\,,1\right)^{T}$, and $k_1=1/2+i/4$, $k_2=1+i/2$. Panel (b): $\pmb{\gamma}_1 = \left(1\,,1/2 \right)^{T}$, $\pmb{\delta}_1 = \left(1\,,1\right)^{T}$, $\pmb{\gamma}_2 = \left( 1\,,e^{i}\right)^{T}$, $\pmb{\delta}_2 = \left( 1\,,1\right)^{T}$, $k_1=1/2+i/4$, $k_2=1+i/2$.}\label{fig:8}
\end{figure}

\paragraph{Special cases.} Next we show that interactions between two fundamental solitons, as well as interactions between a fundamental soliton and a fundamental breather, are special cases of interactions between two fundamental breathers. For two fundamental solitons, one has to set $\kappa_j=0$, for $j=1,2$, and therefore the initial vectors ${\pmb \delta}_j$ can be reduced to ${\pmb \delta}_1 = (1\,, 0)^{T}$, and ${\pmb \delta}_2 = (1\,, 0)^{T}$. It is easy to check from Eqs.~\eref{e:221new} that, if $\kappa_1=\kappa_2=0$, then $e_1=e_2 \equiv 0$, which implies that the two fundamental solitons retain their nature upon interacting with each other. Moreover, the scalars $d_1, d_2$ reduce to:
\begin{gather}
d_1 = \frac{k_2^*}{k_2} \frac{k_1-k_2}{k_1-k_2^*}, \quad d_2 = \frac{k_1^*}{k_1}\frac{k_1-k_2}{k_1^*-k_2},
\end{gather}
which gives $|d_1|= |d_2|$. Consequently, from Eqs.~\eref{e:224new} it follows that the two normalization constants coincide, namely, $\chi_1=\chi_2=\chi$ with
$\chi$ given by \eref{e:38}, and Eqs.~\eref{e:219nw} reduce to the expressions \eref{e:112n} we derived using Manakov's method in Sec.~5. Note that one could alternatively choose ${\pmb \delta}_1=(0\,,1)^{T}$ and ${\pmb \delta}_2=(0\,,1)^{T}$, or  ${\pmb \delta}_1=(1\,,0)^{T}$ and ${\pmb \delta}_2=(0\,,1)^{T}$ and ${\pmb \gamma}$ arbitrary, and we still obtain that the two fundamental solitons retain their nature upon interacting with each other.

Without loss of generality, let us now consider soliton-1 to be a fundamental soliton and soliton-2 to be a fundamental breather. We therefore set $\kappa_1=0$ and let $\kappa_2$ be arbitrary. Then, ${\pmb \delta}_1 =(1\,,0)^{T}$, ${\pmb \delta}_2 = (1\,, \kappa^*)^{T}$, $\kappa \in \mathbb{C}$, and the vectors ${\pmb \gamma}_1$ and ${\pmb \gamma}_2$ are arbitrary. Using Eqs.~\eref{e:221new}, it is easy to check that, if $\kappa_1=0$, then the scalars $e_1, e_2$ are both nonzero and reduce to:
\begin{gather}
e_1 = \frac{\kappa^*}{1 + |\kappa|^2} \frac{k_1^2}{|k_2|^2}\frac{(k_2^*)^2-k_2^2}{(k_1+k_2)(k_1-k_2^*)}, \quad e_2 = \kappa^* \frac{k_1}{k_1^*}\frac{k_1^* + k_2}{k_1 + k_2},
\end{gather}
which shows that although the solution before the interaction is a superposition of a fundamental soliton and a fundamental breather, after the interaction it becomes a superposition of two fundamental breathers, since necessarily $e_1\ne 0$ and $e_2\ne0$. Moreover, the polarization vectors after the interaction are given by Eqs.~\eref{e:219nw}, and the two normalization constants are still given by Eqs.~\eref{e:224new}, but now with $d_1, d_2$:
\begin{gather}\label{e:scald1d2}
d_1 = \frac{k_2}{k_2^*}\frac{k_1 + k_2^*}{k_1 + k_2} \left[ 1 + \frac{k_1^2}{k_2^2}\frac{(k_2^*)^2 - k_2^2}{\left(k_1^2 - (k_2^*)^2\right)\left(1 + |\kappa|^2\right)}\right], \quad
d_2 = \frac{k_1^*}{k_1}\frac{k_1-k_2}{k_1^* - k_2}.
\end{gather}

\section*{Acknowledgements}

BP and AG gratefully acknowledge partial support for this work from the NSF, under grant DMS-2106488.
The authors would also like to thank the Isaac Newton Institute for Mathematical Sciences (Cambridge, UK) for support and hospitality during the programme ``Dispersive Hydrodynamics: mathematics, simulation and experiments, with applications in nonlinear waves'', when work on this paper was undertaken (EPSRC Grant Number EP/R014604/1).

\appendix

\numberwithin{equation}{section}

\section{On the ranks of the transmission coefficients $a$ and $\bar{a}$} \label{appendix:ranks}

In this section, we show that the ranks of the matrices $a$ and $\bar{a}$ need necessarily to be equal at each of the eigenvalues of a given quartet. To do that, we recall the $4 \times 4$ matrix solutions
\begin{gather}\label{e:181}
P(k) = \left( \Phi_{-,1}(k), \Phi_{+,2}(k) \right), \quad \bar{P}(k) = \left( \Phi_{+,1}(k), \Phi_{-,2}(k) \right),
\end{gather}
where the $(x,t)$ dependence has been omitted for brevity, and the bilinear combinations
\begin{gather}
\label{e:AAb}
A(k) = \bar{P}^\dagger(x,t,k^*)P(x,t,k), \quad \bar{A}(k) = P^\dagger(x,t,k^*) \bar{P}(x,t,k),
\end{gather}
introduced in \cite{ABBA}. Using the analyticity properties of the Jost eigenfunctions, one can easily show that the solution $P$ and the bilinear combination $A$ are analytic in $\mathbb{C}^{+}$, while the matrix solution $\bar{P}$ and the bilinear combination $\bar{A}$ are analytic in $\mathbb{C}^{-}$. Next, we discuss in details the correspondence between the ranks of the matrix $P$ and the linear combination $A$ (resp., $\bar{P}$ and $\bar{A}$), which will lead us to the goal of this section.

Let us assume $k_n$ to be a zero of $\det a(k)$ in $\mathbb{C}^{+}$. Then, due to the symmetries of the ccSPE summarized in Sec.~2, $k_n^*$ is a zero of $\det \bar{a}(k)$ in $\mathbb{C}^{-}$. Therefore, it is meaningful to discuss the ranks of the matrices $P(k_n)$ and $\bar{P}(k_n^*)$. It is clear that $\mathrm{rank}P(k_n)\le 3$ and $\mathrm{rank}\bar{P}(k_n^*)\le 3$, since both matrices have zero determinant at the discrete eigenvalues. It is also clear that $\mathrm{rank}P(k_n)$ and $\mathrm{rank}\bar{P}(k_n^*)$ cannot be zero, because the $4 \times 2$ columns in $\Phi_{-,1}(k_n)$ and in $\Phi_{+,2}(k_n)$ are linearly independent, and so are the columns in $\Phi_{+,1}(k_n^*)$ and in $\Phi_{-,2}(k_n^*)$.

One can also prove that $\mathrm{rank}P(k_n)$ and $\mathrm{rank}\bar{P}(k_n^*)$ cannot equal $1$, by contradiction. Let us assume that $\mathrm{rank}P(k_n) =1$. Then the columns of $\Phi_{-,1}(k_n)$ must be linearly dependent, as well as the columns of $\Phi_{+,2}(k_n)$. Let us denote these columns as $\Phi_{-,1}(k_n)\big{|}_{j}$ and $\Phi_{+,2}(k_n)\big{|}_{j}$ respectively, for $j=1,2$.  The above implies that there exist constants $c_1, c_2 \in \mathbb{C}$ not both zero, as well as constants $c_3, c_4 \in \mathbb{C}$ not both zero, such that
\begin{subequations}\label{e:1}
\begin{gather}
c_1 \Phi_{-,1}\big{|}_{1}(x,t,k_n) + c_2 \Phi_{-,1}\big{|}_{2}(x,t,k_n) = 0_{2 \times 1},\\
c_3 \Phi_{+,2}\big{|}_{1}(x,t,k_n) + c_4 \Phi_{+,2}\big{|}_{2}(x,t,k_n) = 0_{2 \times 1},
\end{gather}
\end{subequations}
for all $x \in \mathbb{R}$. Recall that the Jost eigenfunctions and the modified eigenfunctions are related according to \eref{e:defM},
and therefore the last relation becomes
\begin{subequations}\label{e:3}
\begin{gather}
c_1 M_{-,1}\big{|}_{1}(x,t,k_n) + c_2 M_{-,1}\big{|}_{2}(x,t,k_n) = 0_{2 \times 1},\\
c_3 M_{+,2}\big{|}_{1}(x,t,k_n) + c_4 M_{+,2}\big{|}_{2}(x,t,k_n) = 0_{2 \times 1},
\end{gather}
\end{subequations}
for all $x \in \mathbb{R}$. Moreover, recall that
\begin{gather*}
M_{-,1}(x,t,k) \sim \begin{pmatrix}
I_2 \\
0
\end{pmatrix}, \quad x \to - \infty, \quad \quad M_{+,2}(x,t,k) \sim \begin{pmatrix}
0 \\
I_2
\end{pmatrix}, \quad x \to + \infty,
\end{gather*}
and both $M_{-,1}$ and $M_{+,2}$ are analytic in the upper half plane of $\mathbb{C}$, which implies that as $x \to \pm \infty$, the columns of $M_{-,1}$ and the columns of $M_{+,2}$ are linearly independent. Therefore, relation \eref{e:3} does not hold, and \eref{e:1} does not hold as a consequence, which implies that the columns of $\Phi_{-,1}$ and $\Phi_{+,2}$ are linearly independent. This contradicts our initial assumption that $\mathrm{rank}P(k_n)=1$. One can use a similar argument to prove that $\mathrm{rank}\bar{P}(k_n^*)$ cannot equal 1. The above analysis allows us to conclude that the ranks of the matrices $P(k_n)$ and $\bar{P}(k_n^*)$ can only equal $2$ or $3$.

For the matrices $A(k)$ and $\bar{A}(k)$ in \eref{e:AAb}, as shown in \cite{ABBA} one has
\begin{subequations}\label{e:185}
\begin{gather}
A(k) = \bar{P}^\dagger(x,t,k^*)P(x,t,k) = \begin{pmatrix}
a(k) & 0_2 \\
0_2 & \bar{a}^\dagger(k^*)
\end{pmatrix},\\
\bar{A}(k) = P^\dagger(x,t,k^*)\bar{P}(x,t,k) = \begin{pmatrix}
a^\dagger(k^*) & 0_2 \\
0_2 & \bar{a}(k)
\end{pmatrix},
\end{gather}
\end{subequations}
and, consequently,
\begin{subequations}\label{e:1new}
\begin{gather}
\mathrm{rank} A(k_n) = \mathrm{rank} a(k_n) + \mathrm{rank} \bar{a}^\dagger(k_n^*), \label{e:1newa}\\
 \mathrm{rank} \bar{A}(k_n) = \mathrm{rank} a^\dagger (k_n^*) + \mathrm{rank} \bar{a}(k_n).\label{e:1newb}
\end{gather}
\end{subequations}
\begin{prop}
\label{e:propA}
Let $k_n$ be a zero of $\det a(k)$ in $\mathbb{C}^{+}$, $k_n^*$ is a zero of $\det \bar{a}(k)$ in $\mathbb{C}^{-}$. Then one has $\mathrm{rank}P(x,t,k_n) \equiv  \mathrm{rank}\bar{P}(x,t,k_n^*) = 2$ if and only if $a(k_n)\equiv \bar{a}(k_n^*) = 0_2$ (i.e., $\mathrm{rank}A(k_n)=0$). The same holds for the remaining points $-k_n^*$ and $-k_n$ in the discrete spectrum.
\end{prop}
\begin{proof}
Let us prove the first implication, i.e.,
\begin{gather*}
\mathrm{rank}P(x,t,k_n) \equiv  \mathrm{rank}\bar{P}(x,t,k_n^*) = 2 \Rightarrow a(k_n)\equiv \bar{a}(k_n^*) = 0_2.
\end{gather*}
Since $\mathrm{rank}P(x,t,k_n) = 2$, then there exist two linearly independent vectors $\pmb{e}_j = \left( \eta_j, -\xi_j\right)^T$ in $\mathbb{C}^4 / \{0\}$, such that $P(k_n)\pmb{e}_j = 0_{4 \times 1}$, for $j=1,2$. From the definition of $A(k_n)$, it follows that $A(k_n) \pmb{e}_j = \bar{P}^\dagger(k_n^*)P(k_n) \pmb{e}_j = 0_{4 \times 1}$,
which in turn gives
 \begin{gather}\label{eq:7}
a(k_n)\eta_j = 0_{2 \times 1}, \quad \bar{a}^\dagger (k_n^*) \xi_j = 0_{2 \times 1}.
\end{gather}
If the vectors $\eta_1$ and $\eta_2$ are both linearly independent, then the first of Eqs.~\eref{eq:7} implies that the dimension of the kernel of $a(k_n)$ is $2$, and therefore $a(k_n)=0_2$. If the vectors $\xi_1$ and $\xi_2$ are linearly independent, then the second of Eqs.~\eref{eq:7} yields that the dimension of the kernel of $\bar{a}(k_n^*)$ is $2$, and therefore $\bar{a}(k_n^*)=0_2$, which proves the desired result.

Let us now assume that $\eta_1$ and $\eta_2$ are linearly independent and $\xi_1$ and $\xi_2$ are linearly dependent. Then, there exists a complex valued constant $\beta$ such that $\xi_2 = \beta \xi_1$. From the definition of $P(k_n)$, we get
\begin{gather*}
\Phi_{-,1}(x,t,k_n)\eta_1 = \Phi_{+,2}(x,t,k_n)\xi_1,\quad \Phi_{-,1}(x,t,k_n)\eta_2 = \beta \Phi_{+,2}(x,t,k_n)\xi_1,
\end{gather*}
which implies
\begin{gather*}
\Phi_{-,1}(x,t,k_n)\eta_2 = \Phi_{-,1}(x,t,k_n)\left( \beta \eta_1 \right) \Leftrightarrow \Phi_{-,1}(x,t,k_n) \left( \eta_2 - \beta \eta_1 \right) = 0_{4\times 1},
\end{gather*}
and $\eta_2 = \beta \eta_1$, since the columns of $\Phi_{-,1}(x,t,k_n)$ are linearly independent. But the last one implies that the vectors $\eta_1$ and $\eta_2$ are linearly dependent, which contradicts our initial assumption. If the vectors $\eta_1$ and $\eta_2$ are linearly dependent, i.e. if there exists some $\alpha \in \mathbb{C}$ such that $\eta_2 = \alpha \eta_1$, then similarly as before, we have two possibilities for the vectors $\xi_1$ and $\xi_2$. If $\xi_1$ and $\xi_2$ are linearly independent, as before that will lead to a contradiction. Now, if both $\{ \eta_1, \eta_2\}$ and $\{ \xi_1, \xi_2\}$ are linearly dependent, then we can find two constants $\alpha, \beta \in \mathbb{C}$ such that $\eta_2 = \alpha \eta_1$ and $\xi_2 = \beta \xi_1$. In this case, we get
\begin{gather*}
\Phi_{-,1}(x,t,k_n)\eta_1 = \Phi_{+,2}(x,t,k_n)\xi_1,\quad \alpha \Phi_{-,1}(x,t,k_n)\eta_1 = \beta \Phi_{+,2}(x,t,k_n)\xi_1,
\end{gather*}
which gives
\begin{gather*}
\Phi_{+,2}(x,t,k_n)\left( \alpha \xi_1 \right) = \Phi_{+,2}(x,t,k_n)\left( \beta \xi_1\right) \Leftrightarrow \Phi_{+,2}(x,t,k_n)(\alpha - \beta) \xi_1 = 0_{4 \times 1},
\end{gather*}
which in turn implies that $\alpha = \beta$, since the columns of $\Phi_{+,2}(x,t,k_n)$ are linearly independent. But if $\alpha = \beta$, then the vectors $\pmb{e}_1 = \left( \eta_1, -\xi_1 \right)^{T}$ and $\pmb{e}_2 = \left( \alpha \eta_1, - \beta \xi_1 \right)^{T}$ are linearly dependent, which contradicts our initial assumption.

Let us now prove the other implication, i.e., let us assume that $a(k_n) \equiv \bar{a}^\dagger (k_n^*) = 0_2$ and prove that $\mathrm{rank}P(k_n) \equiv \mathrm{rank} \bar{P}(k_n^*) = 2$. If $a(k_n) \equiv \bar{a}^\dagger (k_n^*) = 0_2$, then both matrices have rank equal to zero, and therefore from \eref{e:1new} it follows $\mathrm{rank}A(k_n) = 0$. Using Sylvester's inequality, we obtain
\begin{gather}
0 = \mathrm{rank}A(k_n) \geq \mathrm{rank} P(x,t,k_n) + \mathrm{rank}\bar{P}^\dagger (x,t,k_n^*)- 4.
\end{gather}
Since $\mathrm{rank} P(x,t,k_n)$ and $\mathrm{rank}\bar{P}^\dagger (x,t,k_n^*)$ can be either $2$ or $3$, then for the last inequality to hold, necessarily
\begin{gather*}
\mathrm{rank} P(x,t,k_n) \equiv \mathrm{rank}\bar{P}^\dagger (x,t,k_n^*) = 2.
\end{gather*}
\qed
\end{proof}
Prop.~\ref{e:propA} establishes a one-to-one correspondence between $\mathrm{rank}P(k_n)$ and $\mathrm{rank}A(k_n)$, when $\mathrm{rank}A(k_n)=0$. Here, we discuss the correspondence between the ranks of $P(k_n)$ and $A(k_n)$, when $\mathrm{rank}A(k_n)$ is either $1$ or $2$, and as we will see, together with the previous results, this shows that the matrices $a(k_n)$ and $\bar{a}(k_n^*)$ need necessarily have equal ranks. From Sylvester's inequality, we get
\begin{gather*}
\mathrm{rank}A(k_n) = \mathrm{rank}\left(\bar{P}^\dagger(k_n^*)P(k_n) \right) \geq \mathrm{rank}(\bar{P}^\dagger(k_n^*)) + \mathrm{rank}(P(k_n)) - 4.
\end{gather*}
Let us assume $\mathrm{rank}A(k_n) = 2$, which implies that both matrices $a(k_n)$ and $\bar{a}(k_n^*)$ have rank equal to $1$. We then have the following cases:
\begin{enumerate}
\item[(1a)] $\mathrm{rank}P(k_n) \equiv \mathrm{rank} \bar{P}^\dagger(k_n^*) = 2$. In this case, from Sylvester's inequality we get $2 \geq 2+2-4=0$, which is admissible. However, from Prop.~\ref{e:propA} we have that $\mathrm{rank} P(k_n) = 2$ implies $\mathrm{rank}A(k_n) = 0$, which contradicts the initial assumption.
\item[(1b)] $\mathrm{rank}P(k_n) \equiv \mathrm{rank} \bar{P}^\dagger(k_n^*) = 3$. In this case, Sylvester's inequality yields $2 \geq 3+3-4=2$, which is admissible.
\item[(1c)] $\mathrm{rank}P(k_n)=2$ and $\mathrm{rank}\bar{P}^\dagger(k_n^*) = 3$. In this case, from Sylvester's inequality we get $2 \geq 3+2-4=1$, which is admissible. However, according to Prop.~\ref{e:propA}, $\mathrm{rank} P(k_n) = 2$ implies $\mathrm{rank}A(k_n) = 0$, which contradicts the initial assumption.
\end{enumerate}
Let us now assume $\mathrm{rank}A(k_n) = 1$, which implies that one of the matrices $a(k_n)$ and $\bar{a}(k_n^*)$ has rank equal to $0$ and the other one has rank equal to $1$. We then have the following cases:
\begin{enumerate}
\item[(2a)] $\mathrm{rank}P(k_n) \equiv \mathrm{rank}\bar{P}^\dagger(k_n^*) = 2$. In this case, Sylvester's inequality yields $1 \geq 2+2-4=0$, which is admissible. However, Prop.~\ref{e:propA} has $\mathrm{rank}P(k_n) = 2$ implying $\mathrm{rank}A(k_n) = 0$, which contradicts the initial assumption.
\item[(2b)] $\mathrm{rank}P(k_n) \equiv \mathrm{rank} \bar{P}^\dagger(k_n^*) = 3$. In this case, from Sylvester's inequality we get $1 \geq 3+3-4=2$, which is not admissible.
\item[(2c)] $\mathrm{rank}P(k_n)=2$ and $\mathrm{rank} \bar{P}^\dagger(k_n^*) = 3$ (or vice-versa). In either case, Sylvester's inequality gives $1 \geq 3+2-4=1$, which is admissible. However, Prop.~\ref{e:propA} gives $\mathrm{rank}P(k_n)= 2$ implies $\mathrm{rank}A(k_n) = 0$, which contradicts the initial assumption.
\end{enumerate}
In conclusion, there is no admissible correspondence between the ranks of the matrices $P(k_n)$ and $A(k_n)$ when $\mathrm{rank}A(k_n)=1$, i.e., when $a(k_n)$ and $\bar{a}(k_n^*)$ have unequal ranks. Therefore, the matrices $a(k_n)$ and $\bar{a}(k_n^*)$ need necessarily have equal ranks. Similar arguments would give the result for the other two discrete eigenvalues in each quartet.

\section{Derivation of one-soliton transmission coefficients}

In the following, we derive the expression of the transmission coefficients for 1-fundamental soliton, 1-fundamental breather, and 1-self-symmetric soliton.

\subsection{Transmission coefficients for a 1-fundamental soliton solution}

We start by considering a 1-fundamental soliton solution. Assuming that the potential decays rapidly enough so that Eqs.~\eref{e:28} and ~\eref{e:7new} can be extended slightly above or below the real $k$-axis, one can use these equations to reconstruct the left and right transmission coefficients as limits of suitable blocks of the modified eigenfunctions, namely
\begin{subequations}\label{e:7}
\begin{gather}
a(k) = \lim_{x \to + \infty} M_{-,1}^{\mathrm{up}}(x,t,k)\,, \qquad \bar{a}(k) = \lim_{x \to + \infty} M_{-,2}^{\mathrm{dn}}(x,t,k)\,,\label{e:7a}\\
\bar{c}(k) = \lim_{x \to - \infty} M_{+,1}^{\mathrm{up}}(x,t,k)\,, \qquad c(k) = \lim_{x \to - \infty} M_{+,2}^{\mathrm{dn}}(x,t,k)\,.\label{e:7b}
\end{gather}
\end{subequations}
From the definition of the sectionally meromorphic matrix function $\mu_{\pm}$ in \eref{e:jumpcondition}, one has
\begin{gather}
\label{e:13}
\bar{c}(k) = \lim_{x \to - \infty}\mu^{\mathrm{up}}_{-,1}(x,t,k), \qquad c(k) = \lim_{x \to - \infty}\mu^{\mathrm{dn}}_{+,2}(x,t,k),
\end{gather}
where superscripts $^{\mathrm{up}}$ and $^{\mathrm{dn}}$ denote the upper/lower $2 \times 2$ blocks of the corresponding matrices.
In Sec. 2, we also introduced the sectionally meromorphic matrix function
\begin{gather}
\label{e:15new}
\tilde{\mu}_{\pm}(x,t,k) = \mu_{\infty}^{-1}(x,t)\mu_{\pm}(x,t,k),
\end{gather}
where $\mu_{\infty} = \lim_{|k| \to \infty}\mu_{\pm}(x,t,k)$, such that $\tilde{\mu}_{\pm}$ is normalized, as $k \to \infty$, and using a similar notation for the upper/lower blocks of the functions $\mu_{\infty}$ and $\tilde{\mu}_{\pm}$, we obtain
\begin{subequations}\label{e:16new}
\begin{gather}
\mu_{-,1}^{\mathrm{up}}(x,t,k) = \mu_{\infty,1}^{\mathrm{up}}(x,t)\,\tilde{\mu}_{-,1}^{\mathrm{up}}(x,t,k) + \mu_{\infty,2}^{\mathrm{up}}(x,t)\,\tilde{\mu}_{-,1}^{\mathrm{dn}}(x,t,k)\,,\label{e:16newa}\\
\mu_{+,2}^{\mathrm{dn}}(x,t,k) = \mu_{\infty,1}^{\mathrm{dn}}(x,t)\,\tilde{\mu}_{+,2}^{\mathrm{up}}(x,t,k) + \mu_{\infty,2}^{\mathrm{dn}}(x,t)\,\tilde{\mu}_{+,2}^{\mathrm{dn}}(x,t,k),\label{e:16newb}
\end{gather}
\end{subequations}
which, together with \eref{e:13}, yield
\begin{subequations}\label{e:17}
\begin{gather}
\bar{c}(k) = \lim_{x \to - \infty} \left\{ \mu_{\infty,1}^{\mathrm{up}}(x,t)\,\tilde{\mu}_{-,1}^{\mathrm{up}}(x,t,k) + \mu_{\infty,2}^{\mathrm{up}}(x,t)\,\tilde{\mu}_{-,1}^{\mathrm{dn}}(x,t,k)\right\}
\,,\label{e:17a}\\
c(k) = \lim_{x \to - \infty} \left\{ \mu_{\infty,1}^{\mathrm{dn}}(x,t)\,\tilde{\mu}_{+,2}^{\mathrm{up}}(x,t,k) + \mu_{\infty,2}^{\mathrm{dn}}(x,t)\,\tilde{\mu}_{+,2}^{\mathrm{dn}}(x,t,k) \right\}\,,\label{e:17b}
\end{gather}
\end{subequations}
or equivalently
\begin{subequations}\label{e:18}
\begin{gather}
\bar{c}(k) = \lim_{\xi \to - \infty} \left\{ \hat{\mu}_{\infty,1}^{\mathrm{up}}(\xi,t)\,\breve{\mu}_{-,1}^{\mathrm{up}}(\xi,t,k) + \hat{\mu}_{\infty,2}^{\mathrm{up}}(\xi,t)\,\breve{\mu}_{-,1}^{\mathrm{dn}}(\xi,t,k)\right\}
\,,\label{e:18a}\\
c(k) = \lim_{\xi \to - \infty} \left\{ \hat{\mu}_{\infty,1}^{\mathrm{dn}}(\xi,t)\,\breve{\mu}_{+,2}^{\mathrm{up}}(\xi,t,k) + \hat{\mu}_{\infty,2}^{\mathrm{dn}}(\xi,t)\,\breve{\mu}_{+,2}^{\mathrm{dn}}(\xi,t,k) \right\}\,,\label{e:18b}
\end{gather}
\end{subequations}
where the $x$-dependence has been replaced by the $\xi$-dependence. The upper/lower blocks of $\breve{\mu}_{-,1}$ and $\breve{\mu}_{+,2}$ are obtained by solving the linear system (28) and are given by:
\begin{subequations}\label{e:B.7}
\begin{gather}
\breve{\mu}_{-,1}^{\mathrm{up}}(\xi,t,k) = I_2 + \frac{e^{2i \hat{\theta}_1}}{k-k_1} \breve{\mu}_{+,2}^{\mathrm{up}}(\xi,t,k_1) C_1 + \frac{e^{-2i \hat{\theta}_1^*}}{k + k_1^*} \breve{\mu}_{+,2}^{\mathrm{up}}(\xi,t,-k_1^*)\tilde{C}_1, \label{e:B.7a}\\
\breve{\mu}_{-,1}^{\mathrm{dn}}(\xi,t,k) = \frac{e^{2i \hat{\theta}_1}}{k-k_1} \breve{\mu}_{+,2}^{\mathrm{dn}}(\xi,t,k_1) C_1 + \frac{e^{-2i \hat{\theta}_1^*}}{k + k_1^*} \breve{\mu}_{+,2}^{\mathrm{dn}}(\xi,t,-k_1^*)\tilde{C}_1, \label{e:B.7b}\\
\breve{\mu}_{+,2}^{\mathrm{up}}(\xi,t,k) = \frac{e^{-2i \hat{\theta}_1^*}}{k-k_1^*} \breve{\mu}_{-,1}^{\mathrm{up}}(\xi,t,k_1^*) \bar{C}_1 + \frac{e^{2i \hat{\theta}_1}}{k + k_1} \breve{\mu}_{-,1}^{\mathrm{up}}(\xi,t,-k_1)\bar{\tilde{C}}_1,\label{e:B.7c}\\
\breve{\mu}_{+,2}^{\mathrm{dn}}(\xi,t,k) = I_2 + \frac{e^{-2i \hat{\theta}_1^*}}{k-k_1^*} \breve{\mu}_{-,1}^{\mathrm{dn}}(\xi,t,k_1^*) \bar{C}_1 + \frac{e^{2i \hat{\theta}_1}}{k + k_1} \breve{\mu}_{-,1}^{\mathrm{dn}}(\xi,t,-k_1)\bar{\tilde{C}}_1,\label{e:B.7d}
\end{gather}
\end{subequations}
where the eigenfunctions which appear in Eqs.~\eref{e:B.7} are given by the following expressions for a 1-fundamental soliton:
\begin{subequations}\label{e:44}
\begin{gather}
\breve{\mu}_{-,1}^{\mathrm{up}}(\xi,t,k_1^*) = \frac{1}{\Delta} \Bigg[ I_2 - \frac{(k_1 + k_1^*)e^{2i(\hat{\theta}_1 - \hat{\theta}_1^*)}}{2k_1^*(k_1-k_1^*)^2} ||\pmb{\gamma}_1||^2 \mathrm{diag} (0\,,1) \Bigg]\,,\label{e:44a}\\
\breve{\mu}_{-,1}^{\mathrm{up}}(\xi,t,-k_1) = \frac{1}{\Delta} \Bigg[ I_2 - \frac{(k_1 + k_1^*)e^{2i(\hat{\theta}_1 - \hat{\theta}_1^*)}}{2k_1(k_1-k_1^*)^2} ||\pmb{\gamma}_1||^2 \mathrm{diag} (1\,,0) \Bigg]\,,\label{e:44b}\\
\breve{\mu}_{-,1}^{\mathrm{dn}}(\xi,t,k_1^*) = \frac{1}{\Delta} \Bigg[ \frac{e^{2i \hat{\theta}_1}}{k_1^* - k_1} C_1 + \frac{e^{-2i \hat{\theta}_1^*}}{2k_1^*} \tilde{C}_1 \Bigg]\,,\label{e:44c}\\
\breve{\mu}_{-,1}^{\mathrm{dn}}(\xi,t,-k_1) = \frac{1}{\Delta} \Bigg[ \frac{e^{-2i \hat{\theta}_1^*}}{k_1^* - k_1} \tilde{C}_1 - \frac{e^{2i \hat{\theta}_1}}{2k_1} C_1 \Bigg]\,,\label{e:44d}
\end{gather}
\end{subequations}
and
\begin{subequations}\label{e:54}
\begin{gather}
\breve{\mu}_{+,2}^{\mathrm{up}}(\xi,t,k_1) = \frac{1}{\Delta} \Bigg[ \frac{e^{-2i \hat{\theta}_1^*}}{k_1-k_1^*}\bar{C}_1 + \frac{e^{2i \hat{\theta}_1}}{2k_1}\bar{\tilde{C}}_1 \Bigg]\,,\label{e:54a}\\
\breve{\mu}_{+,2}^{\mathrm{up}}(\xi,t,-k_1^*) = \frac{1}{\Delta} \Bigg[\frac{e^{2i \hat{\theta}_1}}{k_1-k_1^*}\bar{\tilde{C}}_1 - \frac{e^{-2i \hat{\theta}_1^*}}{2k_1^*}\bar{C}_1\Bigg]\,,\label{e:54b}\\
\breve{\mu}_{+,2}^{\mathrm{dn}}(\xi,t,k_1) = \frac{1}{\Delta} \Bigg[ I_2 + \frac{(k_1+k_1^*)e^{2i(\hat{\theta}_1 - \hat{\theta}_1^*)}}{2k_1(k_1-k_1^*)^2}\tilde{C}_1 \bar{\tilde{C}}_1 \Bigg]\,,\label{e:54c}\\
\breve{\mu}_{+,2}^{\mathrm{dn}}(\xi,t,-k_1^*) = \frac{1}{\Delta} \Bigg[ I_2 + \frac{(k_1+k_1^*)e^{2i(\hat{\theta}_1 - \hat{\theta}_1^*)}}{2k_1^*(k_1-k_1^*)^2}C_1 \bar{C}_1 \Bigg]\,,\label{e:54d}
\end{gather}
\end{subequations}
where
\begin{gather*}
\hat{\theta}_1 = \hat{\theta}(k_1), \quad \hat{\theta}_1^* = \hat{\theta}(k_1^*), \quad \hat{\theta}(k) = k \xi - t/4k, \quad \Delta = 1 - \frac{e^{2i(\hat{\theta}_1 - \hat{\theta}_1^*)}}{(k_1-k_1^*)^2}||\pmb{\gamma}_1||^2,
\end{gather*}
and the $(\xi,t)$ dependence in $\hat{\theta}$ has been omitted for brevity. Let us first focus in Eq.~\eref{e:18b}, and compute the limits of the upper/lower blocks of $\breve{\mu}_{+,2}(\xi,t,k)$ as $\xi \to - \infty$.
Inserting Eqs.~\eref{e:44} into \eref{e:B.7c}-\eref{e:B.7d} we obtain
\begin{subequations}\label{e:45}
\begin{gather}
\breve{\mu}^{\mathrm{up}}_{+,2}(\xi,t,k) = \frac{1}{\Delta}\left\{ \frac{e^{- 2i \hat{\theta}_1^*}}{k - k_1^*} \bar{C}_1 + \frac{e^{2i \hat{\theta}_1}}{k + k_1} \bar{\tilde{C}}_1 \right\}\,,\label{e:45a}\\
\breve{\mu}^{\mathrm{dn}}_{+,2}(\xi,t,k) = I_2 + \frac{1}{\Delta}\frac{e^{2i(\hat{\theta}_1 - \hat{\theta}_1^*)}}{(k_1^* - k_1)}\left\{ \frac{1}{k - k_1^*}C_1 \bar{C}_1 + \frac{1}{k + k_1}\tilde{C}_1 \bar{\tilde{C}}_1\right\}\,,\label{e:45b}
\end{gather}
\end{subequations}
where we have taken into account that
\begin{gather*}
\label{e:24new}
\mathrm{diag}(1\,,0) \bar{\tilde{C}}_1 = 0_2, \quad \mathrm{diag}(0\,,1)  \bar{C}_1 = 0_2, \quad C_1 \bar{\tilde{C}}_1 = 0_2, \quad \tilde{C}_1 \bar{C}_1 = 0_2,
\end{gather*}
on account of the form of $C_1$ in the 1-fundamental soliton case, and of the symmetries \eref{e:eq3}.
Since the scattering coefficients are time-independent, we can set $t = 0$ in Eqs.~\eref{e:45} and compute the limits of both sides as $\xi \to -\infty$, which yields
\begin{subequations}\label{e:limitsxi}
\begin{gather}
\lim_{\xi \to - \infty} \breve{\mu}^{\mathrm{up}}_{+,2}(\xi,0,k) = 0_2,\\
\lim_{\xi \to - \infty}\breve{\mu}^{\mathrm{dn}}_{+,2}(\xi,0,k) = I_2 + \frac{k_1 - k_1^*}{k-k_1^*} \frac{C_1 \bar{C}_1}{||\pmb{\gamma}_1||^2} + \frac{k_1 - k_1^*}{k+k_1} \frac{\tilde{C}_1 \bar{\tilde{C}}_1}{||\pmb{\gamma}_1||^2},
\end{gather}
\end{subequations}
assuming that $\nu_j >0$, for $j=1,2$. Combining Eqs. \eref{e:limitsxi} with \eref{e:18b}, we find
\begin{gather}
c(k) = \left( \lim_{\xi \to - \infty}\hat{\mu}_{\infty,2}^{\mathrm{dn}}(\xi,t)\right)\left( I_2 + \frac{k_1 - k_1^*}{k-k_1^*} \frac{C_1 \bar{C}_1}{||\pmb{\gamma}_1||^2} + \frac{k_1 - k_1^*}{k+k_1} \frac{\tilde{C}_1 \bar{\tilde{C}}_1}{||\pmb{\gamma}_1||^2}\right).
\end{gather}
As to the remaining blocks, i.e., relations \eref{e:54} and \eref{e:B.7a}-\eref{e:B.7b},
one can then check that
\begin{subequations}
\begin{gather}
\lim_{\xi \to - \infty} \breve{\mu}^{\mathrm{dn}}_{-,1}(\xi,0,k) = 0_2,\\
\lim_{\xi \to - \infty} \breve{\mu}^{\mathrm{up}}_{-,1}(\xi,0,k) = I_2 - \frac{k_1-k_1^*}{k-k_1} \frac{\bar{C}_1 C_1}{||\pmb{\gamma}_1||^2} - \frac{k_1-k_1^*}{k+k_1^*} \frac{\bar{\tilde{C}}_1 \tilde{C}_1}{||\pmb{\gamma}_1||^2},
\end{gather}
\end{subequations}
and therefore
\begin{gather}\label{e:99}
\bar{c}(k) = \left( \lim_{\xi \to - \infty}\hat{\mu}_{\infty,1}^{\mathrm{up}}(\xi,t)\right)\left( I_2 - \frac{k_1-k_1^*}{k-k_1} \frac{\bar{C}_1 C_1}{||\pmb{\gamma}_1||^2} - \frac{k_1-k_1^*}{k+k_1^*} \frac{\bar{\tilde{C}}_1 \tilde{C}_1}{||\pmb{\gamma}_1||^2}\right).
\end{gather}
One can prove that the left and right scattering coefficients are related by the following symmetries:
\begin{gather}
\label{e:14}
a(k) = \bar{c}^{\dagger}(k^*), \qquad \bar{a}(k) = c^{\dagger}(k^*),
\end{gather}
and therefore
\begin{subequations}\label{e:19}
\begin{gather}
a(k) = \left( I_2 + \frac{k_1-k_1^*}{k-k_1^*} \frac{\bar{C}_1 C_1}{||\pmb{\gamma}_1||^2} + \frac{k_1-k_1^*}{k+k_1}\frac{\bar{\tilde{C}}_1 \tilde{C}_1}{||\pmb{\gamma}_1||^2}\right)\left( \lim_{\xi \to - \infty} \hat{\mu}_{\infty,1}^{\mathrm{up}}(\xi,t)\right)^\dagger\,,\label{e:19a}\\
\bar{a}(k) = \left( I_2 - \frac{k_1-k_1^*}{k-k_1}\frac{C_1 \bar{C}_1}{||\pmb{\gamma}_1||^2} - \frac{k_1 - k_1^*}{k+k_1^*}\frac{\tilde{C}_1 \bar{\tilde{C}}_1}{||\pmb{\gamma}_1||^2}\right)\left( \lim_{\xi \to - \infty} \hat{\mu}_{\infty,2}^{\mathrm{dn}}(\xi,t)\right)^\dagger\,. \label{e:19b}
\end{gather}
\end{subequations}
Moreover, one can check that
$\bar{C}_1 C_1 = - ||\pmb{\gamma}_1||^2 \mathrm{diag}(1\,,0)$ and $\bar{\tilde{C}}_1 \tilde{C}_1 = - ||\pmb{\gamma}_1||^2 \mathrm{diag}(0\,,1)$,
and therefore Eqs.~\eref{e:19} can be simplified to
\begin{subequations}\label{e:4}
\begin{gather}
a(k) = \left( I_2 -
\mathrm{diag} \left( \frac{k_1-k_1^*}{k-k_1^*}\,, \frac{k_1-k_1^*}{k+k_1} \right)\right)\left( \lim_{\xi \to - \infty} \hat{\mu}_{\infty,1}^{\mathrm{up}}(\xi)\right)^\dagger\,,\label{e:4a}\\
\bar{c}(k) =  \left( \lim_{\xi \to - \infty}\hat{\mu}_{\infty,1}^{\mathrm{up}}(\xi)\right)\left( I_2 + \mathrm{diag} \left( \frac{k_1-k_1^*}{k-k_1}\,,\frac{k_1-k_1^*}{k+k_1^*}\right)\right).
\end{gather}
\end{subequations}
To complete the derivation of the transmission coefficients for a 1-fundamental soliton, it remains to compute the limits of the quantities $\hat{\mu}_{\infty,1}^{\mathrm{up}}(\xi)$ and $\hat{\mu}_{\infty,2}^{\mathrm{dn}}(\xi)$ as $\xi \to - \infty$; let us denote
\begin{gather}
B = \left(  \lim_{\xi \to - \infty} \hat{\mu}_{\infty,1}^{\mathrm{up}}(\xi) \right)^\dagger, \quad D = \left( \lim_{\xi \to - \infty} \hat{\mu}_{\infty,2}^{\mathrm{dn}}(\xi)\right)^\dagger.
\end{gather}
This can be done using the asymptotics of $a(k)$ and $\bar{a}(k)$, as $k \to 0$, which recall from \cite{ABBA}, are given by
\begin{subequations}\label{e:9}
\begin{gather}
a(k) = \left( 1 + i k \tilde{\gamma} \right) I_2+O(k^2), \quad \bar{a}(k) = \left( 1-ik\tilde{\gamma}\right) I_2+O(k^2),\label{e:9a}\\
\tilde{\gamma} = \bigints_{- \infty}^{\infty} \left( \sqrt{1 + ||u_x||^2} -1 \right) dx.\label{e:9b}
\end{gather}
\end{subequations}
Moreover, from \eref{e:4a} one can easily compute the asymptotic expansion of $a(k)$ as $k \to 0$:
\begin{gather}
\label{e:8}
a(k) = \left( I_2 + \mathrm{diag} \left( \frac{k_1-k_1^*}{k_1^*}\,, \frac{k_1^* - k_1}{k_1} \right)\right)B + k \, \mathrm{diag} \left( \frac{k_1 - k_1^*}{(k_1^*)^2}\,,\frac{k_1-k_1^*}{k_1^2} \right) B+O(k^2),
\end{gather}
and equating the same orders of $k$ in Eqs.~\eref{e:8} and \eref{e:9a}, we obtain
\begin{gather}\label{e:matrixB}
B = \mathrm{diag} \left( \frac{k_1^*}{k_1}\,,\frac{k_1}{k_1^*} \right),\quad
B = i \tilde{\gamma} \mathrm{diag} \left( \frac{(k_1^*)^2}{k_1-k_1^*}\,, \frac{k_1^2}{k_1 - k_1^*} \right),
\end{gather}
where the first part of this relation gives the matrix $B$ explicitly. Furthermore, equating the two expressions of $B$ in \eref{e:matrixB}, we get the following two relations entry-wise
\begin{gather}\label{e:gmtilde}
\frac{k_1-k_1^*}{(k_1^*)^2} = i \tilde{\gamma} \left( \frac{k_1}{k_1^*}\right), \quad \frac{k_1-k_1^*}{k_1^2} = i \tilde{\gamma} \left( \frac{k_1^*}{k_1}\right) \,,
\end{gather}
which we can both solve independently, and obtain the same result for the quantity $\tilde{\gamma}$
\begin{gather}
i \tilde{\gamma} = \frac{k_1-k_1^*}{|k_1|^2}\equiv\frac{2i\Im k_1}{|k_1|^2}\,.
\end{gather}
Finally, inserting the expression for $B$ into Eqs.~\eref{e:4}, we obtain the expressions for $a(k)$ and $\bar{c}(k)$ for a 1-fundamental soliton
\begin{subequations}
\label{e:1fs_acb}
\begin{gather}
a(k) = \mathrm{diag}\left( \frac{k_1^*}{k_1} \frac{k-k_1}{k-k_1^*}\,, \frac{k_1}{k_1^*}\frac{k+k_1^*}{k+k_1}\right),\\\bar{c}(k) = \mathrm{diag}\left( \frac{k_1}{k_1^*} \frac{k-k_1^*}{k-k_1}\,, \frac{k_1^*}{k_1}\frac{k+k_1}{k+k_1^*}\right),
\end{gather}
\end{subequations}
as functions of $k$ and of the discrete eigenvalue $k_1$, which are independent of the entries of the norming constant. Similarly, we use equation \eref{e:19b} to compute the asymptotic behavior of $\bar{a}(k)$ as $k \to 0$, which in turn we compare with its asymptotics given in \eref{e:9a}. This yields the explicit expression of matrix $D$
which we use together with \eref{e:19b} to obtain:
\begin{subequations}\label{e:24'}
\begin{gather}
\bar{a}(k) = \left( I_2 - \frac{k_1 - k_1^*}{k-k_1} \frac{C_1 \bar{C}_1}{||\pmb{\gamma}_1||^2} - \frac{k_1-k_1^*}{k+k_1^*}\frac{\tilde{C}_1 \bar{\tilde{C}}_1}{||\pmb{\gamma}_1||^2}\right)\times\\ \notag
\times \left( I_2 + \frac{k_1-k_1^*}{k_1} \frac{C_1 \bar{C}_1}{||\pmb{\gamma}_1||^2} - \frac{k_1 - k_1^*}{k_1^*}\frac{\tilde{C}_1 \bar{\tilde{C}}_1}{||\pmb{\gamma}_1||^2}\right)^{-1},\,\label{e:24a}\\
c(k) = \left( I_2 + \frac{k_1-k_1^*}{k_1} \frac{C_1 \bar{C}_1}{||\pmb{\gamma}_1||^2} - \frac{k_1-k_1^*}{k_1^*}\frac{\tilde{C}_1 \bar{\tilde{C}}_1}{||\pmb{\gamma}_1||^2} \right)^{-1} \times \\ \notag
\times \left( I_2 +  \frac{k_1-k_1^*}{k-k_1^*}\frac{C_1 \bar{C}_1}{||\pmb{\gamma}_1||^2} + \frac{k_1-k_1^*}{k+k_1} \frac{\tilde{C}_1 \bar{\tilde{C}}_1}{||\pmb{\gamma}_1||^2}\right).
\end{gather}
\end{subequations}
Notice that the expressions of $\bar{a}$ and $c$ are more complicated than those of $a$ and $\bar{c}$; the transmission coefficients $\bar{a}$ and $c$ depend both on the discrete eigenvalues and on the entries of the norming constants. Summarizing, for a 1-fundamental soliton solution with discrete eigenvalue $k_1$ and associated norming constant $C_1$, we find its transmission coefficients given by \eref{e:1fs_acb} and \eref{e:24'}.
Moreover, one can explicitly compute the inverse matrices which appear in Eqs.~\eref{e:24'} and simplify the expressions of $\bar{a}_j(k)$ and $c_j(k)$ into \eref{e:scdata}.

\subsection{Transmission coefficients for a 1-fundamental breather solution}
Below, we derive the transmission coefficients for a 1-fundamental breather solution. The procedure is similar, but in this case the norming constant $C_1$ takes the form $C_1=(\mu_1 \pmb{\gamma}_1\,, \kappa_1\, \pmb{\gamma}_1)$, for $\mu_1, \kappa_1 \in \mathbb{C}$, and the other ones are obtained from \eref{e:eq3}. Recall that the transmission coefficients $c$ and $\bar{c}$ are given by \eref{e:18}, and the upper/lower blocks of $\breve{\mu}_{-,1}$ and $\breve{\mu}_{+,2}$ are as in Eqs.~\eref{e:B.7}, where
the eigenfunctions which appear in Eqs.~\eref{e:B.7} are given by the following expressions for a 1-fundamental breather:
\begin{subequations}\label{e:161new}
\begin{gather}
\breve{\mu}_{-,1}^{\mathrm{up}}(\xi,t,k_1^*) = \frac{1}{\tilde{\Delta}} \Bigg[ I_2 - \frac{(k_1 + k_1^*)e^{2i(\hat{\theta}_1 - \hat{\theta}_1^*)}}{2k_1^* (k_1 - k_1^*)^2}||\pmb{\gamma}_1||^2 H_{2,2}^{1} \Bigg],\\
\breve{\mu}_{-,1}^{\mathrm{up}}(\xi,t,-k_1) = \frac{1}{\tilde{\Delta}} \Bigg[ I_2 - \frac{(k_1 + k_1^*)e^{2i(\hat{\theta}_1 - \hat{\theta}_1^*)}}{2k_1 (k_1 - k_1^*)^2}||\pmb{\gamma}_1||^2 H_{1,1}^{1} \Bigg],\\
\breve{\mu}_{-,1}^{\mathrm{dn}}(\xi,t,k_1^*) = \frac{1}{\tilde{\Delta}} \Bigg[ \frac{e^{2i \hat{\theta}_1}}{k_1^*-k_1} C_1 + \frac{e^{-2i \hat{\theta}_1^*}}{2 k_1^*} \tilde{C}_1 - \frac{e^{2i(\hat{\theta}_1 -2 \hat{\theta}_1^*)}}{2k_1^*(k_1 - k_1^*)^2}||\pmb{\gamma}_1||^2 \tilde{C}_1 H_{2,2}^{1}\\
\begin{aligned}
& \times \Bigg( I_2 - \frac{1}{\tilde{\Delta}}\Big(I_2 - \frac{e^{2i(\hat{\theta}_1 - \hat{\theta}_1^*)}}{(k_1 - k_1^*)^2}||\pmb{\gamma}_1||^2 H_{2,2}^{1}\Big)\Bigg)\Bigg]\notag
\end{aligned},\\
\breve{\mu}_{-,1}^{\mathrm{dn}}(\xi,t,-k_1) = \frac{1}{\tilde{\Delta}} \Bigg[ \frac{e^{-2i \hat{\theta}_1^*}}{k_1^*-k_1} \tilde{C}_1 - \frac{e^{2i \hat{\theta}_1}}{2 k_1} C_1 + \frac{e^{2i(2\hat{\theta}_1 - \hat{\theta}_1^*)}}{2k_1(k_1 - k_1^*)^2}||\pmb{\gamma}_1||^2 C_1 H_{1,1}^{1}\\
\begin{aligned}
& \times \Bigg( I_2 - \frac{1}{\tilde{\Delta}}\Big(I_2 - \frac{e^{2i(\hat{\theta}_1 - \hat{\theta}_1^*)}}{(k_1 - k_1^*)^2}||\pmb{\gamma}_1||^2 H_{1,1}^{1}\Big)\Bigg)\Bigg],\notag
\end{aligned}
\end{gather}
\end{subequations}
and
\begin{subequations}
\begin{gather}
\breve{\mu}_{+,2}^{\mathrm{up}}(\xi,t,k_1) = \frac{1}{\tilde{\Delta}} \left[ \frac{e^{-2i\hat{ \theta}_1^*}}{k_1-k_1^*}\bar{C}_1 + \frac{e^{2i \hat{\theta}_1}}{2k_1}\bar{\tilde{C}}_1\right],\\
\breve{\mu}_{+,2}^{\mathrm{up}}(\xi,t,-k_1^*) = \frac{1}{\tilde{\Delta}} \left[ \frac{e^{2i \hat{\theta}_1}}{k_1 - k_1^*} \bar{\tilde{C}}_1 - \frac{e^{-2i \hat{\theta}_1^*}}{2k_1^*} \bar{C}_1 \right],\\
\breve{\mu}_{+,2}^{\mathrm{dn}}(\xi,t,k_1) = I_2 -  \frac{e^{2i(\hat{\theta}_1 - \hat{\theta}_1^*)}}{\tilde{\Delta}(k_1-k_1^*)}\left[ \frac{1}{k_1 - k_1^*} C_1 \bar{C}_1 + \frac{1}{2k_1}\tilde{C}_1 \bar{\tilde{C}}_1 \right],\\
\breve{\mu}_{+,2}^{\mathrm{dn}}(\xi,t,-k_1^*) = I_2 + \frac{e^{2i(\hat{\theta}_1 - \hat{\theta}_1^*)}}{\tilde{\Delta}(k_1-k_1^*)}\left[ \frac{1}{2 k_1^*} C_1 \bar{C}_1 - \frac{1}{k_1 - k_1^*} \tilde{C}_1 \bar{\tilde{C}}_1 \right],
\end{gather}
\end{subequations}
where we have introduced the notations
\begin{subequations}
\begin{gather}
H_{1,1}^{1} =
\begin{pmatrix}
|\mu_1|^2 & \kappa_1 \mu_1^*\\[2pt]
\kappa_1^* \mu_1 & |\kappa_1|^2
\end{pmatrix},\quad H_{2,2}^{1} =
\begin{pmatrix}
|\kappa_1|^2 & -\kappa_1 \mu_1^* \\[2pt]
-\kappa_1^* \mu_1 & |\mu_1|^2
\end{pmatrix}, \\
\tilde{\Delta} = 1 - \frac{e^{2i(\hat{\theta}_1 - \hat{\theta}_1^*)}}{(k_1 - k_1^*)^2}(|\mu_1|^2 + |\kappa_1|^2)||\pmb{\gamma}_1||^2.
\end{gather}
\end{subequations}
Taking the appropriate limits as stated in Eqs.~\eref{e:18}, setting $t=0$ because the scattering coefficients are time-independent, and using the symmetry \eref{e:14},
we can determine the following expressions for the transmission coefficients
\begin{subequations}\label{e:143new}
\begin{gather}
a(k) = \left( I_2 + \frac{k_1-k_1^*}{k-k_1^*}\frac{\bar{C}_1 C_1}{\left( |\mu_1|^2 + |\kappa_1|^2\right) || \pmb{\gamma}_1||^2} + \frac{k_1-k_1^*}{k+k_1}\frac{\bar{\tilde{C}}_1 \tilde{C}_1}{\left( |\mu_1|^2 + |\kappa_1|^2\right) || \pmb{\gamma}_1||^2}\right)A^\dagger,\label{e:143newa}\\
\bar{c}(k) = A \left( I_2 + \frac{k_1^*-k_1}{k-k_1} \frac{\bar{C}_1 C_1}{\left( |\mu_1|^2 + |\kappa_1|^2\right) || \pmb{\gamma}_1||^2} + \frac{k_1^*-k_1}{k+k_1^*}\frac{\bar{\tilde{C}}_1 \tilde{C}_1}{\left( |\mu_1|^2 + |\kappa_1|^2\right) || \pmb{\gamma}_1||^2}\right),\label{e:143newb}\\
c(k) = D \left( I_2 + \frac{k_1-k_1^*}{k-k_1^*}\frac{C_1 \bar{C}_1}{\left( |\mu_1|^2 + |\kappa_1|^2\right) || \pmb{\gamma}_1||^2} + \frac{k_1-k_1^*}{k+k_1}\frac{\tilde{C}_1 \bar{\tilde{C}}_1}{\left( |\mu_1|^2 + |\kappa_1^2\right) ||\pmb{ \gamma}_1||^2}\right),\label{e:143newc}\\
\bar{a}(k) = \left( I_2 + \frac{k_1^*-k_1}{k-k_1}\frac{C_1 \bar{C}_1}{\left( |\mu_1|^2 + |\kappa_1|^2\right) || \pmb{\gamma}_1||^2} + \frac{k_1^*-k_1}{k+k_1^*}\frac{\tilde{C}_1 \bar{\tilde{C}}_1}{\left( |\mu_1|^2 + |\kappa_1|^2\right) || \pmb{\gamma}_1||^2} \right)D^\dagger,\label{e:143newd}
\end{gather}
\end{subequations}
where by $A$ and $D$ we denote the $2 \times 2$ constant matrices
\begin{gather}
A = \lim_{\xi \to - \infty} \breve{\mu}_{\infty,1}^{\mathrm{up}}(\xi), \quad D = \lim_{\xi \to - \infty} \breve{\mu}_{\infty,2}^{\mathrm{dn}}(\xi).
\end{gather}
Again, computing the asymptotics of $a(k)$ and $\bar{a}(k)$ as $k \to 0$ in the last equations and equating same orders of $k$ with the asymptotic behavior \eref{e:9}
one can compute the matrices $A$ and $D$ explicitly
\begin{subequations}
\begin{gather}
A = I_2 - \frac{k_1-k_1^*}{k_1}\frac{H_{1,1}^1}{|\mu_1|^2+|\kappa_1|^2} + \frac{k_1-k_1^*}{k_1^*}\frac{H_{2,2}^1}{|\mu_1|^2+|\kappa_1|^2},\\
D = I_2 + \frac{k_1^*-k_1}{k_1^*}\frac{C_1 \bar{C}_1}{\left(|\mu_1|^2+|\kappa_1|^2\right)||\pmb{\gamma}_1||^2} - \frac{k_1^*-k_1}{k_1}\frac{\tilde{C}_1 \bar{\tilde{C}}_1}{\left(|\mu_1|^2+|\kappa_1|^2\right)||\pmb{\gamma}_1||^2},
\end{gather}
\end{subequations}
and inserting these expressions into Eqs.~\eref{e:143new}, we obtain the transmission coefficients for a 1-fundamental breather solution:
\begin{subequations}
\label{e:trans_fb}
\begin{gather}
a_j(k) = I_2 + \frac{k}{k_j}\frac{k_j-k_j^*}{k-k_j^*} \frac{\bar{C}_j C_j}{\left(|\mu_j|^2+|\kappa_j|^2\right)||\pmb{\gamma}_j||^2} - \frac{k}{k_j^*}\frac{k_j-k_j^*}{k+k_j}\frac{\bar{\tilde{C}}_j \tilde{C}_j}{\left(|\mu_j|^2+|\kappa_j|^2\right)||\pmb{\gamma}_j||^2},\\
\bar{c}_j(k) = I_2 + \frac{k}{k_j^*}\frac{k_j^* - k_j}{k-k_j}\frac{\bar{C}_j C_j}{\left(|\mu_j|^2+|\kappa_j|^2\right)||\pmb{\gamma}_j||^2}-\frac{k}{k_j}\frac{k_j^* - k_j}{k+k_j^*}\frac{\bar{\tilde{C}}_j \tilde{C}_j}{\left(|\mu_j|^2+|\kappa_j|^2\right)||\pmb{\gamma}_j||^2},\\
c_j(k) = I_2 + \frac{k}{k_j}\frac{k_j-k_j^*}{k-k_j^*}\frac{C_j \bar{C}_j}{\left( |\mu_j|^2+|\kappa_j|^2\right) ||\pmb{\gamma}_j||^2} - \frac{k}{k_j^*}\frac{k_j-k_j^*}{k+k_j}\frac{\tilde{C}_j \bar{\tilde{C}}_j}{\left( |\mu_j|^2+|\kappa_j|^2\right) ||\pmb{\gamma}_j||^2},\\
\bar{a}_j(k) = I_2 + \frac{k}{k_j^*}\frac{k_j^* - k_j}{k - k_j}\frac{C_j \bar{C}_j}{\left( |\mu_j|^2+|\kappa_j|^2\right) ||\pmb{\gamma}_j||^2} - \frac{k}{k_j}\frac{k_j^* - k_j}{k+k_j^*}\frac{\tilde{C}_j \bar{\tilde{C}}_j}{\left(|\mu_j|^2+|\kappa_j|^2\right) ||\pmb{\gamma}_j||^2},
\end{gather}
\end{subequations}
as functions of $k$, in terms of the discrete eigenvalue and the associated norming constant.

\subsection{Transmission coefficients for a 1-self-symmetric soliton solution}

Below, we derive the transmission coefficients for a 1-self-symmetric soliton solution, i.e., a composite breather with self-symmetric (i.e., purely imaginary) discrete eigenvalues. Recall that the transmission coefficients $c$ and $\bar{c}$ are given by Eqs.~\eref{e:18}
and using Eqs.~\eref{e:80} for the eigenfunctions of a self-symmetric soliton, we then obtain:
\begin{subequations}
\begin{gather}
 \breve{\mu}_{+,2}^{\mathrm{up}}(\xi,t,k) =  \frac{e^{-2i \theta_1^*}}{k - k_1^*} \left( I_2 + \frac{e^{4i \theta_1}}{(k_1^* - k_1)^2} \bar{C}_1 C_1 \right)^{-1} \bar{C}_1,\\
 \breve{\mu}_{+,2}^{\mathrm{dn}}(\xi,t,k) = I_2 + \frac{e^{2i(\theta_1 - \theta_1^*)}}{(k - k_1^*)(k_1^* - k_1)} \left( I_2 + \frac{e^{4i \theta_1}}{(k_1^* - k_1)^2} C_1 \bar{C}_1 \right)^{-1} C_1 \bar{C}_1,
\end{gather}
\end{subequations}
where taking their limits as $\xi \to -\infty$, and setting $t=0$, we get
\begin{subequations}
\begin{gather}
\lim_{\xi \to - \infty} \breve{\mu}_{+,2}^{\mathrm{up}}(\xi,k) = 0_2, \quad \lim_{\xi \to - \infty} \breve{\mu}_{+,2}^{\mathrm{dn}}(\xi,k) = \left( \frac{k-k_1}{k-k_1^*} \right) I_2,
\end{gather}
\end{subequations}
assuming again that $\nu_j>0$, for $j=1,2$. Using Eqs.~\eref{e:18} and the symmetries
\eref{e:14}, we obtain the following expressions for the transmission coefficients:
\begin{subequations}\label{e:178}
\begin{gather}
c(k) = \frac{k-k_1}{k-k_1^*} \left( \lim_{\xi \to - \infty} \hat{\mu}_{\infty,2}^{\mathrm{dn}}(\xi)\right), \quad
\bar{c}(k) = \frac{k - k_1^*}{k - k_1} \left( \lim_{\xi \to - \infty} \hat{\mu}_{\infty,1}^{\mathrm{up}}(\xi) \right),\label{e:178b}\\
a(k) = \frac{k - k_1}{k - k_1^*} \left( \lim_{\xi \to - \infty} \hat{\mu}_{\infty,1}^{\mathrm{up}}(\xi) \right)^\dagger, \quad \bar{a}(k) = \frac{k - k_1^*}{k - k_1} \left(\lim_{\xi \to - \infty} \hat{\mu}_{\infty,2}^{\mathrm{dn}}(\xi) \right)^\dagger.\label{e:178d}
\end{gather}
\end{subequations}
In addition, one can use the asymptotic behavior of $a(k)$, $\bar{a}(k)$ as $k \to 0$, and as done in the previous cases obtain
\begin{gather}\label{e:177new}
\lim_{\xi \to - \infty} \hat{\mu}_{\infty,1}^{\mathrm{up}}(\xi) = -I_2, \quad \lim_{\xi \to - \infty} \hat{\mu}_{\infty,2}^{\mathrm{dn}}(\xi) = -I_2,
\end{gather}
by which Eqs.~\eref{e:178} yield:
\begin{gather}\label{e:176N}
c(k) = \frac{k_1 - k}{k_1 + k} I_2, \quad \bar{c}(k) = \left( \frac{k_1^* - k}{k_1^* + k} \right) I_2, \quad a(k) \equiv c(k), \quad \bar{a}(k) \equiv \bar{c}(k).
\end{gather}

\section{Proof of Lemma \ref{lemmasym}}
\textbf{Lemma \ref{lemmasym}.} \begin{proof}
		Spelling out the condition \eref{Lambdasym} gives
		\begin{gather}\label{e:C1}
					\left(I-\frac{\alpha_1-\alpha_1^*}{\lambda+\alpha_1}\Pi_1^*\right)
			\left(I-\frac{\alpha_2-\alpha_2^*}{\lambda+\alpha_2}\Pi_2^*\right)
			=\notag\\
\left(I-\frac{\alpha_1-\alpha_1^*}{\lambda-\alpha_1^*}\Lambda\Pi_1\Lambda^{-1}\right)
\left(I-\frac{\alpha_2-\alpha_2^*}{\lambda-\alpha_2^*}\Lambda\Pi_2\Lambda^{-1}\right).
		\end{gather}
		The poles on the LHS and RHS must coincide so we must have $\{-\alpha_1,-\alpha_2\}=\{\alpha_1^*,\alpha_2^*\}$. Suppose $-\alpha_1=\alpha_1^*$ (and thus $-\alpha_2=\alpha_2^*$); then we have from the last relation
		\begin{gather}
			\left(I-\frac{\alpha_1-\alpha_1^*}{\lambda-\alpha_1^*}\Pi_1^*\right)
			\left(I-\frac{\alpha_2-\alpha_2^*}{\lambda-\alpha_2^*}\Pi_2^*\right)=\notag\\
			\left(I-\frac{\alpha_1-\alpha_1^*}{\lambda-\alpha_1^*}\Lambda\Pi_1\Lambda^{-1}\right)
			\left(I-\frac{\alpha_2-\alpha_2^*}{\lambda-\alpha_2^*}\Lambda\Pi_2\Lambda^{-1}\right),
		\end{gather}	
		which is rearranged into
		\begin{gather}
\left(I+\frac{\alpha_1-\alpha_1^*}{\lambda-\alpha_1}\Lambda\Pi_1\Lambda^{-1}\right)
\left(I-\frac{\alpha_1-\alpha_1^*}{\lambda-\alpha_1^*}\Pi_1^*\right)=\notag\\
	\left(I-\frac{\alpha_2-\alpha_2^*}{\lambda-\alpha_2^*}\Lambda\Pi_2\Lambda^{-1}\right)
\left(I+\frac{\alpha_2-\alpha_2^*}{\lambda-\alpha_2}\Pi_2^*\right)\,.
		\end{gather}			
		The LHS has no pole at $\alpha_2$ and $\alpha_2^*$, while the RHS has no pole at $\alpha_1$ and $\alpha_1^*$. In view of our assumption on $\alpha_1$ and $\alpha_2$, this means that the common rational function of $\lambda$ represented by this equality is an entire function on $\mathbb{C}$ with limit $I_4$ at infinity. By Liouville theorem, the function is the identity matrix, and for this to hold we find			
			\begin{equation}
			\Pi_j^*=\Lambda\Pi_j\Lambda^{-1}\,,~~j=1,2\,,
		\end{equation}
		and each elementary factor falls into the realm of Lemma~\ref{selfsymmetric:lemma}. This is case 1.	
		
Suppose now that $\alpha_2=-\alpha_1^*$. Hence, the condition \eref{Lambdasym} now reads
		\begin{align}\label{condition2}
			\left(I-\frac{\alpha_1-\alpha_1^*}{\lambda+\alpha_1}\Pi_1^*\right)
			\left(I-\frac{\alpha_1-\alpha_1^*}{\lambda-\alpha_1^*}\Pi_2^*\right)=\notag\\
\left(I-\frac{\alpha_1-\alpha_1^*}{\lambda-\alpha_1^*}\Lambda\Pi_1\Lambda^{-1}\right)
\left(I-\frac{\alpha_1-\alpha_1^*}{\lambda+\alpha_1}\Lambda\Pi_2\Lambda^{-1}\right).
		\end{align}
		This takes the form of a special case of the refactorization Theorem~\ref{Threfac}, so we have
		\begin{eqnarray}
			\Lambda\,\Pi_1\,\Lambda^{-1}=\phi^{-1}\,\Pi_2^*\,\phi\,,~~\Lambda\,\Pi_2\,\Lambda^{-1}=\phi^{-1}\,\Pi_1^*\,\phi
		\end{eqnarray}
		with
		\begin{equation}
	\phi=2\alpha_1^*I+(\alpha_1-\alpha_1^*)(\Pi_1^*+\Pi_2^*)\,.
		\end{equation}
		This shows that $\Pi_1$ and $\Pi_2$ necessarily have the same rank and that $\Pi_1^*+\Pi_2^*$ commutes with $\phi$, and therefore
		\begin{eqnarray}
			\Pi_1^*+\Pi_2^*=\phi^{-1}\,(\Pi_1^*+\Pi_2^*)\,\phi=	\Lambda\,(\Pi_1+\Pi_2)\,\Lambda^{-1}\,.
		\end{eqnarray}		
		Let us set
		\begin{equation}\label{e:GAmma}
			\Gamma=\Pi_1^*-\Lambda\,\Pi_2\,\Lambda^{-1}=\Lambda\,\Pi_1\,\Lambda^{-1}-\Pi_2^*\,.
		\end{equation}
Then $\Gamma=\Gamma^\dagger$, and therefore $\Gamma$ is diagonalizable, with real eigenvalues, by a unitary matrix $M$:
$$\Gamma=M\,D\,M^\dagger\,,~~D={\rm diag}(d_1,d_2,d_3,d_4)\,,~~d_j\in \mathbb{R} \,,~~j=1,\dots,4\,\,.$$
Since ${\rm Tr}\,\Gamma={\rm rank}\,\Pi_1-{\rm rank}\,\Pi_2=0$, we get
$$d_1+d_2+d_3+d_4=0\,.$$
Note that $\Gamma=\Lambda\,\Gamma^*\,\Lambda^{-1}$, and since $M$ is unitary, we get 		
		\begin{eqnarray}
	M\,D\,M^\dagger=\Lambda\,M^*\,D\,M^T\,\Lambda^{-1}\Leftrightarrow M^\dagger\,\Lambda\,M^*\,D = D \,M^\dagger\,\Lambda\,M^*\,.
\end{eqnarray}
Let now $A=M^\dagger\,\Lambda\,M^*$. $A$ is an antisymmetric matrix which commutes with $D$. It is also unitary, so its determinant has modulus equal to $1$. This yields the conditions 	
\begin{eqnarray}
	\label{conditionsA}
	(d_i-d_j)A_{ij}=0\,,~~i<j\,,~~|A_{12}A_{34}-A_{13}A_{24}+A_{23}A_{14}|^2=1\,.
\end{eqnarray}
This means that at least two $d_j$ must be equal as otherwise $A=0$, which is a contradiction. Suppose all $d_j$ are equal so $\Gamma=d I_4$, $d\in \mathbb{R}$. Note that
\begin{equation}
	\Gamma=\Pi_1^*- I_4 + I_4 -\Lambda\,\Pi_2\,\Lambda^{-1}
\end{equation}
so
\begin{eqnarray}
	-\Gamma^2&=&\left( I_4 -\Pi_1^*-(I_4-\Lambda\,\Pi_2\,\Lambda^{-1})\right)(\Pi_1^*-\Lambda\,\Pi_2\,\Lambda^{-1})\nonumber\\
	&=&-(I_4-\Pi_1^*)\Lambda\,\Pi_2\,\Lambda^{-1}-(I_4-\Lambda\,\Pi_2\,\Lambda^{-1})\Pi_1^*\nonumber\\
	&=&-(I_4-\Lambda\,\Pi_2\,\Lambda^{-1}-\Gamma)\Lambda\,\Pi_2\,\Lambda^{-1}-(I_4-\Lambda\,\Pi_2\,\Lambda^{-1})(\Gamma+\Lambda\,\Pi_2\,\Lambda^{-1})\nonumber\\	\label{propertyGamma}
	&=&\Gamma\,\Lambda\,\Pi_2\,\Lambda^{-1}+\Lambda\,\Pi_2\,\Lambda^{-1}\,\Gamma-\Gamma\,.
\end{eqnarray}
If $\Gamma=d I_4$, since $\Pi_2$ is neither the identity nor $0$, we must have $d=0$, and therefore
\begin{equation}
	\label{relation-Pi}
	\Pi_1^*=\Lambda\,\Pi_2\,\Lambda^{-1}\,,~~\Pi_2^*=\Lambda\,\Pi_1\,\Lambda^{-1}\,.
\end{equation}

Suppose now without loss of generality that $d_1=d_2=d_3$ and $d_4\neq d_j$, $j=1,2,3$. This is in contradiction with \eref{conditionsA}, which requires $A_{14}=A_{24}=A_{34}=0$ but $|A_{12}A_{34}-A_{13}A_{24}+A_{23}A_{14}|^2=1$. Similarly, we cannot have wlog $d_1=d_2$ and $d_3 \neq d_j$, $j=1,2,4$ and $d_4 \neq d_j$, $j=1,2,3$. Hence, it remains the possibility $d_1=d_2$ and $d_3=d_4$. Because of the zero trace of $\Gamma$, this means that $\Gamma=dM\,\Sigma_3\,M^\dagger$, $d \in \mathbb{R}$. It also means that only $A_{12}$ and $A_{34}$ can be nonzero in the matrix $A$, hence $M$ must be a unitary matrix such that
$$M^\dagger\,\Lambda\,M^*=\begin{pmatrix}
	0 & A_{12} & 0 & 0\\
	-A_{12} & 0 & 0 & 0\\
	0 & 0 & 0 & A_{34}\\
	0 & 0 & -A_{34} & 0
\end{pmatrix}\,.$$
Denoting $\check{\Pi}_2=M^\dagger\,\Lambda\,\Pi_2\,\Lambda^{-1}\,M$, we can use \eref{propertyGamma} again to deduce
\begin{equation}\label{e:154}
	d(\Sigma_3-d I_4 -\Sigma_3\check{\Pi}_2-\check{\Pi}_2\Sigma_3)=0\,.
\end{equation}
If $d=0$, we are back to the case \eref{relation-Pi}. Otherwise $d\neq 0$ and $\Gamma=dM\Sigma_3M^\dagger$ is invertible. We show next that this leads to a contradiction so that this case is rejected.
We can write $T(\lambda)$ given in \eref{e:matrixT} as a sum of partial fractions
$$T(\lambda)=I+\frac{R_1}{\lambda-\alpha_1^*}+\frac{R_2}{\lambda-\alpha_2^*},$$
where the residues are given by
$$R_1=-(\alpha_1-\alpha_1^*)\Pi_1\left(I-\frac{\alpha_2-\alpha_2^*}{\alpha_1^*-\alpha_2^*}\Pi_2\right)\,,~~R_2=-(\alpha_2-\alpha_2^*)\left(I-\frac{\alpha_1-\alpha_1^*}{\alpha_2^*-\alpha_1^*}\Pi_1\right)\Pi_2\,.$$
Now recall that we are in the case $\alpha_2=-\alpha_1^*$ and that we are imposing the symmetry \eref{Lambdasym} on $T(\lambda)$. Hence we must have $R_2^*=-\Lambda\,R_1\,\Lambda^{-1}$ which yields explicitly
\begin{eqnarray}	\left(I-\frac{\alpha_1-\alpha_1^*}{\alpha_1+\alpha_1^*}\Pi_1^*\right)\Pi_2^*=\Lambda\,\Pi_1\,\Lambda^{-1}\Lambda\,\left(I-\frac{\alpha_1-\alpha_1^*}{\alpha_1+\alpha_1^*}\Pi_2\right)\,\Lambda^{-1}\,.
\end{eqnarray}
Note that $\alpha_1+\alpha_1^*\neq 0$ in the present case $\alpha_2=-\alpha_1^*$, since we assume in general that $\alpha_2\neq \alpha_1$.
Since we are considering the case where $\Pi_2^*=\Lambda\,\Pi_1\,\Lambda^{-1}-\Gamma$, recalling that $\Lambda\,\Gamma^*\,\Lambda^{-1}=\Gamma$ and noting that
\begin{equation}
	\label{commute1}
	\left[\Pi_1^*,\Lambda\,\Pi_1\,\Lambda^{-1}\right]=0\,,
\end{equation}
we obtain the condition
$$\frac{\alpha_1-\alpha_1^*	}{\alpha_1+\alpha_1^*}\left(\Pi_1^*-\Lambda\,\Pi_1\,\Lambda^{-1}\right)\Gamma=\Gamma\,.$$
Since $\Gamma$ is supposed to be invertible, this yields
$$\frac{\alpha_1-\alpha_1^*	}{\alpha_1+\alpha_1^*}\left(\Pi_1^*-\Lambda\,\Pi_1\,\Lambda^{-1}\right)=I\,,$$
which is a contradiction since the LHS has zero trace while the RHS obviously does not.
\end{proof}

\section{Hodograph transformation and exact two soliton solution of the ccSPE}

The ccSPE \eref{e:ccSP} admits an equivalent form written in conservation law:
\begin{gather}
\left( \rho^{-1}\right)_t - \frac{1}{2} \left( ||\pmb{u}||^2 \rho^{-1} \right)_x = 0, \quad \rho^{-1} = \sqrt{1 + ||\pmb{u}_x||^2},
\end{gather}
which is satisfied provided that the components $u_j$ and their conjugates satisfy Eq.~\eref{e:ccSP}. We can now define a hodograph transformation
\begin{gather}\label{e:4differential}
d \xi = \rho^{-1} dx + \frac{\rho^{-1}}{2}||\pmb{u}||^2 dt,
\end{gather}
which converts the initial variables $(x,t)$ to $(\xi,t)$. Recall that the travel-time parameter $\xi$ is defined by
\begin{gather}\label{e:5xi}
\xi = \bigint_{0}^{x} \sqrt{1 + ||\pmb{u}_y||^2}dy - \bigint_{0}^{\infty} \left( \sqrt{1 + ||\pmb{u}_y||^2} - 1 \right) dy,
\end{gather}
which implies
\begin{gather}
d \xi = \frac{\partial \xi}{\partial x} dx + \frac{\partial \xi}{\partial t}dt,
\end{gather}
where
\begin{subequations}
\begin{gather*}
\partial \xi/\partial x = \frac{\partial}{\partial x}\bigint_{0}^{x} \sqrt{1 + ||\pmb{u}_y||^2}dy = \sqrt{1 + ||\pmb{u}_x||^2 } \equiv \rho^{-1},\,\\
\partial \xi/\partial t = \frac{\partial}{\partial t}\bigint_{0}^{x} \sqrt{1 + ||\pmb{u}_y||^2}dy = \bigint_{0}^{x} \partial_t \rho^{-1} dy = \frac{1}{2}\bigint_{0}^{x} \partial_y \left( ||\pmb{u}||^2 \rho^{-1}\right) dy \equiv \frac{ \rho^{-1}}{2}||\pmb{u}||^2,
\end{gather*}
\end{subequations}
which coincides with \eref{e:4differential}. Moreover, the hodograph transformation converts the Lax pair \eref{e:1.3b} into the following one for $\Psi(\xi,t) = \Phi(x,t)$:
\begin{subequations}\label{e:195}
\begin{gather}
\Psi_{\xi} = R(\xi,t;k) \Psi, \quad R(\xi,t;k) = -i k \rho(\xi,t)\Sigma_3 + k \Sigma_3 V_{0,\xi},\\
\Psi_{t} = S(\xi,t;k) \Psi, \quad S(\xi,t;k) = \frac{i}{4k} \Sigma_3 - \frac{i}{2} V_0,
\end{gather}
\end{subequations}
where
\begin{gather*}
\Sigma_3 = \diag(1,1,-1,-1), \quad V_0 = \begin{pmatrix}
0_{2 \times 2} & U\\
V & 0_{2 \times 2}
\end{pmatrix}, \quad U = \begin{pmatrix}
-i u_1 & -i u_2\\
iu_2^* & -iu_1^*
\end{pmatrix}, \quad V = U^\dagger.
\end{gather*}
The compatibility condition $\Psi_{\xi t} = \Psi_{t \xi}$, gives a two-component coupled complex equation
\begin{gather} \label{e:196}
u_{j, \xi t} = \rho u_j, \quad \rho = 1 - \frac{1}{2} \int ||\textbf{u}||^2_{\xi} dt, \quad \text{for} \quad j=1,2.
\end{gather}\\
Conversely, one can define the inverse hodograph transformation
\begin{gather}
dx = \rho d \xi + \frac{\rho}{2} ||\textbf{u}||^2 dt,
\end{gather}
which converts the Lax pair \eref{e:195} and Eq.~\eref{e:196} into the Lax pair \eref{e:1.3b} and the ccSPE \eref{e:ccSP}, respectively. Therefore, the ccSPE equation \eref{e:ccSP} is equivalent to equation \eref{e:196} through the hodograph transformation \eref{e:4differential}, provided that $\rho > 0$.

Next, we provide a formula for exact 2-soliton solutions, which we used to verify numerically some of the results in this paper. The formula can be obtained as a special case of the ones derived in \cite{FL22} using B\"acklund-Darboux transformations, but has been adapted to the notations of this work and to some differences in the variables used in \cite{FL22} compared to this work. First of all, we introduce the eigenfunctions $\Phi_1$ and $\Phi_2$, given by
\begin{gather}\label{e:192}
\Phi_j = \begin{pmatrix}
\mu_j^* e^{-i \theta_j}\\
\kappa_j^* e^{-i \theta_j}\\
\alpha_j e^{i \theta_j}\\
\gamma_j e^{i \theta_j}
\end{pmatrix}, \quad \theta_j = \xi k_j-t/4k_j, \quad j=1,2,
\end{gather}
where $k_j$, are the two discrete eigenvalues, $(\alpha_j, \beta_j)^{T}$ is the non-zero column of the norming constant $C_j$, and $\mu_j, \kappa_j$ are the two multiplicative constants, such that the two columns of $C_j$ are proportional to each other. Specifically, the norming constant $C_j$ associated to the discrete eigenvalue $k_j$, for which $\Phi_j$ is given by \eref{e:192}, is assumed to be a rank-one matrix of the form $C_j = ( \mu_j \pmb{\gamma}_j, \kappa_j \pmb{\gamma}_j)$, where $\pmb{\gamma}_j = (\alpha_j, \beta_j)^{T}$, for $j=1,2$.
We then define the $4 \times 4$ matrix $Y$ as follows
\begin{gather}
Y = \left[ \Phi_1, \quad \Lambda \Phi_1^*, \quad \Phi_2, \quad  \Lambda \Phi_2^*\right] \equiv \begin{bmatrix}
Y1 \\
Y2
\end{bmatrix},
\end{gather}
where we introduced the matrix $\Lambda = \diag (i\sigma_2,i\sigma_2)$.
Next, we denote by $M_2$ the $4\times 4$ matrix
\begin{gather}
M_2 = \begin{pmatrix}
M_{11} & M_{12}\\
M_{21} & M_{22}
\end{pmatrix},
\end{gather}
where $M_{ij}$ are the $2 \times 2$ blocks defined as
\begin{subequations}
\begin{gather}
M_{11} = \left( \frac{|k_1|^2}{k_1^* - k_1} \Phi_1^\dagger \Phi_1 \right) I_2, \quad M_{22} = \left( \frac{|k_2|^2}{k_2^* - k_2} \Phi_2^\dagger \Phi_2 \right) I_2,\\[2pt]
M_{12} = \begin{pmatrix}
\frac{k_2 k_1^*}{k_1^* - k_2} \Phi_1^\dagger \Phi_2 & \frac{k_1^* k_2^*}{-k_1^* - k_2^*} \Phi_1^\dagger \Lambda \Phi_2^* \\
\frac{-k_1 k_2}{k_1 + k_2} \Phi_1^T \Lambda \Phi_2 & \frac{k_1 k_2^*}{-k_1 + k_2^*} \Phi_1^T \Phi_2^*
\end{pmatrix},\\[2pt]
M_{21} = \begin{pmatrix}
\frac{k_1 k_2^*}{k_2^* - k_1} \Phi_2^\dagger \Phi_1 & \frac{k_1^* k_2^*}{-k_1^* - k_2^*} \Phi_2^\dagger \Lambda \Phi_1^* \\
\frac{-k_1 k_2}{k_1 + k_2} \Phi_2^T \Lambda \Phi_1 & \frac{k_1^* k_2}{-k_2 + k_1^*} \Phi_2^T \Phi_1^*
\end{pmatrix}.
\end{gather}
\end{subequations}
The exact two soliton solution of the ccSPE is then given by
\begin{gather}
\label{e:198}
Q[2] = Y_2 M_2^{-1} Y_1^\dagger, \quad Q[2] = \begin{pmatrix}
q_1 & q_2\\
q_2^* & -q_1^*
\end{pmatrix}.
\end{gather}
Note here that the vector solution $\mathbf{u}$ of the ccSPE \eref{e:ccSP}, and the vector solution $\mathbf{q}$ of the ccSPE in \cite{FL22}, namely
\begin{gather}
\mathbf{q}_{xt} + \mathbf{q} - \frac{1}{2} (||\mathbf{q}||^2 \mathbf{q}_x)_x = 0, \quad \mathbf{q} = (q_1, q_2)^{T},
\end{gather}
are connected via the relation $u_1 = -i q_1^*$ and $u_2 = -i q_2$, for the two Lax pairs to match, provided that $k=1/\lambda$ holds for the discrete eigenvalues of the two systems. Therefore, one can recover the exact 2-soliton solution of \eref{e:ccSP} using equation \eref{e:198}, modulo the change of dependent and independent variables given above. Moreover, in the case of the 2-soliton solution the quantity $\rho$ takes the form
\begin{gather}
\rho = 1 - \ln_{\xi t} \left( \det M_2 \right),
\end{gather}
which after we equate with the expression of $\rho$ in relation \eref{e:196} yields
\begin{gather}
||u||^2 = 2 \ln_{tt}\left( \det M_2 \right).
\end{gather}

\section{Computation of the transmission coefficients from the Darboux matrices}

It is instructive to derive the scattering coefficients of the soliton solutions created from the zero solution by the dressing factor/Darboux matrix $T_{k_1,\Pi_1}(k)$. It is a standard result of the dressing method that the scattering matrix $S(k)$ defined in \eref{e:8new},
is given in this case by
\begin{gather}\label{newe:2}
	S(k) = \lim_{\xi \to +\infty} T_{k_1,\Pi_1}(k) \left(\lim_{\xi \to -\infty} T_{k_1,\Pi_1}(k)\right)^{-1},
\end{gather}
where the limits can be calculated at $t=0$ (recall that the $(\xi,t)$ dependence of $T_{k_1,\Pi_1}(k)$ enters through $\Pi_1(\xi,t)$.)
Thus, all we need to evaluate are the following limits
\begin{equation}
	\lim_{\xi \to \pm\infty} \Pi_1(\xi,0)\,.
\end{equation}
Let us focus on the case of a generic and a non-generic composite breather. Recall that in those cases,
\begin{equation}
\Pi_1(\xi,0)=\Phi_1(\xi)(\Phi_1^\dagger(\xi) \Phi_1(\xi))^{-1}\Phi_1^\dagger(\xi)\,,
\end{equation}
where
\begin{subequations}
	\begin{gather}
		\Phi_1(\xi)=e^{-ik_1 \xi \Sigma_3}(\phi_1,\phi_2)\,,\\
		\phi_1=\begin{pmatrix}
			{\pmb \delta}_1\\
			{\pmb \gamma}_1
		\end{pmatrix}\,,~~
		\phi_2=\begin{pmatrix}
			{\pmb \tau}_1\\
			{\pmb \omega}_1
		\end{pmatrix}\,,~~	{\pmb \delta}_1,{\pmb \gamma}_1,{\pmb \tau}_1,{\pmb \omega}_1\in \mathbb{C}^2\,,~~j=1,2\,,
	\end{gather}
\end{subequations}
and for a generic composite breather, we assume that $\det({\pmb \delta}_1\,, {\pmb \tau}_1) \neq 0$ and $\det({\pmb \gamma}_1\,, {\pmb \omega}_1) \neq 0$, while for a non generic composite breather, we can either assume $\det({\pmb \delta}_1\,, {\pmb \tau}_1) = 0$ and $\det({\pmb \gamma}_1\,, {\pmb \omega}_1) \neq 0$, or $\det({\pmb \delta}_1\,, {\pmb \tau}_1) \neq 0$ and $\det({\pmb \gamma}_1\,, {\pmb \omega}_1) = 0$.
Direct calculations give the following estimates for each case.
\begin{itemize}
	\item[1.] {\bf Case $\det({\pmb \delta}_1\,, {\pmb \tau}_1) \neq 0$ and $\det({\pmb \gamma}_1\,, {\pmb \omega}_1) \neq 0$} (generic composite breather):
	\begin{gather}
		\Pi_1 \sim \begin{pmatrix}
			I_2 & 0_2\\
			0_2 & 0_2
		\end{pmatrix}, \quad \xi \to + \infty, \quad \quad \Pi_1 \sim \begin{pmatrix}
			0_2 & 0_2\\
			0_2 & I_2
		\end{pmatrix}, \quad \xi \to - \infty,
	\end{gather}
	and therefore
	\begin{subequations}
		\begin{gather}
			T_{k_1,\Pi_1}(k) \sim \begin{pmatrix}
				\frac{k-k_1}{k-k_1^*} \frac{k+k_1^*}{k+k_1}I_2 & 0_2\\
				0_2 & I_2
			\end{pmatrix},\quad \xi \to + \infty,\\
			T_{k_1,\Pi_1}(k) \sim \begin{pmatrix}
				I_2 & 0_2\\
				0_2 & \frac{k-k_1}{k-k_1^*} \frac{k+k_1^*}{k+k_1}I_2
			\end{pmatrix},\quad \xi \to - \infty.
		\end{gather}
	\end{subequations}
	\item[2.] {\bf Case $\det({\pmb \delta}_1\,, {\pmb \tau}_1) = 0$ and $\det({\pmb \gamma}_1\,, {\pmb \omega}_1) \neq 0$} (non generic composite breather):
	\begin{subequations}
		\begin{gather}
			\Pi_1 \sim \begin{pmatrix}
				\pi_{{\pmb \delta}_1} & 0_2\\
				0_2 & \pi_{{\pmb \omega}_{\alpha}}
			\end{pmatrix}, \quad \xi \to + \infty,
			\quad \quad \Pi_1 \sim \begin{pmatrix}
				0_2 & 0_2\\
				0_2 & I_2
			\end{pmatrix}, \quad \xi \to - \infty,\\
			\pi_{{\pmb \delta}_1} = \frac{{\pmb \delta}_1 {\pmb \delta}_1^\dagger}{{\pmb \delta}_1^\dagger {\pmb \delta}_1}, \quad \pi_{{\pmb \omega}_{\alpha}} = \frac{{\pmb \omega}_{\alpha} {\pmb \omega}_{\alpha}^\dagger}{{\pmb \omega}_{\alpha}^\dagger {\pmb \omega}_{\alpha}}, \quad {\pmb \omega}_{\alpha} = {\pmb \omega}_1 - \alpha {\pmb \gamma}_1,
		\end{gather}
	\end{subequations}
	and therefore
	\begin{subequations}
		\begin{gather}
			T_{k_1,\Pi_1}(k) \sim\begin{pmatrix}
			T_{k_1,\Pi_1}^{\mathrm{up}}(k) & 0_2\\
				0_2 & T_{k_1,\Pi_1}^{\mathrm{dn}}(k)
			\end{pmatrix}, \quad \xi \to + \infty\\
			T_{k_1,\Pi_1}(k) \sim \begin{pmatrix}
				I_2 & 0_2\\
				0_2 & \frac{k-k_1}{k-k_1^*} \frac{k+k_1^*}{k+k_1}I_2
			\end{pmatrix}, \quad \xi \to - \infty,
		\end{gather}
	\end{subequations}
where we introduced the notations $^{\mathrm{up}}$ and $^{\mathrm{dn}}$ for the upper/lower diagonal blocks of the Darboux matrix, and
\begin{subequations}
\begin{gather}
T_{k_1,\Pi_1}^{\mathrm{up}}(k) = \left( I_2 - \frac{k}{k_1}\frac{k_1^*-k_1}{k_1^*-k}\pi_{{\pmb \delta}_1} \right) \left( I_2 - \frac{k}{k_1^*}\frac{k_1^*-k_1}{k_1+k} \sigma_2 \pi_{{\pmb \delta}_1}^* \sigma_2 \right),\\
T_{k_1,\Pi_1}^{\mathrm{dn}}(k) = 				\left( I_2 - \frac{k}{k_1}\frac{k_1^*-k_1}{k_1^*-k}\pi_{{\pmb \omega}_{\alpha}} \right) \left( I_2 - \frac{k}{k_1^*}\frac{k_1^*-k_1}{k_1+k} \sigma_2 \pi_{{\pmb \omega}_{\alpha}}^* \sigma_2 \right).
\end{gather}
\end{subequations}	
	\item[3.] {\bf Case $\det({\pmb \delta}_1\,, {\pmb \tau}_1) \neq 0$ and $\det({\pmb \gamma}_1\,, {\pmb \omega}_1) = 0$} (non-generic composite breather):
	\begin{subequations}
		\begin{gather}
			\Pi_1 \sim \begin{pmatrix}
				I_2 & 0_2\\
				0_2 & 0_2
			\end{pmatrix}, \quad \xi \to + \infty, \quad \quad \Pi_1 \sim \begin{pmatrix}
				\pi_{{\pmb \tau}_\alpha} & 0_2\\
				0_2 & \pi_{{\pmb \gamma}_1}
			\end{pmatrix}, \quad \xi \to - \infty,\\
			\pi_{{\pmb \tau}_\alpha} = \frac{{\pmb \tau}_\alpha {\pmb \tau}_\alpha^\dagger}{{\pmb \tau}_\alpha^\dagger {\pmb \tau}_\alpha}, \quad {\pmb \tau}_\alpha={\pmb \tau}_1-\alpha {\pmb \delta}_1   , \quad\pi_{{\pmb \gamma}_1} = \frac{{\pmb \gamma}_1 {\pmb \gamma}_1^\dagger}{{\pmb \gamma}_1^\dagger {\pmb \gamma}_1},
		\end{gather}
	\end{subequations}
	and therefore
	\begin{subequations}
		\begin{gather}
			T_{k_1,\Pi_1}(k) \sim \begin{pmatrix}
				\frac{k-k_1}{k-k_1^*}\frac{k+k_1^*}{k+k_1}I_2 & 0_2\\
				0_2 & I_2
			\end{pmatrix}, \quad \xi \to + \infty,\\			
			T_{k_1,\Pi_1}(k) \sim 			\begin{pmatrix}
				T_{k_1,\Pi_1}^{\mathrm{up}}(k) & 0_2\\
				0_2 & T_{k_1,\Pi_1}^{\mathrm{dn}}(k)
			\end{pmatrix}, \quad \xi \to - \infty,
		\end{gather}
	\end{subequations}
where here
\begin{subequations}
\begin{gather}
T_{k_1,\Pi_1}^{\mathrm{up}}(k) = \left( I_2 - \frac{k}{k_1}\frac{k_1^*-k_1}{k_1^*-k}\pi_{{\pmb \tau}_\alpha} \right) \left( I_2 - \frac{k}{k_1^*}\frac{k_1^*-k_1}{k_1+k} \sigma_2 \pi_{{\pmb \tau}_\alpha}^* \sigma_2 \right),\\
T_{k_1,\Pi_1}^{\mathrm{dn}}(k) =
\left( I_2 - \frac{k}{k_1}\frac{k_1^*-k_1}{k_1^*-k}\pi_{{\pmb \gamma}_1} \right) \left( I_2 - \frac{k}{k_1^*}\frac{k_1^*-k_1}{k_1+k} \sigma_2 \pi_{{\pmb \gamma}_1}^* \sigma_2 \right).
\end{gather}
\end{subequations}
\end{itemize}
Now, using the above estimates and Eq.~\eref{newe:2}, we can compute the scattering coefficients $a(k)$ and $\bar{a}(k)$ (note that consistently with pure soliton solutions, we find $b(k)=0=\bar{b}(k)$). To simplify the expressions, we use identities such as
\begin{eqnarray}
	\left( I_2 - \frac{k}{k_1}\frac{k_1^*-k_1}{k_1^*-k}\pi \right)^{-1}&=&\left( I_2 - \frac{k}{k_1^*}\frac{k_1^*-k_1}{k-k_1}\pi \right)\\&=&
	\frac{k_1}{k_1^*}\frac{k-k_1^*}{k-k_1}\left( I_2 - \frac{k}{k_1}\frac{k_1^*-k_1}{k_1^*-k}\sigma_2\pi^*\sigma_2 \right),
\end{eqnarray}
valid for any rank-$1$ two-dimensional projector $\pi$ (recall that $\sigma_2\pi^*\sigma_2=I_2 -\pi$).
\begin{itemize}
	\item[1.] {\bf Case $\det({\pmb \delta}_1\,, {\pmb \tau}_1) \neq 0$ and $\det({\pmb \gamma}_1\,, {\pmb \omega}_1) \neq 0$} (generic composite breather):
	\begin{gather}
		a(k) = \frac{k-k_1}{k-k_1^*} \frac{k+k_1^*}{k+k_1}I_2, \quad \bar{a}(k) = \frac{k-k_1^*}{k-k_1} \frac{k+k_1}{k+k_1^*}I_2,
	\end{gather}
which implies that the points $k=k_1\,, -k_1^*$, are double zeros for the determinant of $a(k)$ on the upper half plane, but still simple poles for the RHP, because $a(k)$ becomes the identical zero matrix when evaluated at these points (resp., similar statement holds for the points $k=-k_1\,, k_1^*$ and the transmission coefficient $\bar{a}(k)$).
	\item[2.] {\bf Case $\det({\pmb \delta}_1\,, {\pmb \tau}_1) = 0$ and $\det({\pmb \gamma}_1\,, {\pmb \omega}_1) \neq 0$} (non-generic composite breather):
	\begin{subequations}
		\begin{gather}
			a(k) = \left( I_2 - \frac{k}{k_1}\frac{k_1^*-k_1}{k_1^*-k}\pi_{{\pmb \delta}_1} \right) \left( I_2 - \frac{k}{k_1^*}\frac{k_1^*-k_1}{k_1+k} \sigma_2 \pi_{{\pmb \delta}_1}^* \sigma_2 \right),\\
\bar{a}(k) = \left( I_2 + \frac{k}{k_1}\frac{k_1^*-k_1}{k+k_1^*}\pi_{{\pmb \omega}_{\alpha}} \right) \left( I_2 - \frac{k}{k_1^*}\frac{k_1^*-k_1}{k-k_1} \sigma_2 \pi_{{\pmb \omega}_{\alpha}}^* \sigma_2 \right),					
		\end{gather}
	\end{subequations}
	and in this case one can show that
	\begin{gather}
		\det a(k) = \frac{k-k_1}{k-k_1^*}\frac{k+k_1^*}{k+k_1}, \quad \det \bar{a}(k) = \frac{k+k_1}{k+k_1^*}\frac{k-k_1^*}{k-k_1},
	\end{gather}
which implies that the determinant of $a(k)$ has simple zeros at the points $k_1\,, -k_1^*$ on the upper half plane, and that $a(k)$ is nonzero when evaluated at these points (resp., the determinant of $\bar{a}(k)$ has simple zeros at the points $-k_1\,, k_1^*$ on the lower half plane, and $\bar{a}(k)$ is nonzero when evaluated at these points).
	\item[3.] {\bf Case $\det({\pmb \delta}_1\,, {\pmb \tau}_1) \neq 0$ and $\det({\pmb \gamma}_1\,, {\pmb \omega}_1) = 0$} (non-generic composite breather):
	\begin{subequations}
		\begin{gather}
a(k) = \left( I_2 - \frac{k}{k_1^*}\frac{k_1^*-k_1}{k_1+k} \pi_{{\pmb \tau}_\alpha} \right)\left(I_2 - \frac{k}{k_1}\frac{k_1^*-k_1}{k_1^*-k} \sigma_2 \pi_{{\pmb \tau}_\alpha}^* \sigma_2\right),\\		
			\bar{a}(k) =\left( I_2 - \frac{k}{k_1^*}\frac{k_1^*-k_1}{k-k_1} \pi_{{\pmb \gamma}_1} \right)\left( I_2 + \frac{k}{k_1}\frac{k_1^*-k_1}{k+k_1^*}\sigma_2 \pi_{{\pmb \gamma}_1}^* \sigma_2 \right),
		\end{gather}
	\end{subequations}
	which again gives us
	\begin{gather}
		\det a(k) = \frac{k-k_1}{k-k_1^*}\frac{k+k_1^*}{k+k_1}, \quad \det \bar{a}(k) = \frac{k+k_1}{k+k_1^*}\frac{k-k_1^*}{k-k_1},
	\end{gather}
	and we reach the same conclusion as before for the zeros of $\det a(k)$ and $\det \bar{a}(k)$.
\end{itemize}
The above computations show that in both cases the transmission coefficients of a non-generic composite breather lead to rank-1 norming constants (see Sec.~2), and hence when either $\det({\pmb \delta}_1\,, {\pmb \tau}_1) = 0$ or $\det({\pmb \gamma}_1\,, {\pmb \omega}_1) = 0$, a non generic composite breather reduces to a fundamental breather. The above calculations show this spectrally and one can also see this at the level of the solution. For instance, when ${\pmb \tau_1}= \alpha {\pmb \delta}_1$, \eref{e:171compbr} reduces to $(141)$ with the role of ${\pmb \gamma}_1$ for the fundamental breather now played by the combination ${\pmb \omega}_1-\alpha{\pmb \gamma}_1$ where ${\pmb \omega}_1, {\pmb \gamma}_1$ are the vectors characterizing the (non generic) composite breather.

\end{document}